\documentclass[journal]{IEEEtran}

\usepackage[T1]{fontenc}
\usepackage{graphicx}
\usepackage{cite}
\usepackage{caption}
\usepackage{subcaption}
\usepackage{epstopdf}
\usepackage{lipsum}

\usepackage[pdftex]{pict2e}
\usepackage{overpic}
\usepackage{amssymb}
\usepackage{amsmath}
\usepackage{amsthm}
\usepackage{array}
\usepackage{color}
\usepackage{url}
\usepackage{hyperref}
\usepackage{bbm}

\theoremstyle{remark}

\newtheorem{lemma}{Lemma}

\usepackage[dvipsnames,svgnames]{xcolor}
\usepackage[textwidth=30mm]{todonotes}

\newcommand{\vect}[1]{\mathbf{#1}}

\def\Real{\mathbb{R}}
\def\Complex{\mathbb{C}}
\def\Integer{\mathbb{Z}}

\def\diag{\mathrm{diag}}

\def\jinc{\mathrm{jinc}}

\def\dof{\mathrm{dof}}

\def\sinc{\mathrm{sinc}}

\def\Ttran{\mbox{\tiny $\mathrm{T}$}}
\def\CN{\mathcal{N}_{\mathbb{C}}} %Complex Gaussian
\def\imagunit{\mathsf{j}} % Imaginary number
\def\Ex{\mathbb{E}}

\IEEEoverridecommandlockouts

\graphicspath{{./images/}}

\begin{document}

\title{Nyquist Sampling and Degrees of Freedom\\ of Electromagnetic Fields\vspace{-0.0cm}} % for Communications

\author{
\IEEEauthorblockN{Andrea Pizzo, \emph{Member, IEEE}, Andrea de Jesus Torres, \emph{Student Member, IEEE}, \\Luca Sanguinetti, \emph{Senior Member, IEEE}, Thomas L. Marzetta, \emph{Life Fellow, IEEE}\vspace{-.5cm}
\thanks{Part of this work was presented at the IEEE 21st International Workshop on Signal Processing Advances in Wireless Communications (SPAWC), Georgia, US, 2020 \cite{PizzoSPAWC20}.}
%\newline
\thanks{A.~Pizzo was with University of Pisa. He is now with Universitat Pompeu Fabra, 08018 Barcelona, Spain (email: andrea.pizzo@upf.edu).}
\thanks{A.~d.~J.~Torres, and L.~Sanguinetti are with University of Pisa, 56122 Pisa, Italy (e-mail: andrea.dejesustorres@phd.unipi.it; luca.sanguinetti@unipi.it).}
\thanks{T.~Marzetta is with New York University, 11201 Brooklyn, New York (e-mail: tom.marzetta@nyu.edu).}
}}

% make the title area
\maketitle

\begin{abstract}
A signal space approach is presented to study the Nyquist sampling, number of degrees of freedom and reconstruction of an electromagnetic field under arbitrary scattering %{\color{blue}scattering} 
conditions. 
Conventional signal processing tools, such as the multidimensional sampling theorem and Fourier theory, are used to provide a linear system theoretic interpretation of electromagnetic wave propagation, thereby revealing the spatially bandlimited nature of electromagnetic fields. Their spatial bandwidth is dictated by the selectivity of the underlying scattering that allows establishing the Nyquist spatial sampling with a reduction of the number of fields’ samples needed to be processed.
%Scalar electromagnetic fields are considered for simplicity, which physically correspond to acoustic propagation in general or electromagnetic propagation under certain conditions. The developed approach is extended to study ensembles of a stationary random electromagnetic field that is representative of propagation in different environments.

\vspace{-0.1cm}
\end{abstract}

\smallskip

\IEEEpeerreviewmaketitle

\vspace{-.5cm}

\section{Introduction}
%\AP{long story short: the DoF tells us how many samples per unit of space we should use in the cardinal series. The total number of samples is dictated by the observation area.}

The Kotelnikov-Shannon-Whittaker sampling theorem \cite{Shannon_Noise,whittaker_1928,Kotelnikov2001} states that any squared-integrable signal of finite bandwidth $\Omega/\pi$~Hz, ${e(t) \in\mathcal{B}_\Omega}$ for $t\in \Real$, can be perfectly reconstructed from its {samples} taken at equally spaced ${Q^\star = \pi/\Omega}$~seconds apart by the {reconstruction} formula (also called cardinal series)
\begin{equation} \label{cardinal_series}
e(t) =   \sum_{n \in \Integer} e\left(n Q^\star\right) f_{\Omega}\left(t - n Q^\star \right)
\end{equation}
where $f_{\Omega}(t) = \sinc\left(t/Q^\star\right)$.
%Since $\{e(n Q^\star)\}$ interpolate $e(t)$ for each $n\in\Integer$, we can rewrite~\eqref{cardinal_series} as an interpolating formula:
%\begin{equation} \label{cardinal_series_infinite}
%e(t) =   \sum_{n \in \Integer} e\left(n Q^\star\right) f_{\Omega}\left(t - n Q^\star \right)
%\end{equation}
%while defining an interpolating function
%\begin{align}  \label{interpolation_function_time}
%f_{\Omega}(t) &= \frac{Q^\star}{2\pi} \int_{-\infty}^\infty \mathbbm{1}_{|\omega|\le\Omega}(\omega) e^{\imagunit \omega t}\, d\omega  =  \sinc\left(t/Q^\star\right).
%\end{align}
Orthonormality is a fundamental property of the reconstruction sinc-functions, which allows to geometrically interpret the samples as coordinates of a basis set of functions~\cite{Shannon_Noise}. 
In general, $e(t)$ may not be bandlimited, and an ideal low-pass filtering operation is required prior to sampling with reduced reconstruction accuracy \cite{Unser}. 
%Within a finite time interval of duration $T$, only a finite number of coordinates are needed to specify the signal approximately, the minimum of which are the  \emph{degrees of freedom} (DoF) of $e(t)$.

The generalization of the sampling theorem in \eqref{cardinal_series} to %{\color{blue}spatially bandlimited fields} 
spatially bandlimited fields is known as \emph{multidimensional sampling theorem}~\cite{SamplingBook,MultivariateSPBook,PETERSEN1962279,Agrell2004}.
With respect to the classical definition of a bandlimited signal in the frequency domain, this notion applies in the spatial-frequency (wavenumber) domain \cite[Ch.~8]{FranceschettiBook}.
Its application requires an accurate description of the field in the spectral domain, which is tied to its physical nature.
%Unlike the classical sampling theorem, which relies on the signal's bandwidth $\Omega/\pi$ only, in multiple dimensions , as the shape of the spectral support of $e(\vect{r})$ may vary significantly according to the application. 
For wireless researchers, the class of three-dimensional (3D) electromagnetic fields is of great interest as they allow transfer of information from one point to another \cite{Migliore_Horse}.
A field of this sort generates a continuum of spatial samples in the form of an analog physical process. 
Yet the noise allows for a certain error in the field description, calling for a discrete representation through spatial sampling. 
%Our focus is on scalar electromagnetic fields.
%However, practical systems are inevitably constrained by a finite precision, calling for an information-lossless finite description of the field~\cite{PizzoTWC21}.

%\AP{Se la vogliamo fare piu signal processing oriented, una frasetta del tipo: multiple sensors may be potentially located to sample the field for signal processing or communication purposes.
%\LS{Da qua sopra non ci capisce quale sia la motivazione}
%\subsection{Contributions}

%The emergence of metamaterials enables a full manipulation of the propagation environment surrounding the user, which may represent a major paradigm change in wireless communications systems in terms of data rate, coverage, and enhanced security/privacy \cite{Hu2018,Rui2019,BJORNSON20193,Wu2020}.
%Once deployed, metamaterials can turn any man-made structure like building and wall into a smart surface \cite{Wu2020}.
%However, full exploitation of the potentialities offered by smart surfaces is only possible via an understanding of the electromagnetic wave propagation phenomena \cite{Migliore_Horse,PizzoTWC21}.

In processing analog fields, it is highly desirable to minimize the sampling rate because the amount of data processing, power consumption, and digital storage is proportional to the number of samples to be processed~\cite{OppenheimBook}. This is of paramount importance at high frequencies (e.g., millimeter wave and sub-terahertz spectrum) where the complex nature of propagation channels requires to collect a large amount of measurements~\cite{Rappaport2019}.
%channel for two main reasons. First, the level of propagation details that needs to be captured increases compared to sub-6~GHz \cite{Rappaport2019}. Second, as the wavelength shortens, asymptotic results such as sampling theorem and degrees of freedom (DoF) become more accurate \cite{Shannon_Noise,whittaker_1928,Kotelnikov2001}.}}

%where digital signal processing and hardware power consumption are both critical concerns [REFS].
Motivated by the above considerations, this paper revisits the sampling theory for electromagnetic fields with potential applications to, e.g., channel sounding, sensing and localization at high frequencies.
%This paper considers scalar electromagnetic fields~\cite{ChewBook}, and exploits the fact that electromagnetic fields $e(\vect{r})$ are essentially bandlimited in the spectral domain~\cite{Bucci87,Bucci98}, and thus, the multidimensional sampling theorem can directly be applied without recurring to a filtering operation. 
%with a clear implication on sampling, reconstruction, and DoF of an electromagnetic field. 
We deviate from the electromagnetic literature~\cite{Bucci87,Bucci98}, and provide a signal processing interpretation of wave propagation under arbitrary scattering conditions. Precisely, 3D wireless propagation can be modeled as a two-dimensional (2D) linear and space-invariant (LSI) system with low-pass filtering behavior.
%Unlike~\cite{Bucci87,Bucci98}, a signal processing interpretation of wave propagation, under arbitrary scattering conditions, is provided, based on which 3D wireless propagation can be modeled as a two-dimensional (2D) linear and space-invariant (LSI) system with low-pass filtering behavior. %Unlike time-domain signals, no projection from $\mathcal{L}_2$ to $\mathcal{B}_\Omega$ is needed whatsoever. 
%We provide a spectral description of the class of 3D electromagnetic fields $e(\vect{r})$ and apply the general theory to solve the Nyquist sampling and reconstruction problem, and compute the number of DoF.
Physics-driven approaches similar to the one used in this paper can be found in~\cite{PoonDoF,Kennedy2007,Hanlen2007}, where the number of degrees of freedom (DoF) that can be extracted from electromagnetic fields, confined in a space region with spherical symmetry, is computed.

Unlike~\cite{PoonDoF,Kennedy2007,Hanlen2007}, we consider regions of rectangular symmetry and perform the analysis in Cartesian coordinates, by leveraging the connection with Fourier theory~\cite{Sayeed2002,PizzoIT21,PizzoJSAC20,MarzettaISIT,MarzettaNokia}. 
In addition, we connect the DoF result with the Nyquist sampling and field reconstruction problem. Notice that asymptotic results such as DoF and sampling theorem become more accurate as the wavelength shortens (i.e., high frequencies). Both were studied in~\cite{Hu2018} under isotropic scattering.

We generalize~\cite{Hu2018} to an arbitrary configuration of scatterers and provide an alternative explanation of the results via a signal processing approach. The development is finally extended to ensembles of stationary electromagnetic Gaussian random fields~\cite{PizzoIT21,PizzoJSAC20,MarzettaISIT,MarzettaNokia}.
%Here, we not only provide a better understanding of them via a signal processing approach but also generalize the results to arbitrary propagation conditions. Specifically, to abstract from a particular configuration of scatterers, we extend our analysis to ensembles of stationary  electromagnetic Gaussian random fields 
%and must be interpreted on average. 
%The reconstruction via sampling theorem is still possible and provides us with the best linear interpolator in the mean squared error (MSE) sense for any finite number of samples~\cite{Agrell2004,PETERSEN1962279}.

%\LS{Non e' male ma non punge! Leggerei Lund (vedi anche commento nelle conclusioni) e qualche paper che lo cita per avere qualche idea sulla applicazioni/motivazione.}
%\AP{ho aggiunto la parte in magenta. Ci sono solo tre paper di signal processing che citano Lund. Secondo me abbiamo abbastanza refs di SP.}

\subsection{Outline of the Paper}

The remainder of this paper is organized as follows. In Section~\ref{sec:Preliminaries}, we review the sampling theorem and put much emphasis on the interplay among Nyquist sampling, reconstruction, and number of DoF. In Section~\ref{sec:electromagnetic}, we provide a spectral characterization of the class of 3D electromagnetic fields. This is first used in Section~\ref{sec:Nyquist_sampling} to solve the Nyquist sampling and reconstruction problem and later in Section~\ref{sec:DoF} to compute the DoF. 
The generalization of the developed framework to ensembles of stationary electromagnetic random fields is provided in Section~\ref{sec:stationary}.
Numerical results are given in Section~\ref{sec:numerical} to validate the developed theory under practical settings. Conclusions are drawn in Section~\ref{sec:conclusions}.

\subsection{Notation}
We use upper (lower) case letters for spatial-frequency (spatial) entities. Boldfaced letters indicate vectors and matrices. The superscripts $^{\Ttran}$, $^{-1}$, and $^{1/2}$ stand for transposition, inverse, and matrix square root, respectively. $\diag(\vect{x})$ indicates the diagonal matrix with elements from $\vect{x}$. $\det(\vect{X})$ is the determinant of $\vect{X}$. $\vect{I}_n$ is the $n$-dimensional identity matrix. Calligraphic letters indicate sets. $m(\mathcal{X})$ denotes the Lebesgue measure, $\mathbbm{1}_{\mathcal{X}}(x)$ is the indicator function. $\Real^n$ and $\Integer^n$ are the $n$-dimensional spaces of real-valued and integer numbers, $|\cdot|$ denotes absolute value, $\lceil x \rceil$ denotes the least integer greater than or equal to $x$, $\sinc(x)=\sin(\pi x)/(\pi x)$ is the sinc function, $\jinc(x) = J_1(x)/x$ is the jinc function where $J_1(x)$ denotes the Bessel function of the first kind with order $1$. 
A general point $(\vect{r},z) \in\Real^3$ is described by its Cartesian coordinates $(x,y,z)$ with $\|\vect{r}\| = \sqrt{x^2 + y^2 + z^2}$ the Euclidean norm.
$\nabla^2 = \frac{\partial^2}{\partial x^2} + \frac{\partial^2}{\partial y^2} + \frac{\partial^2}{\partial z^2}$ is the Laplacian operator. $\Ex\{\cdot\}$ denotes the expectation operator. The notation $n \sim \CN(0, \sigma^2)$ stands for a circularly symmetric complex Gaussian random variable with variance $\sigma^2$. 

\section{Preliminaries} \label{sec:Preliminaries}

To make the paper self-contained, we briefly review the sampling theorem in one and two dimensions. This is preparatory for studying spatial electromagnetic fields.

%\subsection{{\color{blue}Time-domain Signals}}
\subsection{Time-domain Signals}

Let $e(t)$ be a baseband signal with its spectrum ${E}(\omega)$.
%\footnote{A centered spectrum can always be obtained by opportunely choosing the frequency of the bandpass downconversion.}
Uniform sampling of $e(t)$ at equally spaced $Q$ seconds apart yields the sampled signal ${e}_{\rm s}(t) =   \sum_{n \in \Integer} {e}(n Q) \delta(t- n Q)$
%The sampled signal ${e}_{\rm s}(t)$ obtained from the samples $\{{e}(n Q)\}$ of $e(t)$, taken at equally spaced $Q$ seconds apart, is 
%\begin{equation} \label{smalLcale_t1_ISO_sampled}
%{e}_{\rm s}(t) =   \sum_{n \in \Integer} {e}(n Q) \delta(t- n Q)
%\end{equation}
with spectrum~\cite[Sec.~3]{SamplingBook}
\begin{align} \label{smalLcale_t1_ISO_sampled_freq_v1}
E_{\rm s}(\omega) &=  \sum_{n \in \Integer} {e}(n Q) e^{-\imagunit \omega n Q}\\&= \frac{1}{Q} \sum_{\ell \in \Integer} {E}(\omega - \ell P)  \label{smalLcale_t1_ISO_sampled_freq}
\end{align}
where the frequency replication period $P$ is such that
\begin{equation} \label{repetition_sampling}
P Q = 2\pi.
\end{equation}
%Suppose $e(t) \in \mathcal{B}_\Omega$, i.e., it has finite bandwidth $\Omega/\pi$ (in Hz). 
%From~\eqref{smalLcale_t1_ISO_sampled_freq}, the frequency replicas of $E(\omega)$ do not overlap, (i.e., no \emph{aliasing} occurs) when $P \ge 2\Omega$ or, equivalently, $Q \le \pi/\Omega$.
For any bandlimited signal $e(t) \in \mathcal{B}_\Omega$ of bandwidth $\Omega/\pi$ (in Hz), the Nyquist sampling interval
\begin{equation} \label{optimal_sampling_time}
Q^\star = \pi/\Omega
\end{equation}
provides the minimum number of samples per unit of time (i.e., the Nyquist rate ${\mu = 1/Q^\star}$ (in samples/s)) that is needed for perfect reconstruction of $e(t)$ from its samples $\{{e}(n Q)\}$.
Precisely, under Nyquist sampling, $E(\omega)$ can be perfectly retrieved by low-pass filtering the fundamental replica of $E_{\rm s}(\omega)$ in~\eqref{smalLcale_t1_ISO_sampled_freq} with no \emph{aliasing},
\begin{equation} \label{filtering}
E(\omega) =  E_{\rm s}(\omega) \left(Q^\star \mathbbm{1}_{|\omega|\le\Omega}(\omega) \right).
\end{equation}
The reconstruction formula in~\eqref{cardinal_series} for $e(t)$ at any $t\in(-\infty,\infty)$ is obtained by an inverse Fourier transform of~\eqref{filtering} with $E_{\rm s}(\omega)$ given by~\eqref{smalLcale_t1_ISO_sampled_freq_v1} \cite[Sec.~3]{SamplingBook}, \cite{Shannon_Noise}. 
The interpolating sinc-function does not depend on the spectral characteristics of $E(\omega)$, but only on the measure of its support, i.e., the bandwidth $\Omega$.

%{\color{blue}Given the continuous nature of $e(t)$, signals belonging to $\mathcal{B}_\Omega$ span an infinite-dimensional space. However, noise and quantization effects allows a certain error in the representation so that bandlimited signals are amenable to a representation over a set of approximately DoF dimension} \cite[Sec.~2]{FranceschettiBook}.
Given the continuous nature of $e(t)$, signals belonging to $\mathcal{B}_\Omega$ span an infinite-dimensional space. However, noise and quantization effects allows a certain error in the representation so that bandlimited signals are amenable to a representation over a set of approximately DoF dimension \cite[Sec.~2]{FranceschettiBook}.
%In principle, $\mathcal{B}_\Omega$ is an infinite-dimensional function space as a countably-infinite number of coefficients is needed to specify any signal exactly. In practice, however, $\mathcal{B}_\Omega$ has an effective dimension that allows us to truncate~\eqref{cardinal_series} up to a countable number of samples for any given accuracy. The minimum number of such samples defines the DoF of $e(t)$ \cite[Sec.~2]{FranceschettiBook}.
Within a time interval of duration $T$, since there are $1/Q^\star$ samples/s, we have a total of~\cite[Eq.~(2.14)]{FranceschettiBook} 
\begin{equation} \label{Shannon_DoF}
\dof_{[-\Omega,\Omega]} =  \frac{\Omega T}{\pi}
\end{equation}
significant samples in the interval. 
%the DoF of $e(t)$, obtained by taking samples at equally spaced $Q^\star$ seconds apart, are given by~\cite[Eq.~(2.14)]{FranceschettiBook}
The implication of \eqref{Shannon_DoF} is that only a finite number of samples carries the essential information contained in the signal and can be used to reconstruct it. The reconstruction error decreases as $\Omega T$ increases and vanishes at infinity \cite[Sec.~2]{FranceschettiBook}.
Sampling above the Nyquist rate does not create any additional DoF and adds an insignificant contribution to the signal's reconstruction.

% Due to orthonormality of the interpolating sinc-functions in \eqref{cardinal_series}, the DoF in (7) specifies the minimum number of coefficients needed to represent any signal up to a certain level of accuracy. Alternatively said, sampling above the Nyquist rate does not create any additional DoF, but only add redundancy to the signal’s representation.

%\subsection{{\color{blue}2D Spatial Fields}} \label{sec:sampling_2D}
\subsection{2D Spatial Fields} \label{sec:sampling_2D}

%\begin{align}
%\begin{pmatrix}
%\vect{k}_{\rm c}/\kappa \\
%k_z(\vect{k}_{\rm c})/\kappa
%\end{pmatrix}
%=
%\begin{pmatrix}
% \sin(\theta_{\rm c}) \cos(\phi_{\rm c}) \\
% \sin(\theta_{\rm c}) \sin(\phi_{\rm c}) \\
% \cos(\theta_{\rm c})
%\end{pmatrix}
%\end{align}
The generalization of a 1D uniform sampling to a 2D lattice $\vect{Q} \vect{n}$ given $\vect{Q} \in \Real^{2\times 2}$ the (non-singular) sampling matrix yields ${e}_{\rm s}(\vect{r}) =  \sum_{\vect{n} \in \Integer^2} {e}(\vect{Q} \vect{n}) \delta(\vect{r} - \vect{Q} \vect{n})$
%\begin{equation} \label{smalLcale_t2_ISO_sampled_vec}
%{e}_{\rm s}(\vect{r}) =  \sum_{\vect{n} \in \Integer^2} {e}(\vect{Q} \vect{n}) \delta(\vect{r} - \vect{Q} \vect{n})
%\end{equation}
 \cite[Sec.~6]{SamplingBook}.
As an example, a simple 2D rectangular sampling is obtained when $\vect{Q} = \diag(q_x,q_y)$ where $q_x$ and $q_y$ are the sampling intervals along the $x-$ and $y$-axis, respectively.
The spatial Fourier transform of ${e}_{\rm s}(\vect{r})$ is \cite[Sec.~6]{SamplingBook}
\begin{align} \label{smalLcale_t2_ISO_sampled_freq}   
E_{\rm s}(\vect{k}) &=  \sum_{\vect{n} \in \Integer^2} {e}(\vect{Q} \vect{n}) e^{-\imagunit \vect{k}^{\Ttran} \left(\vect{Q} \vect{n}\right)}\\&=  
\frac{1}{|\det(\vect{Q})|} \sum_{\boldsymbol{\ell} \in \Integer^2} {E}\left(\vect{k} - \vect{P} \boldsymbol{\ell}\right)  
\end{align}
where the periodicity matrix $\vect{P}\in \Real^{2\times 2}$ is related to $\vect{Q}$ as
\begin{equation} \label{sampling_repetition}
\vect{P}^{\Ttran} \vect{Q} = 2\pi \vect{I}_2.
\end{equation} 
For any bandlimited field $e(\vect{r}) \in \mathcal{B}_{\mathcal{K}}$ of \emph{spatial bandwidth} $m(\mathcal{K})$ (in (rad/s)$^2$), we denotes with $\vect{P}^\star$ the Nyquist periodicity matrix that provides the minimum interspacing among adjacent replicas of ${E}(\vect{k})$ with no overlapping.
The associated Nyquist sampling matrix
\begin{equation} \label{Nyquist_sampling_matrix}
\vect{Q}^\star = 2\pi (({\vect{P}^\star})^{\Ttran})^{-1}
\end{equation}
provides the minimum number of samples per unit of space, i.e., the Nyquist density (in samples/m$^2$) \cite{MultivariateSPBook,AgrellSPL}
\begin{equation} \label{Nyquist_rate}
\mu = \frac{1}{|\det(\vect{Q}^\star)|}
\end{equation} 
that is needed for perfect reconstruction of $e(\vect{r})$ from its samples $\{{e}(\vect{Q}^\star \vect{n})\}$.
%{\color{blue}This extends the notion of Nyquist sampling rate (in samples/m) to spatial fields.}
This extends the notion of Nyquist sampling rate (in samples/m) to spatial fields.
%Under this condition, sampling is an information-lossless operation and the optimal sampling matrix $\vect{Q}^\star$ is the one that minimizes the sampling rate (i.e., the Nyquist rate)
As in~\eqref{filtering}, under Nyquist sampling, $E(\vect{k})$ can be perfectly retrieved by low-pass filtering $E_{\rm s}(\vect{k})$,
\begin{equation} \label{filtering_spatial}
E(\vect{k}) =  E_{\rm s}(\vect{k}) \left( |\det(\vect{Q}^\star)| \mathbbm{1}_{\mathcal{K}}(\vect{k}) \right).
\end{equation}
%\footnote{The Nyquist sampling of a time-domain signal $e(t)$ generates a periodic replication of $E(\omega)$ with no interstitial spacing. In the spatial domain, however, adjacent replicas of $E(\vect{k})$ are separated spatially. Within these spacings, the spectrum of the reconstruction function $F_{\mathcal{K}}(\vect{k})$ can be chosen arbitrarily so that infinitely many reconstruction functions may be found for any Nyquist sampling.}
%\footnote{In general, there are infinitely many reconstruction functions $f_{\mathcal{K}}(\vect{r})$ whose spectrums $F_{\mathcal{K}}(\vect{k})$ must satisfy the Nyquist criteria. One of such functions is found according to~\eqref{filtering_spatial}, but others may be found }
The 2D counterpart of~\eqref{cardinal_series} is obtained via an inverse Fourier transform of~\eqref{filtering_spatial} with $E_{\rm s}(\vect{k})$ in~\eqref{smalLcale_t2_ISO_sampled_freq}~\cite[Eq.~(6.33)]{SamplingBook}:
\begin{equation} \label{cardinal_series_spatial}
e(\vect{r}) =   \sum_{\vect{n} \in \Integer^2} e(\vect{Q}^\star \vect{n}) f_{\mathcal{K}}(\vect{r} - \vect{Q}^\star \vect{n})
\end{equation}
with interpolating function \cite[Eq.~(6.34)]{SamplingBook}
\begin{align} \label{interp_function}
f_{\mathcal{K}}(\vect{r}) & = \frac{|\det(\vect{Q}^\star)|}{(2 \pi)^2} \int_{\Real^2} \mathbbm{1}_{\mathcal{K}}(\vect{k}) e^{\imagunit \vect{k}^{\Ttran} \vect{r}} \, d\vect{k}.
%\\
%& = |\det(\vect{Q})| \int_{\mathcal{K}} e^{\imagunit \vect{k}^{\Ttran} \vect{r}} \, d\vect{k}.
\end{align}
Similarly to the 1D case, Nyquist sampling and interpolating function do not depend on the values assumed by $E(\vect{k})$, but solely by the shape of its support $\mathcal{K}$. Unfortunately, these cannot be computed univocally from its measure (as for 1D signals) as multidimensional supports are described by more than one parameter \cite{PETERSEN1962279}.
%{\color{blue}Our focus is on 3D scalar electromagnetic fields generated by arbitrary scattering conditions. A scalar formulation physically corresponds to acoustic wave propagation in general \cite{Hanlen2007,Kennedy2004} or electromagnetic wave propagation in a source-free environment \cite{ChewBook}. Their spectral support is characterized next.}
Our focus is on 3D scalar electromagnetic fields generated by arbitrary scattering conditions. A scalar formulation physically corresponds to acoustic wave propagation in general \cite{Hanlen2007,Kennedy2004} or electromagnetic wave propagation in a source-free environment \cite{ChewBook}. Their spectral support is characterized next.
%Other physical phenomena such as acoustic waves propagation are described by a scalar formulation \cite{Hanlen2007, Kennedy2004}.
%As seen in Section~\ref{sec:electromagnetic}, the propagation behavior of $e(\vect{r})$ along $z$ is exactly known a-priori. 
%, as it does not depend on the scattering conditions, but solely by the $z$-plane at which we measure the field (see also Fig.~\ref{fig:non_isotropic_spectrum}).}

%The same reasoning can be applied to a field $\vect{r}$ with disconnected support, but lies within a bounded wavenumber region \cite{PETERSEN1962279}. 

%Compared to \eqref{smalLcale_t1_ISO_sampled} and \eqref{smalLcale_t1_ISO_sampled_freq}, in \eqref{smalLcale_t2_ISO_sampled_vec} and \eqref{smalLcale_t2_ISO_sampled_freq} the scalar quantities $Q$ and $P$ are replaced by the matrix quantities $\vect{Q}$ and $\vect{P}$, respectively. The lattice density $1/|\det(\vect{Q})|$ denotes the number of spatial samples per unit of area (in samples/m$^2$) and can be regarded as the multidimensional counterpart of the sampling rate $1/Q$ 
%\footnote{Notice that an infinite number of spectra having supports of different shapes may be found for a given repetition matrix \cite{PETERSEN1962279}.} 
%Notice that this does not require the support of ${E}(\vect{k})$ to be connected, but only that it lies within a bounded wavenumber region \cite{PETERSEN1962279}. 

\section{Spectral Characterization of Electromagnetic Fields} \label{sec:electromagnetic}

%Based on the sampling theorem, any signal ${e(t) \in \mathcal{B}_\Omega}$ can approximately be identified by a countable number of samples taken at Nyquist rate, which is univocally determined by the bandwidth of its spectrum.
The generalization of sampling theorem to multiple dimensions requires a spectral characterization of the considered field.
This is addressed next for \emph{electromagnetic fields}, generated by the interaction of a radiated field with an arbitrary configuration of scatterers in the half-space $\{z<0\}$. Upon interaction, a scattered field propagating in a 3D homogeneous and isotropic medium is measured at receiver, modeling a non line-of-sight propagation scenario.
Under these settings, the Maxwell's equations reduce to the vector (source-free) homogeneous Helmholtz equations \cite[Eq.~(1.2.20)]{ChewBook}
\begin{equation} \label{vector_Helmholtz_eq}
\left(\nabla^2 + \frac{\partial^2}{\partial z^2} \right) \vect{e}(\vect{r},z) + \kappa^2 \vect{e}(\vect{r},z) = 0
\end{equation}
where $\vect{e}(\vect{r},z)$ with ${\vect{r} = (x,y) \in\Real^2}$ can be either the electric or magnetic vector fields and $\kappa = 2\pi/\lambda$ is the wavenumber with $\lambda$ the wavelength. 
We do our analysis in Cartesian coordinates, i.e., $\vect{e}(\vect{r},z) = \hat{\vect{x}} e_x + \hat{\vect{y}} e_y + \hat{\vect{z}} e_z$. In this case,~\eqref{vector_Helmholtz_eq} can be addressed independently along each coordinate \cite{Wang88}.\footnote{Separability does not hold in spherical coordinates where~\eqref{vector_Helmholtz_eq} must be solved in vector form~\cite[Sec.~1.2]{ChewBook}.}
%Other physical phenomena can be described using the scalar homogeneous Helmholtz equation such as acoustic waves propagation \cite{Hanlen2007, Kennedy2004}.
%This requires an accurate description of the spectral characteristics of $e(t)$ through the signal's bandwidth $\Omega$.
%and provide a linear-system theoretic interpretation of the wave propagation phenomena.
%A signal-space approach, analogous to the one developed by Shannon~\cite{Shannon_Noise} for a signal $e(t) \in \mathcal{B}_\Omega$, can be extended to a spatial field of electromagnetic nature, i.e., an electromagnetic field $e(\vect{r})$, as shown next. 
In a line-of-sight scenario, the capability of an observer to resolve features of the radiated field is distance-dependent. 
This effect should be accounted for by considering an inhomogeneous equation driven by the distribution of the radiator \cite{Miller}.
%In this section, we characterize the spectrum of a 3D scalar field $e(x,y,z)$ of electromagnetic nature evaluated at any $z$-plane in the half-space $z\ge 0$ and provide insights into its structure.

\subsection{Scalar Homogeneous Helmholtz Equation}

\begin{figure} [t!]
        \centering
	\begin{overpic}[width=0.6\columnwidth,tics=10]{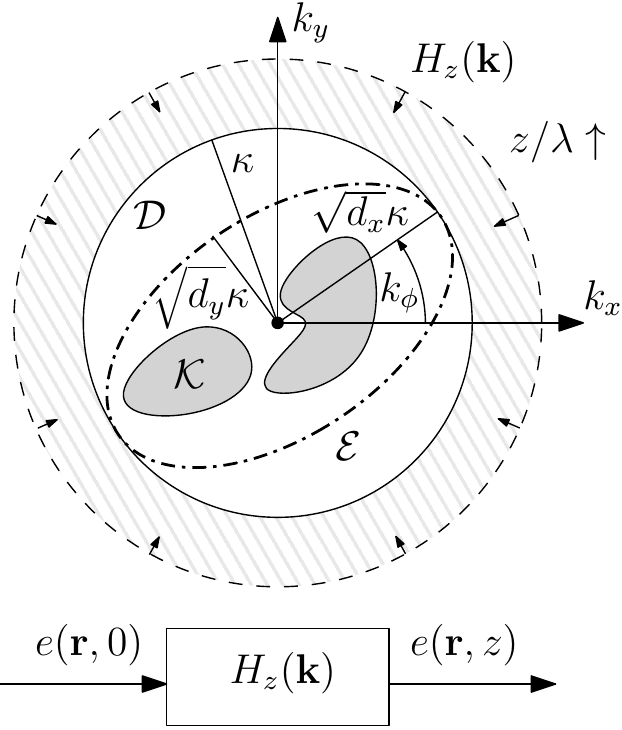}
\end{overpic} \vspace{-0.0cm}
                \caption{{\small Wave propagation in a 3D homogeneous, isotropic, and source-free medium can be modeled as a  low-pass filter.}} \vspace{-0.3cm}
                \label{fig:non_isotropic_spectrum}        
\end{figure}

Without loss of generality, we denote with $e(\vect{r},z)$ either $e_x$, $e_y$ or $e_z$. This field obeys the scalar homogeneous Helmholtz equation \cite[Eq.~(1.2.21)]{ChewBook}
\begin{equation} \label{Helmholtz}
\left(\nabla^2 + \frac{\partial^2}{\partial z^2} \right)e(\vect{r},z) + \kappa^2 e(\vect{r},z) = 0
\end{equation}
or in the horizontal spatial-frequency (wavenumber) domain,
\begin{equation} \label{Helmholtz_Fourier}
 \frac{\partial^2}{\partial z^2} E(\vect{k},z) + (\kappa^2 - \|\vect{k}\|^2) E(\vect{k},z)  = 0
\end{equation}
where ${\vect{k} = (k_x,k_y)\in\Real^2}$ are the horizontal wavenumber Cartesian coordinates and
\begin{equation} \label{Fourier_transform}
E(\vect{k},z) = \int_{\Real^2}  e(\vect{r},z) e^{-\imagunit \vect{k}^{\Ttran} \vect{r}} \, d\vect{r}.
\end{equation}
For any fixed $\vect{k} \in\Real^2$,~\eqref{Helmholtz_Fourier} is a second-order ordinary differential equation  in $z$ with constant coefficients whose general solution in the half-space $\{z\ge0\}$ is of the form 
\begin{equation} \label{Helmholtz_Fourier_solution}
E(\vect{k},z) = E(\vect{k}) e^{\imagunit k_z(\vect{k}) z},
\end{equation}
where $E(\vect{k})$ is an arbitrary complex-valued coefficient and $k_z(\vect{k})$ is defined as
\begin{equation} \label{k_z}
k_z(\vect{k}) =
\begin{cases}
\sqrt{\kappa^2 - \|\vect{k}\|^2} &  \vect{k} \in \mathcal{D} \\
\imagunit \sqrt{\|\vect{k}\|^2 - \kappa^2} & \text{elsewhere}
\end{cases}
\end{equation}
given 
\begin{equation} \label{disk}
\mathcal{D} = \{\vect{k}\in\Real^2 : \|\vect{k}\|^2\le\kappa^2\}
\end{equation}
as a disk of radius $\kappa$; see Fig.~\ref{fig:non_isotropic_spectrum}. Notice that $\vect{k}$ can vary independently in $\Real^2$ and hence $k_z(\vect{k})$ in~\eqref{k_z} can be either real- or imaginary-valued with positive imaginary part.\footnote{The positive sign ensures that $e^{\imagunit k_z(\vect{k}) z}$ does not blow up at infinity, also known as the Sommerfeld's radiation condition \cite{ChewBook}.} Particularly, $k_z(\vect{k})$ is real-valued within $\vect{k} \in \mathcal{D}$ and  imaginary-valued otherwise.
The spatial field $e(\vect{r},z)$ with $\{z\ge0\}$ obeying~\eqref{Helmholtz} is obtained via a 2D inverse spatial Fourier transform of~\eqref{Helmholtz_Fourier_solution}:
\begin{equation} \label{Helmholtz_Fourier_general_solution}
e(\vect{r},z) = \frac{1}{(2\pi)^2}\int_{\Real^2} E(\vect{k}) e^{\imagunit k_z(\vect{k}) z} e^{\imagunit \vect{k}^{\Ttran} \vect{r}} \, d\vect{k}.
\end{equation}
%Clearly, $e(\vect{r},z)$ satisfies the Helmholtz equation in~\eqref{Helmholtz}, as verified by direct substitution of~\eqref{Helmholtz_Fourier_general_solution} into~\eqref{Helmholtz} and exchanging of the integral and derivative operations.

%\subsection{Circularly-bandlimited Low-pass Filter} \label{sec:migration}
\subsection{System-Theoretic Interpretation of Wave Propagation} \label{sec:migration}

\begin{figure} [t!]
        \centering
        \begin{subfigure}[t]{\columnwidth} \centering 
	\begin{overpic}[width=.9\columnwidth,tics=10]{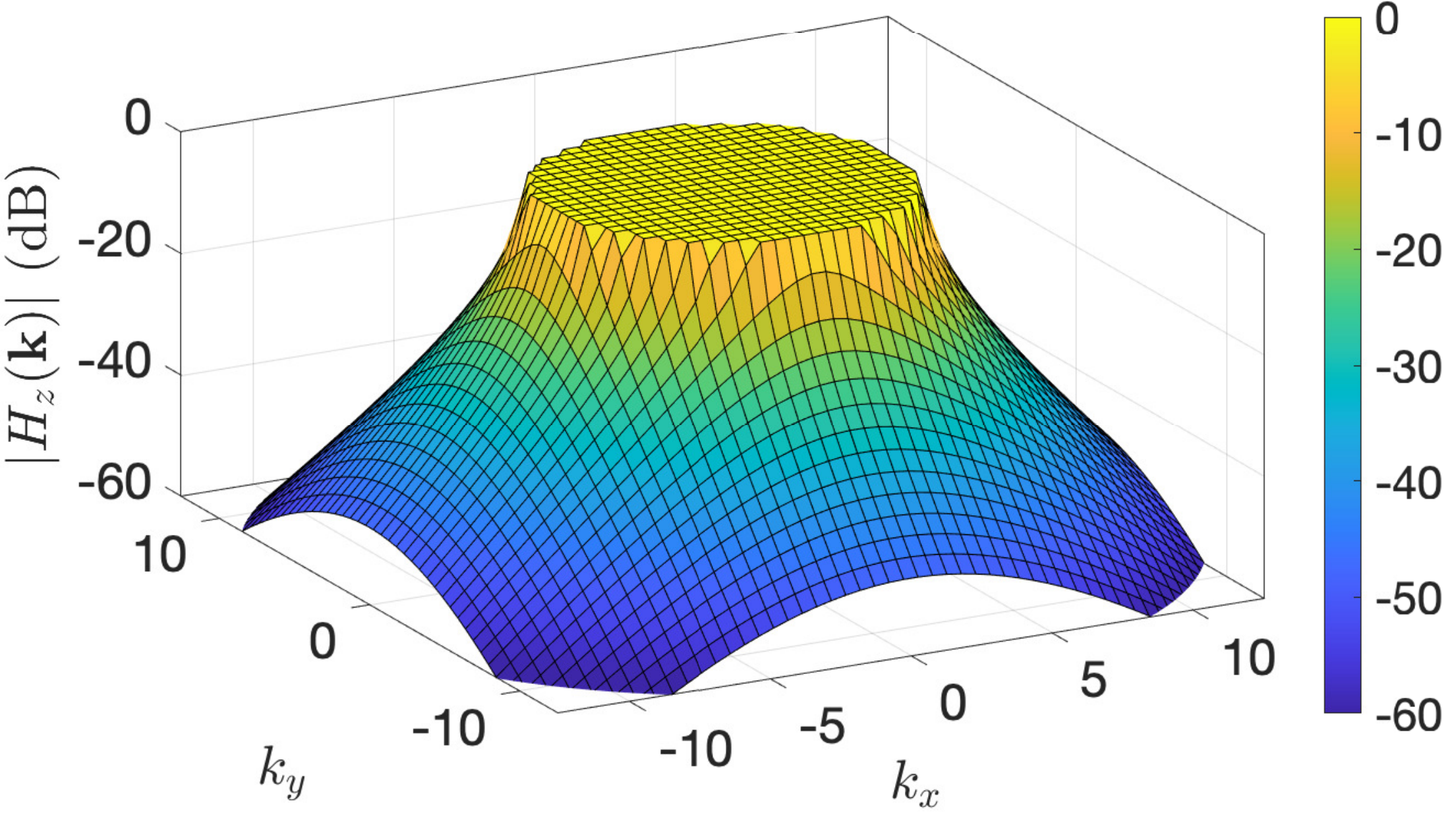}
\end{overpic} \vspace{-0.2cm}
                \caption{$z/\lambda = 1$.} \vspace{0.0cm}
                \label{fig:LSI_filter_1lambda} 
        \end{subfigure}   
             
        \begin{subfigure}[t]{\columnwidth} \centering  
        	\begin{overpic}[width=.9\columnwidth,tics=10]{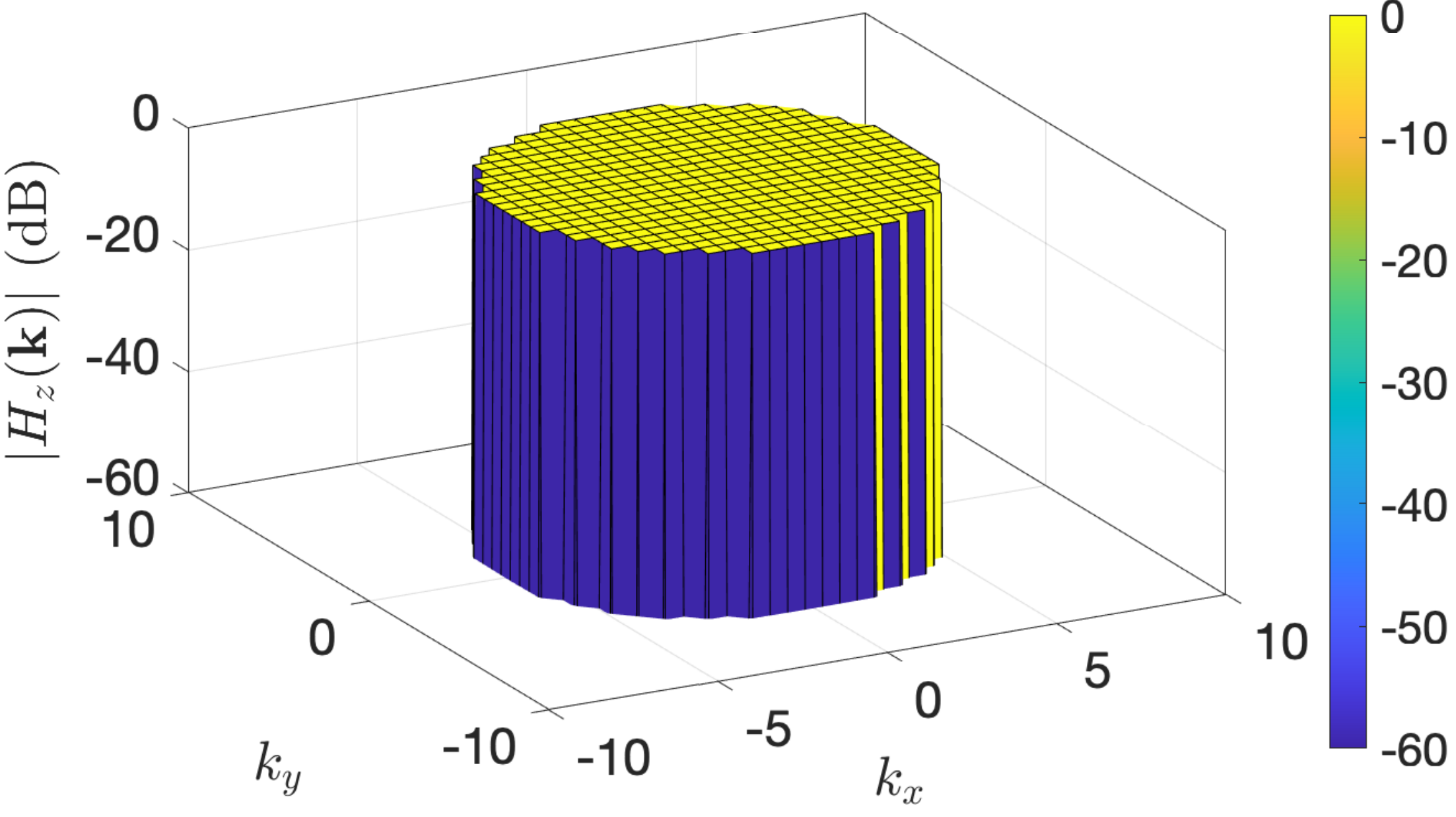}
\end{overpic}  \vspace{-0.0cm}
                \caption{$z/\lambda = 10$.} 
                \label{fig:LSI_filter_10lambda}
        \end{subfigure}\vspace{0.0cm}
        \caption{{\small Magnitude (in dB) of $H_z(\vect{k})$ in~\eqref{filter} for different $z$.}}
        \label{fig:LSI_filter}\vspace{-0.0cm}
\end{figure}

The spectrum-wise multiplication in~\eqref{Helmholtz_Fourier_solution} reveals that $e(\vect{r},z)$ with $\{z\ge0\}$ can be obtained by passing a 2D field $e(\vect{r},0)$ with spectrum $E(\vect{k})$ through an LSI system (see Fig.~\ref{fig:non_isotropic_spectrum}) with wavenumber response
\begin{equation} \label{filter}
H_z(\vect{k}) = e^{\imagunit k_z(\vect{k}) z}.
\end{equation}
%Let us rewrite \eqref{Helmholtz_Fourier_general_solution} as
%\begin{equation} \label{spectral_representation_z_filter}
%e(\vect{r},z) = \frac{1}{(2\pi)^2} \int_{\Real^2} E(\vect{k})  H_z(\vect{k}) e^{\imagunit \vect{k}^{\Ttran} \vect{r}} \, d\vect{k}
%\end{equation}
This operation is known in physics as \emph{migration} (e.g., \cite[Sec.~7]{MultivariateSPBook}), and it is a direct consequence of the Helmholtz equation in~\eqref{Helmholtz}, which acts as an LSI operator projecting the number of observable 3D field configurations onto a lower-dimensional 2D space \cite{Franceschetti}.  
%The behaviour of the migration filter $H_z(\vect{k})$ in~\eqref{filter} depend jointly on the spatial-frequency (wavevector) $\vect{k}$ and the $z$-plane at which the output field is observed, as discussed next.
%In the half-space $z\ge 0$, we consider only the positive solution in~\eqref{Helmholtz_Fourier_solution} as the negative exponential $e^{- \imagunit k_z(\vect{k}) z}$ will correspond to an increasingly higher oscillation along $z$.
Depending on $\vect{k}$, $H_z(\vect{k})$ in~\eqref{filter} is either an oscillatory or an exponentially-decaying function of ${z}$, i.e.,
\begin{equation} \label{filter_lowpass}
H_z(\vect{k}) = 
\begin{cases}
e^{\imagunit \sqrt{\kappa^2 - \|\vect{k}\|^2} \, z} & \vect{k}\in\mathcal{D} \\
e^{- \sqrt{\|\vect{k}\|^2 -\kappa^2} \, z} & \text{elsewhere}.
\end{cases}
\end{equation}
In the case $\vect{k} \in \mathcal{D}$, $H_z(\vect{k})$ is an all-pass filter that simply introduces a phase-shift along $z$. This implies that $e(\vect{r},z)$, at any $z$-plane is fully determined by $e(\vect{r},0)$. Outside of this support, $H_z(\vect{k})$ introduces an attenuation with exponential pace along $z$. 
%Notice that all wavevectors with constant $\|\vect{k}\|$ are attenuated the same, as the exponential decaying factor $k_z(\vect{k})$ in~\eqref{k_z} depends on $\|\vect{k}\|$ only. 
%Notice that the support of $H_z(\vect{k})$ is not fixed but depends on $z$. 
This is quantified in Fig.~\ref{fig:LSI_filter} where the magnitude of $H_z(\vect{k})$ is plotted (in dB) for $z/\lambda = \{1, 10\}$. At $z=10 \, \lambda$, the field outside of $\mathcal{D}$ is attenuated more than $60$~dB. This means that, for distances $z$ larger than tens of wavelengths (i.e., for any practical distance from scatterers), we may consider $H_z(\vect{k})$ as a \emph{low-pass filter} with maximum circular bandwidth 
\begin{equation} \label{measure_disk}
m(\mathcal{D}) = \pi \kappa^2
\end{equation}
which increases proportionally with the squared value of the operating frequency.
This will be assumed in the remainder.

\subsection{Plane-Wave Decomposition} \label{sec:physics_interpretation}

The integral in~\eqref{Helmholtz_Fourier_general_solution} has an intuitive physical interpretation in terms of a \emph{plane-wave decomposition} of $e(\vect{r},z)$ where each plane wave impinges on the point $(\vect{r},z)$ from direction $(\vect{k}/\kappa,k_z(\vect{k})/\kappa)$ and has complex-valued amplitude $E(\vect{k})$~\cite{PizzoJSAC20}. 
% -- determined by the angular selectivity of the scattering mechanism
There are two types of plane waves in nature: the discarded \emph{evanescent waves} and the remaining \emph{propagating waves}. Hence, the field $e(\vect{r},z)$ in~\eqref{Helmholtz_Fourier_general_solution} can be approximated as
\begin{align} \label{field_spectral_planewave}
e(\vect{r},z) & =  \int_{\mathcal{D}} E(\vect{k}) e^{\imagunit (\vect{k}^{\Ttran} \vect{r} + k_z(\vect{k}) z)} \, d\vect{k}.
\end{align}  
%Notice that the inclusion of migration $e^{\imagunit k_z(\vect{k}) z}$ into the Fourier spatial basis $e^{\imagunit \vect{k}^{\Ttran} \vect{r}}$ provides us with an alternative representation in the angular domain via cosine directions, with respect to the wavenumber domain. 
An equivalent representation of~\eqref{field_spectral_planewave} can be obtained in the angular domain via cosine directions, which can be related to elevation $\theta \in[0,\pi/2]$ and azimuth $\phi\in[0,2\pi)$ angles as
\begin{align} \label{wavevector}
\begin{pmatrix}
\vect{k} \\
k_z(\vect{k})
\end{pmatrix}
=
\begin{pmatrix}
\kappa \sin(\theta) \cos(\phi) \\
\kappa \sin(\theta) \sin(\phi) \\
\kappa \cos(\theta)
\end{pmatrix}.
\end{align}
Being interested in leveraging linear system theory and Fourier theory, we will privilege the wavenumber interpretation. The angular domain representation will be used in Section~\ref{sec:stationary} to characterize the angular selectivity of the scattering. 
%For example, the spherical coordinates provide an intuitive physical understanding of the wave propagation phenomena angularly. 
%Despite this useful property, for sampling purposes, we next pursue our analysis in the wavenumber domain instead. This will allow us to exploit the one-to-one correspondence with time-domain signals and Fourier theory, as seen in Section~\ref{sec:Preliminaries}.

\section{Nyquist Sampling and Reconstruction} \label{sec:Nyquist_sampling}

%The 2D sampling theorem states that any field $e(\vect{r}) \in \mathcal{B}_\mathcal{K}$ of wavenumber support $\vect{k}\in\mathcal{K}$ can be perfectly reconstructed from its noiseless samples taken at Nyquist rate \cite[Sec.~6]{SamplingBook}.
Let us observe $e(\vect{r},z)$ over an infinite $z$-oriented plane and drop the functional dependance on $z$, i.e., ${e(\vect{r},z)= e(\vect{r})}$, as the sole effect introduced along the $z$-axis is the low-pass filtering operation introduced by migration. 
%Next, we first review the multidimensional sampling theorem for a 2D field and apply it to solve the Nyquist sampling and reconstruction problems of $e(\vect{r},z)$ under different scattering conditions. 
%{\color{blue}Reconstruction of the field at different $z$ requires compensating for the phase delay accumulated by the field along its travel along the $z$-axis.}
Reconstruction of the field at different $z$ requires compensating for the phase delay accumulated by the field along its travel along the $z$-axis.
%Since the migration filter in \eqref{filter} introduces only a phase shift, reconstruct the field along the $z$ axis requires creating a space-shifted version of the field measured at a reference $z$-plane.
Differently from 3D reconstruction of an object, electromagnetic fields require to capture a single (not multiple) ``image'' of the field, which rests on the Huygen’s electromagnetic principle \cite{ChewBook}.
We consider scattering generated from a deterministic configuration of scatterers. Propagation into random media is treated in Section~\ref{sec:stationary}. 

%We will develop our theory under ideal sampling conditions and infinite observation area and provide some insights on the major implication of a non-ideal sampling at the end. 
%The impact of operating with a finite number of samples is evaluated numerically in Section~\ref{sec:numerical}.

\subsection{Half-wavelength Rectangular Sampling}

The classical half-wavelength rectangular sampling arises as the Nyquist sampling for a spectrum  with wavenumber support \begin{equation} \label{rect}
\mathcal{R} = \{\vect{k}=(k_x,k_y)\in\Real^2 : |k_x|\le\kappa, |k_y|\le\kappa\}.
\end{equation}
The problem of finding the Nyquist periodicity matrix $\vect{P}^\star_\mathcal{R}$ that yields the minimum separation among replicas is geometrically formulated as the arrangement of squares of sizes $2\kappa$ that achieves the highest density.
Due to the rectangular symmetry of $\mathcal{R}$ in~\eqref{rect}, this is solved, e.g., along the $k_x$-axis. The Nyquist sampling interval follows from~\eqref{optimal_sampling_time} as $Q^\star_x =  {\pi}/{\kappa} = \lambda/2$ while using ${\kappa=2\pi/\lambda}$.
Exchanging the $k_x$ and $k_y$ axes and rearranging the results into a matrix form yields \cite{Wang88}
\begin{equation} \label{rectangular_sampling}
\vect{Q}^\star_{\mathcal{R}} =  \frac{\lambda}{2} \, \vect{I}_2.
\end{equation}
The Nyquist density corresponding to~\eqref{rectangular_sampling} is given by
\begin{equation} \label{density_sampling_rect} 
\mu_{\mathcal{R}} = \frac{1}{|\det(\vect{Q}^\star_{\mathcal{R}})|} = \frac{4}{\lambda^2} %  \qquad \text{samples/m}^2.
\end{equation}
%Interestingly, \eqref{spatial_efficiency} is dominated by the ratio of the distances between a spatial sample and its closest neighbourhood, which is given by $\lambda/\sqrt{3}$ and $\lambda/2$ for the hexagonal and rectangular samplings, respectively.
The interpolating function for the reconstruction of $e(\vect{r}) \in \mathcal{B}_\mathcal{R}$ is obtained from its 1D counterpart $f_{\Omega}(t) = \sinc\left(\Omega t/\pi\right)$ while replacing $\Omega$ with $\kappa=2\pi/\lambda$ and  leveraging separability
\begin{equation} \label{interpolation_function_rect}
f_{\mathcal{R}}(\vect{r}) = \sinc\left(\frac{2 x}{\lambda}\right) \sinc\left(\frac{2 y}{\lambda}\right).
\end{equation}
Based on the physical discussion in Section~\ref{sec:physics_interpretation}, half-wavelength sampling involves a portion of evanescent waves into the analysis, which occurs only when the observation plane is located in the close proximity of the scatterers (see Fig.~\ref{fig:LSI_filter}). For any practical scenario, it leads to an underutilization of the wavenumber spectrum with consequent loss of efficiency. We will develop on this observation later on.
%Hence, a  should be used in wireless applications operating in the reactive near-field region only, as some of the samples are linearly redundant each other. 

\subsection{Isotropic Scattering Environments}
 
\begin{figure} [t!]
        \centering
        \begin{subfigure}[t]{\columnwidth} \centering 
	\begin{overpic}[width=0.8\columnwidth,tics=10]{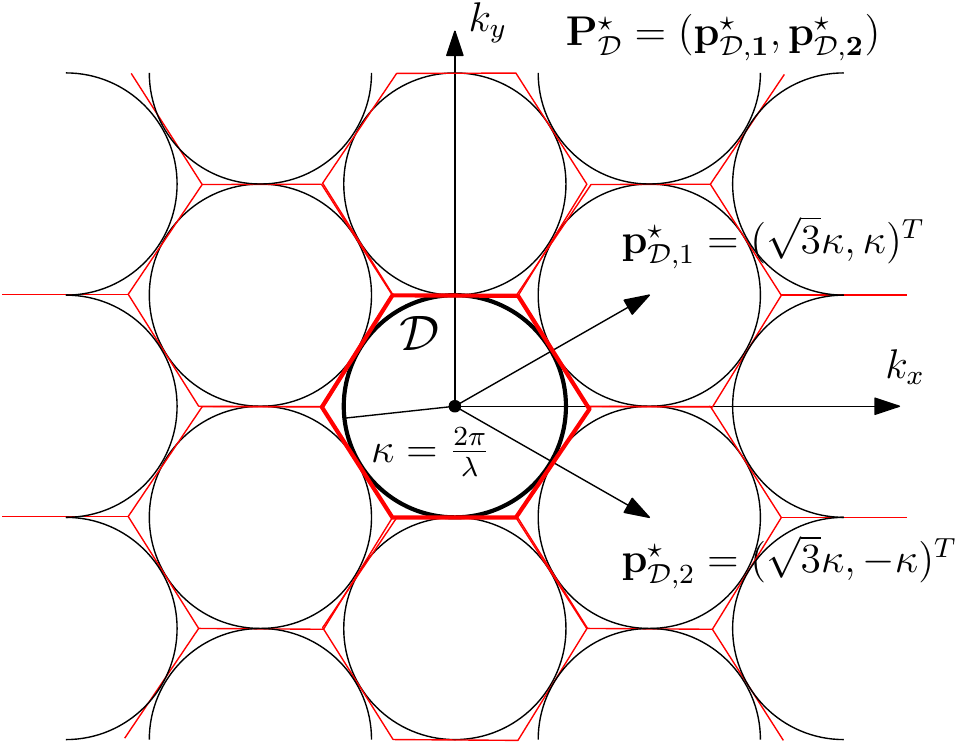}
\end{overpic} \vspace{-0.0cm}
                \caption{Circle packing in the wavenumber domain.} \vspace{0.5cm}
                \label{fig:circle_packing}  
        \end{subfigure}       
          
        \begin{subfigure}[t]{\columnwidth} \centering  
        	\begin{overpic}[width=.65\columnwidth,tics=10]{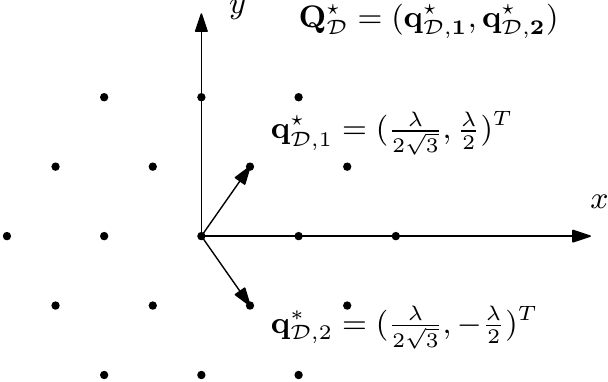}
\end{overpic}  \vspace{-0.0cm}
                \caption{Hexagonal sampling in the spatial domain.} 
                \label{fig:hexagonal_sampling}
        \end{subfigure}\vspace{0.0cm}
        \caption{{\small Nyquist sampling under isotropic scattering.}}
        \label{fig:sampling}\vspace{-0.0cm}
\end{figure}

If no directionality is enforced by the scatterers, $e(\vect{r})\in\mathcal{B}_\mathcal{D}$. This case physically corresponds to an isotropic scattering environment where plane waves carry equal power from all possible directions ${\theta \in[0,\pi/2]}$ and ${\phi\in[0,2\pi)}$ \cite{PizzoJSAC20,PizzoIT21}. 
Hence, the problem of finding the Nyquist periodicity matrix $\vect{P}^\star_\mathcal{D}$ is
geometrically found by solving a packing problem of circles of equal radius $\kappa$  \cite[Sec.~1.4]{MultivariateSPBook}. 
The solution is illustrated in Fig.~\ref{fig:circle_packing} and obtained by first inscribing each circle with a regular hexagon of side lengths ${2\kappa}/{\sqrt{3}}$ (red line) and, then, replicating this hexagonal arrangement periodically such that no aliasing occurs \cite[Eq.~(1.148)]{MultivariateSPBook}
\begin{align} \label{hex_replication_proof}
\vect{P}^\star_{\mathcal{D}} & =
\begin{pmatrix}
\sqrt{3} \kappa & \sqrt{3} \kappa  \\
\kappa  &  - \kappa
\end{pmatrix}.
\end{align}
Plugging~\eqref{hex_replication_proof} into~\eqref{Nyquist_sampling_matrix} as in~\cite[Eq.~(1.153)]{MultivariateSPBook} while substituting  $\kappa=2\pi/\lambda$ yields
\begin{align} \label{hexagonaLampling_2}  
\vect{Q}^\star_{\mathcal{D}} & =
\begin{pmatrix}
\frac{\pi}{\kappa \sqrt{3}} & \frac{\pi}{\kappa \sqrt{3}}  \\
\frac{\pi}{\kappa}  &  - \frac{\pi}{\kappa}
\end{pmatrix} 
%\mathop{=}^{(a)} 
=
\begin{pmatrix}
\frac{\lambda}{2 \sqrt{3}} & \frac{\lambda}{2 \sqrt{3}}  \\
\frac{\lambda}{2}  &  - \frac{\lambda}{2}
\end{pmatrix}
\end{align}
and it corresponds to a \emph{hexagonal sampling}, as illustrated in Fig.~\ref{fig:hexagonal_sampling}.
The Nyquist density follows from~\eqref{Nyquist_rate} as
\begin{equation} \label{density_sampling_hex}
\mu_{\mathcal{D}} = \frac{1}{|\det(\vect{Q}^\star_{\mathcal{D}} )|} = \frac{2 \sqrt{3}}{\lambda^2}. %  \qquad (\text{in samples/m}^2).
\end{equation}
Notice that $\vect{Q}^\star_{\mathcal{D}}$ is not unique as there exist other sampling matrices all achieving the same efficiency \cite{PETERSEN1962279,Agrell2004}, e.g., by rotating the wavenumber axes by $\pi/2$.
Compared to the half-wavelength rectangular sampling, we notice that $\mathcal{R}$ in~\eqref{rect} corresponds to the minimum squared embedding of $\mathcal{D}$. In turn, a hexagonal sampling has a \emph{sampling efficiency gain} of 
\begin{equation} \label{spatial_efficiency}
1-\frac{\mu_{\mathcal{D}}}{\mu_{\mathcal{R}}} = 1- \frac{\sqrt{3}}{2} = 13.4\%
\end{equation}
compared to rectangular sampling \cite[Sec.~1.4]{MultivariateSPBook}.

The interpolating function should be obtained by computing~\eqref{interp_function} over the same wavenumber support that is used for determining the Nyquist sampling matrix, which for any $e(\vect{r})\in\mathcal{B}_\mathcal{D}$ corresponds to the regular hexagon inscribing $\mathcal{D}$ (see Fig.~\ref{fig:circle_packing}). A suboptimal, yet accurate, solution is given below by computing~\eqref{interp_function} over $\mathcal{D}$ \cite[Sec.~VIII]{PETERSEN1962279}. 

\begin{lemma}  \label{sec:interpolation_function_iso}
Suppose that $e(\vect{r}) \in \mathcal{B}_{\mathcal{K}}$ with $\mathcal{K} = \mathcal{D}$, then perfect reconstruction of $e(\vect{r})$ is achieved by using~\eqref{cardinal_series_spatial} with Nyquist sampling matrix in \eqref{hexagonaLampling_2} and
\begin{equation} \label{interpolation_function_iso}
f_{\mathcal{D}}(\vect{r}) =  \frac{\pi}{\sqrt{3}} \, \jinc\left(\frac{2\pi \|\vect{r}\|}{\lambda}\right).
\end{equation}
\end{lemma}
\begin{proof}
Despite this result is given already in \cite[Eq.~(74)]{PETERSEN1962279}, the proof is only sketched for a general multi-dimensional isotropically bandlimited field with multivariate argument $\vect{r} \in\Real^N$. Since this result is preparatory for non-isotropic scattering, the proof is given in Appendix~\ref{app:interp_iso}.
\end{proof}

The above interpolating function provides a high reconstruction accuracy while limiting the complexity due to its rotational symmetry. This can also be easily extended to non-isotropic scenarios, as studied next.
%\begin{equation} \label{hexagonaLampling_2}
%\vect{Q}^\star_{\mathcal{D}}  = 
%\begin{pmatrix}
%\frac{\lambda}{2 \sqrt{3}} & \frac{\lambda}{2 \sqrt{3}}  \\
%\frac{\lambda}{2}  &  - \frac{\lambda}{2}
%\end{pmatrix}.
%\end{equation}
%This number is given by $\rho^* = {\lambda^2}/({2 \sqrt{3}})$ and $\rho^*_{\mathrm{near}} = {\lambda^2}/{4}$, for the far-field and near-field case, respectively, which is .
%since the density of spatial samples per unit area is given by $1/|\vect{Q}|$, the number of antennas per unit area $M$ to be deployed for the two sampling methods is measured in units of $\mathrm{antennas}/m^2$ and it is given by
%\begin{equation}
%M_{\mathrm{hex}} = \frac{1}{|\vect{Q}_{\mathrm{hex}}|} = \frac{\sqrt{3} \kappa^2}{2 \pi^2} = \frac{\lambda^2}{2 \sqrt{3}}, \qquad M_{\mathrm{rect}} = \frac{1}{|\vect{Q}_{\mathrm{rect}}|} = \frac{\kappa^2}{\pi^2} = \frac{\lambda^2}{4}.
%\end{equation}
  
\subsection{Non-isotropic Scattering Environments}

%\begin{figure} [t!]
%        \centering
%	\begin{overpic}[width=.85\columnwidth,tics=10]{ellipse_packing}
%\end{overpic} \vspace{-0.0cm}
%                \caption{Ellipse packing in the wavenumber domain.} \vspace{0.5cm}
%                \label{fig:ellipse_packing} 
%\end{figure}

\begin{figure} [t!]
        \centering
         \begin{subfigure}[t]{\columnwidth} \centering  
        	\begin{overpic}[width=.9\columnwidth,tics=10]{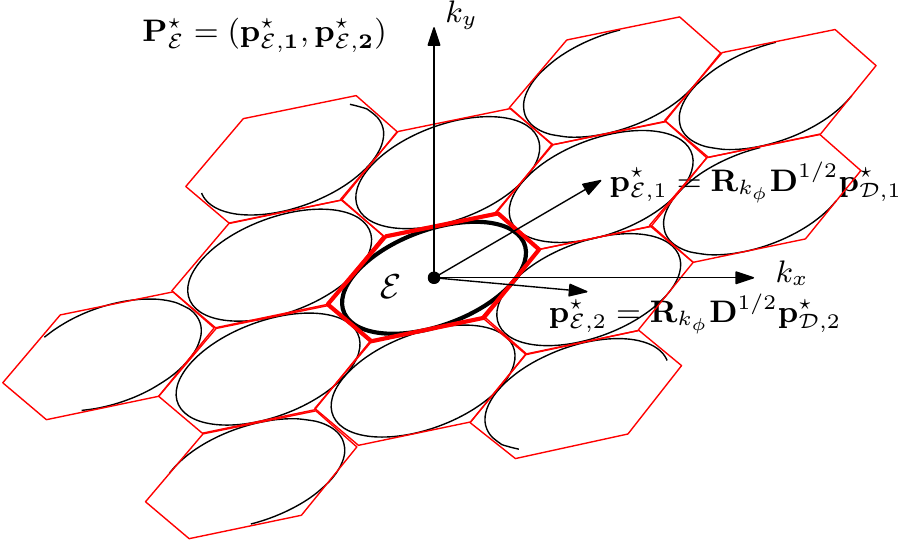}
\end{overpic}  \vspace{-0.0cm}
                \caption{Ellipse packing in the wavenumber domain.}  \vspace{0.5cm}
                \label{fig:ellipse_packing}
        \end{subfigure}
         \begin{subfigure}[t]{\columnwidth} \centering  
	\begin{overpic}[width=0.6\columnwidth,tics=10]{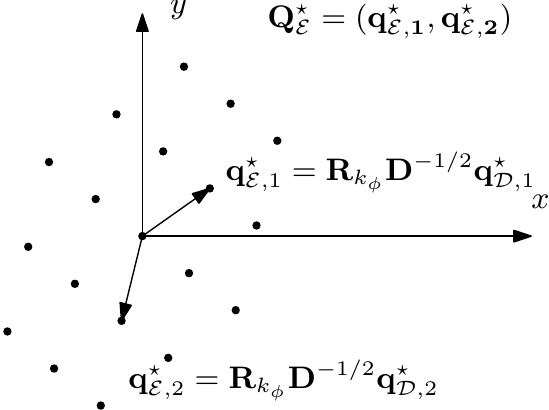}
\end{overpic} \vspace{-0.0cm}
                \caption{Elongated hexagonal sampling in the spatial domain.} \vspace{0.5cm}
                \label{fig:elongated_hexagonal_sampling}  
        \end{subfigure}   \vspace{0.0cm}
        \caption{{\small Nyquist sampling under non-isotropic scattering.}}
        \label{fig:ellipse_packing_circle}\vspace{-0.0cm}
\end{figure}

%\begin{figure}
%\centering
%\includegraphics[width=.6\columnwidth,tics=10]{ellipse_circle}
%\caption{Affine mapping between a circle and an ellipse.}
%\label{fig:ellipse_circle}
%\end{figure}

%Without loss of generality, we consider a channel field $e(\vect{r})$ whose wavenumber support is centered on the origin of the axes. This condition is similar to assuming the frequency spectrum of a time-domain baseband channel to be centered around $\omega=0$ and physically implies that the receive plane is oriented towards the modal direction of propagation.

%{\color{blue}Blockage effects make $e(\vect{r})$ directive in space limiting the field's spectrum to a smaller support $\mathcal{K} \subseteq \mathcal{D}$, as shown in Fig.~\ref{fig:non_isotropic_spectrum}. Hence, the scattering mechanism translates into another filtering operation, on top of migration filtering.}
Blockage effects make $e(\vect{r})$ directive in space limiting the field's spectrum to a smaller support $\mathcal{K} \subseteq \mathcal{D}$, as shown in Fig.~\ref{fig:non_isotropic_spectrum}. Hence, the scattering mechanism translates into another filtering operation, on top of migration filtering.
%Wave propagation in its most general form is non-isotropic and characterized by a spectral support $\mathcal{K} \subseteq \mathcal{D}$ of arbitrary shape, as shown in Fig.~\ref{fig:non_isotropic_spectrum}.
We will assume a centered support $\mathcal{K}$, which physically corresponds to an observation plane oriented perpendicularly to the modal propagation direction. The general case where the incoming field impinges from arbitrary direction is left for future work but it can, in principle, be addressed by translating the sampling theory of bandpass signals~\cite[Sec.~12.4]{OppenheimBook} to spatially bandlimited fields. 

Finding the Nyquist periodicity matrix $\vect{P}^\star_\mathcal{K}$ under arbitrary non-isotropic conditions is generally too complicated: $(i)$ $\mathcal{K}$ should be an analytically describable region (e.g., a circle or a rectangle); $(ii)$ for a non-connected spectrum one would have to find the tightest embedding (with minimum measure) of $\mathcal{K}$; and $(iii)$ we should be able to solve the resulting $\mathcal{K}$-packing problem geometrically \cite{PETERSEN1962279}.
%Except where the spectral functions occupy analytically describable regions (e.g., hypercubes, parallelepipeds, hyperspheres), the selection of the optimum vector set {uk} would seem to be a trial-and-error pro- cedure. It is, of course, difficult to conceive of a physical process whose spectrum might exhibit a meaningfully odd shape; moreover, in a prac- tical engineering application, one would have to weigh the disadvan- tages of inefficient sampling against those of instrumenting a complicated reconstruction function. 
%Instead, we investigate other possible sub-optimal Nyquist sampling strategies that are obtained by embedding $E(\vect{k})$ into, e.g., a square, rectangle, or a circle. The corresponding periodicity matrices and sampling strategies for these embeddings may be obtained as specified in Theorem~\ref{th:sampling_hex} and Theorem~\ref{th:sampling_rect} by rescaling opportunely.
In general, due to the above difficulties in solving the optimal problem, we resort to a sub-optimal strategy through the embedding of $\mathcal{K}$ into a larger and connected support. 
A possible embedding is provided by the 2D ellipse $\mathcal{E}\subseteq \mathcal{K}$ depicted in Fig.~\ref{fig:non_isotropic_spectrum}. This is defined as \cite[Sec.~2]{BoydBook}
\begin{equation} \label{ellipse}
%\mathcal{E} = \{\vect{k}\in\Real^2 : (\vect{k} - \vect{k}_0)^{\Ttran} \vect{P}^{-1} (\vect{k} - \vect{k}_0) \le 1\}
\mathcal{E} = \{\vect{k}\in\Real^2 : \vect{k}^{\Ttran} \vect{G}^{-1} \vect{k} \le \kappa^2\}
\end{equation}
where $\vect{G} \in\Real^{2\times 2}$ is a symmetric positive definite matrix (i.e., non-singular) that specifies how the ellipse extends in every direction from its center. Clearly, $\mathcal{D}$ in~\eqref{disk} is obtained as a particular instance of~\eqref{ellipse} by setting $\vect{G} = \vect{I}_2$.  
Although other embeddings may be used (e.g., square, circle, rectangle), the elliptical embedding is typically the most tight and analytically tractable.
Any ellipse is obtainable from a circle via an affine mapping ${\varphi : \mathcal{D} \to \mathcal{E}}$ defined as
\begin{equation}
\varphi(\vect{k}^\prime) = \vect{G}^{1/2} \vect{k}^\prime = \vect{k}
\end{equation}
to every $\vect{k}^\prime \in\mathcal{D}$.
The transformation involves a scaling of the Cartesian axes by $0 \le \sqrt{d_1}, \sqrt{d_2} \le 1$ and a counterclockwise rotation by an angle $k_\phi\in[0,2\pi]$; see Fig.~\ref{fig:non_isotropic_spectrum}. Altogether, we obtain
\begin{equation} \label{transf_matrix}
\vect{G}^{1/2} = \vect{R}_\phi \vect{D}^{1/2},
\end{equation}
where $ \vect{D}^{1/2} = \diag(\sqrt{d_1},\sqrt{d_2})$ with entries being the (normalized) lengths of the semi-axes of $\mathcal{E}$ and 
\begin{equation}
\vect{R}_{k_\phi} = 
\begin{pmatrix}
\cos(k_\phi)  & -\sin(k_\phi) \\
\sin(k_\phi) & \cos(k_\phi)
\end{pmatrix}.
\end{equation}
%:
%\begin{equation} \label{rotation_matrix}
%\vect{R} =
%\begin{pmatrix}
%\cos(\phi) & -\sin(\phi) \\
%\sin(\phi) & \cos(\phi)
%\end{pmatrix}.
%\end{equation}
%\begin{align}  \label{periodicity_matrix_ellipse}
%\vect{P}^\star_{\mathcal{E}} & =  \vect{G}^{1/2} \vect{P}^\star_{\mathcal{D}}
%%\\& 
%%= \vect{R} \vect{D}^{1/2} \vect{P}^\ast_{\boldsymbol{\xi}}
%\end{align}
%\begin{align}
%\vect{Q}^\star_{\mathcal{E}}  &
%%\mathop{=} 2\pi \left( \left(\vect{P}^\star_{\mathcal{D}}\right)^{-1} \vect{G}^{-1/2} \right)^{\Ttran} \\& \label{sampling_matrix_ellipse_proof}
%= (\vect{G}^{-1/2})^{\Ttran} \vect{Q}^\star_{\mathcal{D}} 
%\end{align}

The Nyquist periodicity matrix $\vect{P}^\star$ for the elliptical embedding of $\mathcal{K}$ is geometrically formulated as an ellipse packing problem, as illustrated in Fig.~\ref{fig:ellipse_packing}. The corresponding Nyquist sampling matrix is as follows.  

\begin{lemma} \label{th:sampling_non_iso}
Suppose that $e(\vect{r}) \in \mathcal{B}_\mathcal{K}$ with $\mathcal{K}\subseteq \mathcal{E}$, then the Nyquist sampling matrix $\vect{Q}^\star_{\mathcal{E}}$ is given by
\begin{equation}  \label{sampling_matrix_ellipse}
\vect{Q}^\star_{\mathcal{E}} =  (\vect{G}^{-1/2})^{\Ttran} \vect{Q}^\star_{\mathcal{D}} = \vect{R}_{k_\phi} \vect{D}^{-1/2} \vect{Q}^\star_{\mathcal{D}}
\end{equation}
where $\vect{Q}^\star_{\mathcal{D}}$ is defined  in~\eqref{hexagonaLampling_2}.
% and $(\vect{G}^{-1/2})^{\Ttran}$ is defined in~\eqref{inv_transformation}.
\end{lemma}
\begin{proof}
%The proof is given in Appendix~\ref{app:sampling_non_iso}. 
% and is built upon the affine mapping $\varphi : \mathcal{D} \to \mathcal{E}$ that turns circles into ellipses.
By virtue of the affine mapping, the ellipse packing problem is turned into a circle packing problem, whose solution is known and given by $\vect{P}^\star_{\mathcal{D}}$ in~\eqref{hex_replication_proof}. The Nyquist periodicity matrix is then given by
\begin{equation} \label{periodicity_matrix_ellipse}
\vect{P}^\star_{\mathcal{E}} =  \vect{G}^{1/2} \vect{P}^\star_{\mathcal{D}}.
\end{equation}
The Nyquist sampling matrix in \eqref{sampling_matrix_ellipse} follows from~\eqref{Nyquist_sampling_matrix} due to matrices invertibility while replacing $\vect{Q}^\star_{\mathcal{D}} = 2\pi (({\vect{P}^\star_{\mathcal{D}}})^{\Ttran})^{-1}$.
%\begin{equation} \label{inv_transformation}
%(\vect{G}^{-1/2})^{\Ttran} = \vect{R}_\phi \vect{D}^{-1/2}.
%\end{equation}
%This is obtained from by expanding the matrix inversion using~\eqref{periodicity_matrix_ellipse} and substituted the $\vect{Q}^\star_{\mathcal{D}}$ expression in~\eqref{hexagonaLampling_2}. 
\end{proof}

Notice that $\vect{Q}^\star_{\mathcal{E}}$ in Lemma~\ref{th:sampling_non_iso}  is obtained as a product between the isotropic Nyquist sampling matrix $\vect{Q}^\star_{\mathcal{D}}$ and the inverse matrix $(\vect{G}^{-1/2})^{\Ttran}$.
% = \vect{R} \vect{D}^{-1/2}$, which is obtainable from~\eqref{transf_matrix}. 
Consequently, the sampling scheme in this case is a rotated and scaled version of the one obtained under isotropic scattering conditions, i.e., an \emph{elongated hexagonal sampling}, as illustrated in Fig.~\ref{fig:elongated_hexagonal_sampling}.
How the hexagonal sampling is stretched in space is determined by $\vect{G}$ in~\eqref{transf_matrix}, whose parameters are uniquely determined by the angular selectivity of the scattering.

For the Nyquist sampling matrix $\vect{Q}^\star_{\mathcal{E}}$ in~\eqref{sampling_matrix_ellipse}, the Nyquist density is computed from~\eqref{Nyquist_rate} as
\begin{align} \label{density_sampling_ell}  
\mu_{\mathcal{E}}  = \frac{1}{|\det(\vect{Q}^\star_{\mathcal{E}} )|}  & \mathop{=}^{(a)} \frac{\sqrt{d_1} \sqrt{d_2}}{|\det(\vect{Q}^\star_{\mathcal{D}})|} 
%\mathop{=}^{(a)} \frac{1}{|\det(\vect{Q}^\star_{\mathcal{D}})|} \frac{1}{|\det((\vect{G}^{-1/2})^{\Ttran})|}  \\&
 \mathop{=}^{(b)} \sqrt{d_1} \sqrt{d_2} \mu_{\mathcal{D}} %  \qquad (\text{in samples/m}^2)
\end{align} 
where $(a)$ follows from~\eqref{sampling_matrix_ellipse} and the unitary property of rotation matrices, and $(b)$ from~\eqref{density_sampling_hex}. Here, $\sqrt{d_1}$ and $\sqrt{d_2}$ are the (normalized) lengths of the semi-axes of $\mathcal{E}$ in~\eqref{ellipse}; see Fig.~\ref{fig:non_isotropic_spectrum}.
Since $0 \le d_1, d_2 \le 1$, we have that $\mu_{\mathcal{E}} \le \mu_{\mathcal{D}}$ from which we see that the elongated hexagonal sampling has a {sampling efficiency gain} of
\begin{equation} \label{spatial_efficiency_ellipse}
1-\frac{\mu_{\mathcal{E}}}{\mu_{\mathcal{D}}} 
%1-\frac{\sqrt{d_1 d_2} \mu_{\mathcal{D}}}{\mu_{\mathcal{R}}} 
= 1- \sqrt{d_1} \sqrt{d_2}
\end{equation}
compared to hexagonal sampling, and of
\begin{equation} \label{spatial_efficiency_ellipse_rect}
1-\frac{\mu_{\mathcal{E}}}{\mu_{\mathcal{R}}} 
%=  1-\frac{\sqrt{d_1 d_2} \mu_{\mathcal{D}}}{\mu_{\mathcal{R}}} 
= 1- \frac{\sqrt{3}}{2} \sqrt{d_1} \sqrt{d_2}
\end{equation}
compared to half-wavelength sampling, as obtained by using \eqref{spatial_efficiency} and \eqref{density_sampling_ell}.
%This is zero when $\mathcal{E}=\mathcal{D}$ (i.e., isotropic case).
For example, consider a scenario where the scatterers subtend a solid angle spanning the entire azimuthal horizon and half range in elevation angle. This is exactly modeled as an ellipse $\mathcal{E}$ of parameters $\sqrt{d_1}=1$ and $\sqrt{d_2}=0.5$. The sampling efficiency gains in~\eqref{spatial_efficiency_ellipse} and~\eqref{spatial_efficiency_ellipse_rect} are approximately equal to $30\%$ and $39\%$, respectively.

The interpolating function for the elliptical embedding is given next.
% Notice that this gain is even higher when comparing the optimal hexagonal sampling scheme to a rectangular sampling at $\lambda/2$. In this latter case the achievable gain is $1- {\sqrt{3}}/({2\sqrt{2}}) \approx 39\%$. 

\begin{lemma} \label{sec:interpolation_function_noniso}
Suppose that $e(\vect{r}) \in \mathcal{B}_{\mathcal{K}}$ with $\mathcal{K} \subseteq \mathcal{E}$, then perfect reconstruction of $e(\vect{r})$ is achieved by using~\eqref{cardinal_series_spatial} with Nyquist sampling matrix in \eqref{sampling_matrix_ellipse} and
\begin{equation} \label{interpolation_function_noniso}
f_{\mathcal{E}}(\vect{r}) =   \frac{\pi}{\sqrt{3}} \, \jinc\left(\frac{2 \pi \| \vect{D}^{1/2} \vect{r}\|}{\lambda} \right)
\end{equation}
where $\vect{D}^{1/2} = \diag(\sqrt{d_1},\sqrt{d_2})$ is diagonal with elements given by the (normalized) lengths of the semi-axes of $\mathcal{E}$ in~\eqref{ellipse}.
\end{lemma}
\begin{proof}
The proof is given in Appendix~\ref{app:interp_noniso}.
\end{proof}

Notice that~\eqref{interpolation_function_noniso} is obtained from its isotropic counterpart~\eqref{interpolation_function_iso} by applying the same affine map used for sampling, except for the rotation, as it has no effect on $f_{\mathcal{D}}(\vect{r})$.
%This is a general fact that as the symmetry exhibited by the wavenumber support $\vect{k}\in\mathcal{K}$ maps directly to the corresponding interpolating function $f_{\mathcal{K}}(\vect{r})$.

%\subsection{3D Reconstruction}
% 
%As seen in Section~\ref{sec:electromagnetic}, the propagation behavior of $e(\vect{r})$ along $z$ is exactly known a-priori. 
%%, as it does not depend on the scattering conditions, but solely by the $z$-plane at which we measure the field (see also Fig.~\ref{fig:non_isotropic_spectrum}).
%In particular, for values of $z$ larger than few wavelengths, the electromagnetic field measured at any $z$ plane is nonetheless a space-shifted version of itself. 
%Thus, once the 2D field $e(\vect{r},z_0)$ with $z_0\ge 0$ is reconstructed through~\eqref{cardinal_series_spatial}, the 3D volumetric reconstruction of $e(\vect{r},z)$ for any $z\ge z_0$ is simply achieved via a constant phase-shift. 
%Differently from classical 3D volumetric reconstruction of an object, electromagnetic fields require to capture a single (not multiple) ``image'' of the field.

\section{Degrees of freedom} \label{sec:DoF}
%
%We now turn to the problem of computing the number of DoF that can be extracted from a 1D, 2D, and 3D rectangular region of space via a 3D electromagnetic field $e(\vect{r},z) \in\mathcal{B}_\mathcal{K}$. By virtue of the bandlimited property, we expect that only a finite number of samples is needed to specify the field approximately, the minimum of which denotes the DoF of $e(\vect{r},z)$. 

The DoF of an electromagnetic field $e(\vect{r})$, within a given observation area, corresponds to the minimum number of samples needed to represent the field over the same region. Hence, the DoF per unit of area coincide to the Nyquist sampling density that should be used for reconstruction. We next compute the DoF of all three propagation scenarios seen in Sec.~\ref{sec:Nyquist_sampling} and compare the results to the Nyquist density previously defined.

Let us consider a field with squared wavenumber support of side length $2 \kappa$, i.e., $e(\vect{r}) \in\mathcal{B}_\mathcal{R}$. Due to symmetry, focus on the $k_x$-axis only.
Over a spatial interval of length $L$, since there are $2/\lambda$ significant samples/m, the total number of samples is
%The number of DoF of $e(t) \in\mathcal{B}_\Omega$ that can be extracted from an interval of duration $T$ is given by~\eqref{Shannon_DoF}. This follows from the orthogonality of the shifted-versions of $f_\Omega(t)$ in~\eqref{cardinal_series} over $t\in[0,T]$ as $\Omega T\to \infty$ \cite[Eq.~(5)]{Unser}, which serve as a basis set of functions for $\mathcal{B}_\Omega$. 
%Similarly, due to separability, the shifted-versions of $f_\mathcal{R}(\vect{r})$ in~\eqref{interpolation_function_rect} becomes orthogonal  as $L/\lambda \to \infty$. Since  $e(\vect{r}) \in\mathcal{B}_\mathcal{R}$ has $2/\lambda$ significant samples/m, the total number of samples within a segment of length $L$ is
\begin{equation}
\dof_{[-\kappa,\kappa]} = \frac{2 L}{\lambda}
\end{equation}
which is the spatial counterpart of the Shannon's formula~\eqref{Shannon_DoF}. The physical implication of \eqref{DoF_rect} is that only a finite number of directions carries the essential information contained in the field and leads to perfect reconstruction as $L/\lambda \to\infty$.
%Naturally, over any squared region $\mathcal{A}\subseteq \Real^2$ of length $L$,
Generalization to a squared region of size $L$ is straightforward 
\begin{equation} \label{DoF_rect}
\dof_\mathcal{R} = \left(\frac{2 L}{\lambda}\right)^2.
\end{equation}
%Plugging~\eqref{interpolation_function_rect} into~\eqref{cardinal_series_spatial}, it follows that the DoF that can be extracted from a 2D squared region $\mathcal{A}\subseteq \Real^2$ of side length $L$ are
Moreover, the expansion of a rectangle into a parallelepiped does not offer additional DoF due to migration filter \cite{MarzettaISIT,PizzoSPAWC20}, which linearly relates the field's configurations at different $z$-planes and its rooted in the Huygen's electromagnetic principle \cite[Sec.~III]{Kennedy2007}\cite{PizzoIT21}. 
Expectedly, the DoF per unit of area returns the Nyquist density in \eqref{density_sampling_rect}, $\mu_\mathcal{R} = \dof_\mathcal{R}/L^2$.
This provides an intuitive interpretation of the number of DoF as the minimum number of samples within a given area that are needed to perfectly reconstruct the field (under ideal sampling). Considering more samples will only add redundancy to the discrete representation of the field.

The DoF formula in \eqref{DoF_rect} for squared bandlimited fields is a particular instance of Landau's formula~\cite{Landau,Franceschetti}
\begin{equation} \label{DoF_rect_area}
\dof_\mathcal{R} = \left(\frac{\kappa L}{\pi}\right)^2 = \frac{m(\mathcal{A}) m(\mathcal{R})}{(2 \pi)^2}
\end{equation}
function of the field's bandwidth ${m(\mathcal{R})= 4\kappa^2}$ (rad/m)$^2$ and observation area ${m(\mathcal{A}) = L^2}$ m$^2$. 
Expectedly, under isotropic scattering,
 \begin{align}  \label{DoF_D}
\dof_\mathcal{D} &  = \frac{m(\mathcal{A}) m(\mathcal{D})}{(2 \pi)^2} = \pi \left(\frac{L}{\lambda}\right)^2
\end{align}
where we have used $m(\mathcal{D}) = \pi \kappa^2$. 
%An intuitive explanation of \eqref{DoF_D} may be found in the 2D Fourier harmonics with fundamental frequency of $2\pi/L$ rad/m, as they provide a basis set of functions over $\mathcal{A}$. Counting the wavenumber samples taken at $2\pi/L$ rad/m apart within a circular support $\mathcal{D}$ of radius $\kappa=2\pi/\lambda$ yields \eqref{DoF_D} \cite{PizzoSPAWC20}.\footnote{Particularly,~\eqref{DoF_D} is obtained by rescaling the Cartesian axes in the wavenumber domain by $L/2\pi$ and solving a Gauss circle problem \cite{Wolfram}.}
Replicating the result for $\mu_\mathcal{R}$ to the isotropic case requires normalizing \eqref{DoF_D} by $m(\mathcal{A})$. This yields a slightly lower value of $\mu_\mathcal{D}$ in \eqref{density_sampling_hex}. The source of this error is the (suboptimal) hexagonal embedding applied prior to derive \eqref{hexagonaLampling_2} and associated density in \eqref{density_sampling_hex} that is not accounted by the DoF computation. This is bypassed by the DoF formula.
% ; see Fig.~\ref{fig:sampling}(a).
Clearly, the \emph{DoF loss} is equal to
\begin{align}   \label{DoF_loss}
 & 1 - \frac{\dof_{\mathcal{R}}}{\dof_{\mathcal{D}}} 
 = 1- \frac{m(\mathcal{D})}{m(\mathcal{R})} = 1- \frac{\pi}{4} \approx 21.5\%.
\end{align}
Hence, roughly one-fifth of the DoF generated by the scatterers are practically lost during transfer of information, when evanescent waves included in $\mathcal{R}$ are discarded.

%\begin{equation} \label{dof_D_nyquist}
%\mu_\mathcal{D} \ne \dof_\mathcal{D}/L^2 = \left(\frac{\pi}{\lambda}\right)^2
%\end{equation}

%In practice, due to the scattering mechanism, $e(\vect{r}) \in\mathcal{B}_\mathcal{K}$. 
%Solving the Nyquist sampling problem exactly for any $\mathcal{K}$ of arbitrary shape is hard without recurring to an embedding, which is true even under isotropic propagation; see Fig.~\ref{fig:circle_packing}. Consequently, the shifted-versions of $f_\mathcal{K}(\vect{r})$ in~\eqref{cardinal_series_spatial} are not orthogonal in general.
%A possible orthogonal basis set of functions over $\vect{r}\in\mathcal{A}$ is represented by the 2D Fourier harmonics with fundamental frequency of $2\pi/L$ rad/m. This yields the Fourier series expansion, which provides an exact representation of $e(\vect{r})$ asymptotically as ${L/\lambda \to \infty}$.

%We rescale the  by L/2π so that samples are now separated by a unit distance; that is, they form a 2D lattice. By doing so, the problem reduces to count the number of lattice points within a centered disk of radius κL/2π, i.e., a Gauss circle problem [29]. The solution is approximately given by [29, Eq. (7)]
Building on the above results, under non isotropic scattering,
\begin{equation} \label{DoF_K}
\dof_\mathcal{K} = \frac{m(\mathcal{A}) m(\mathcal{K})}{(2 \pi)^2}.
\end{equation}
With an elliptical embedding,~\eqref{DoF_K} reduces to
%The area reduction incurred from $\mathcal{D}$ to $\mathcal{K}$ inevitably entails a decrease of the number of DoF that can be extracted via $e(\vect{r})$. Indeed, the DoF carried by an electromagnetic field under non-isotropic propagation are upper bounded by those of an isotropic scenario.
%For an elliptical embedding of $\mathcal{K}$ into $\mathcal{E}$, this DoF reduction can be approximately be quantified as 
\begin{align} \label{DoF_ell}
\dof_{\mathcal{E}} =  \frac{m(\mathcal{A}) m(\mathcal{E})}{(2 \pi)^2} 
%&\mathop{=}^{(a)} \sqrt{d_1 d_2}\left\lceil \frac{m(\mathcal{A}) m(\mathcal{D})}{(2 \pi)^2} \right\rceil\\&
=\sqrt{d_1} \sqrt{d_2} \, \dof_{\mathcal{D}}
\end{align}
which is due to the affine mapping ${\varphi : \mathcal{D} \to \mathcal{E}}$ yielding the following relation in terms of Lebesgue measures: %the Lebesgue measure of the image set $\mathcal{E}$ is found by multiplying the measure of $\mathcal{D}$ by the determinant of the transformation matrix $\vect{G}^{1/2}$ as 
\begin{align}
m(\mathcal{E})  = m(\mathcal{D}) |\det(\vect{G}^{1/2})|  = m(\mathcal{D}) \sqrt{d_1} \sqrt{d_2}.
\end{align}
Since $0\le d_1,d_2\le 1$, it thus follows that $\dof_{\mathcal{E}} \le \dof_{\mathcal{D}}$, as it should be (see, e.g.,\cite{PoonDoF}). In fact, non-isotropic scattering implies a reduced dimensionality of the function space spanned by the field, which translates into a lower number of DoF. This is the reason why a lower Nyquist sampling rate is required with non-isotropic scattering; see~\eqref{density_sampling_ell} and \eqref{spatial_efficiency_ellipse}. For the non-isotropic scenario, the relationship between DoF and Nyquist density is as tight as the embedding of $\mathcal{K}$ with $\mathcal{E}$.

\section{Stochastic propagation} \label{sec:stationary}

So far, we have considered a deterministic configuration of scatterers that is known a priori.
However, this assumption requires an accurate predictions of wave propagation, which relies on numerical electromagnetic solvers of Maxwell's equations, and hence, it is too site-specific \cite{PizzoTWC21}.
A convenient way out to embrace different propagation conditions is by modeling $e(\vect{r})$ as a circularly symmetric complex Gaussian stationary random field \cite{MarzettaISIT,PizzoJSAC20,PizzoIT21,MarzettaNokia}. 
Given this, we next show how the results derived in Section~\ref{sec:Nyquist_sampling} and Section~\ref{sec:DoF} for a deterministic field can be extended to ensembles of random electromagnetic fields.
%A field of this sort is representative of wave propagation into hypothetically different environments.
%so that a second-order description fully describes the field statistically

\subsection{Spectral Characterization} \label{sec:wave_propagation}  

%In the wavenumber domain, a 3D spatial Fourier transform of \eqref{Helmholtz_corr} yields the spherical constraint $\|\vect{k}\|^2 + k_z^2 = \kappa^2$ on the power spectral density 
%\begin{equation}
%S(\vect{k},k_z) = \int_{\Real^2} \int_{-\infty}^{\infty} c(\vect{r},z) e^{-\imagunit (\vect{k}^{\Ttran} \vect{r} + k_z z)} \, d\vect{r} dz
%\end{equation}
%of $e(\vect{r},z)$.
%Thus, the most general form of $S(\vect{k},k_z)$ is \cite{MarzettaISIT,PizzoJSAC20}:
%\begin{equation} \label{3Dpower_spectral_density}
%S(\vect{k},k_z) = A^2(\vect{k},k_z) \delta(\|\vect{k}\|^2 + k_z^2 - \kappa^2)
%\end{equation}
%for any arbitrary real-valued non-negative function $A(\vect{k},k_z)$, called spectral factor. After computing the inverse Fourier transform of~\eqref{3Dpower_spectral_density} over $k_z$ and embedding all constants into $A(\vect{k},k_z)$ we obtain \cite{PizzoJSAC20}
%\begin{equation} \label{3Dpower_spectral_density_z}
%S(\vect{k},z) = \frac{A^2(\vect{k})}{k_z(\vect{k})}  e^{\imagunit k_z(\vect{k}) z}
%\end{equation}
%where we define $A(\vect{k}) = A(\vect{k},k_z(\vect{k}))$.
%The aucorrelation function $c(\vect{r},z)$ is obtained by taking the 2D inverse spatial Fourier transform of \eqref{3Dpower_spectral_density_z} over $\vect{k}$ as

 \begin{figure*}[t!]
 \begin{subfigure}{0.33\textwidth}
  \includegraphics[width=\linewidth]{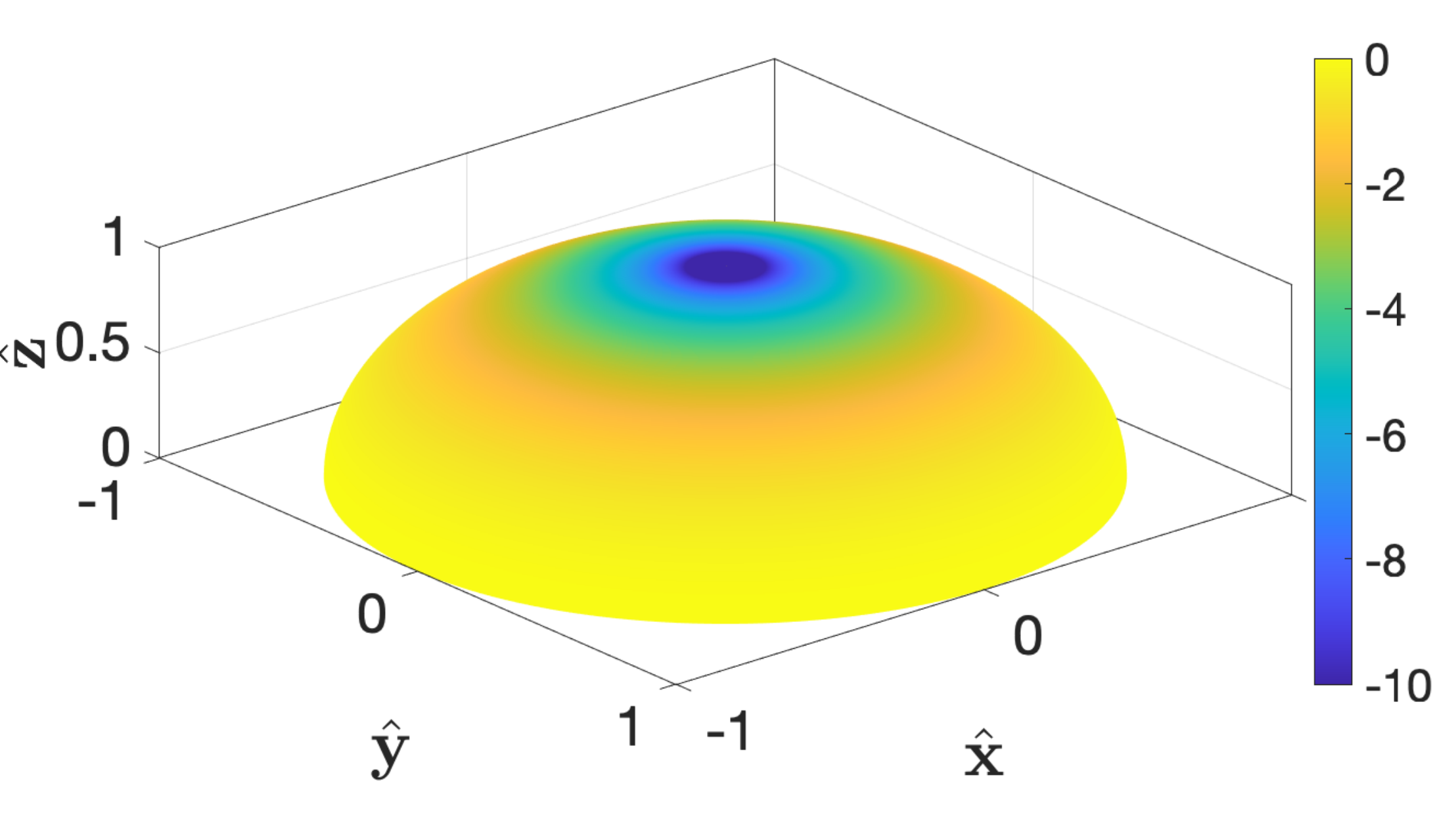}
  \caption{Isotropic}\label{fig:Sphe_iso}
\end{subfigure}\hfill
\begin{subfigure}{0.33\textwidth}
  \includegraphics[width=\linewidth]{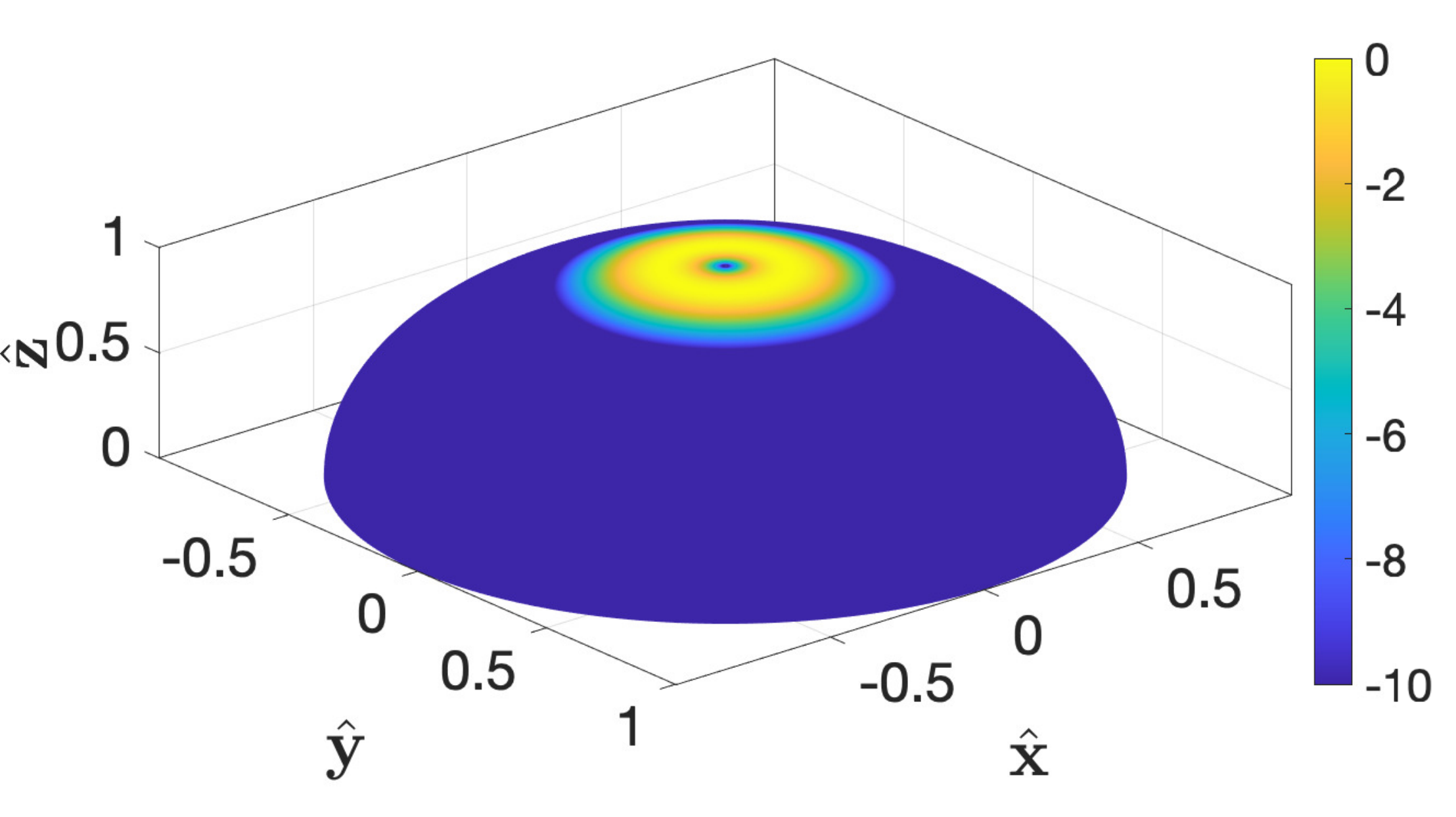}
  \caption{Non-isotropic}\label{fig:Sphe_noniso}
\end{subfigure}\hfill
\begin{subfigure}{0.33\textwidth}%
  \includegraphics[width=\linewidth]{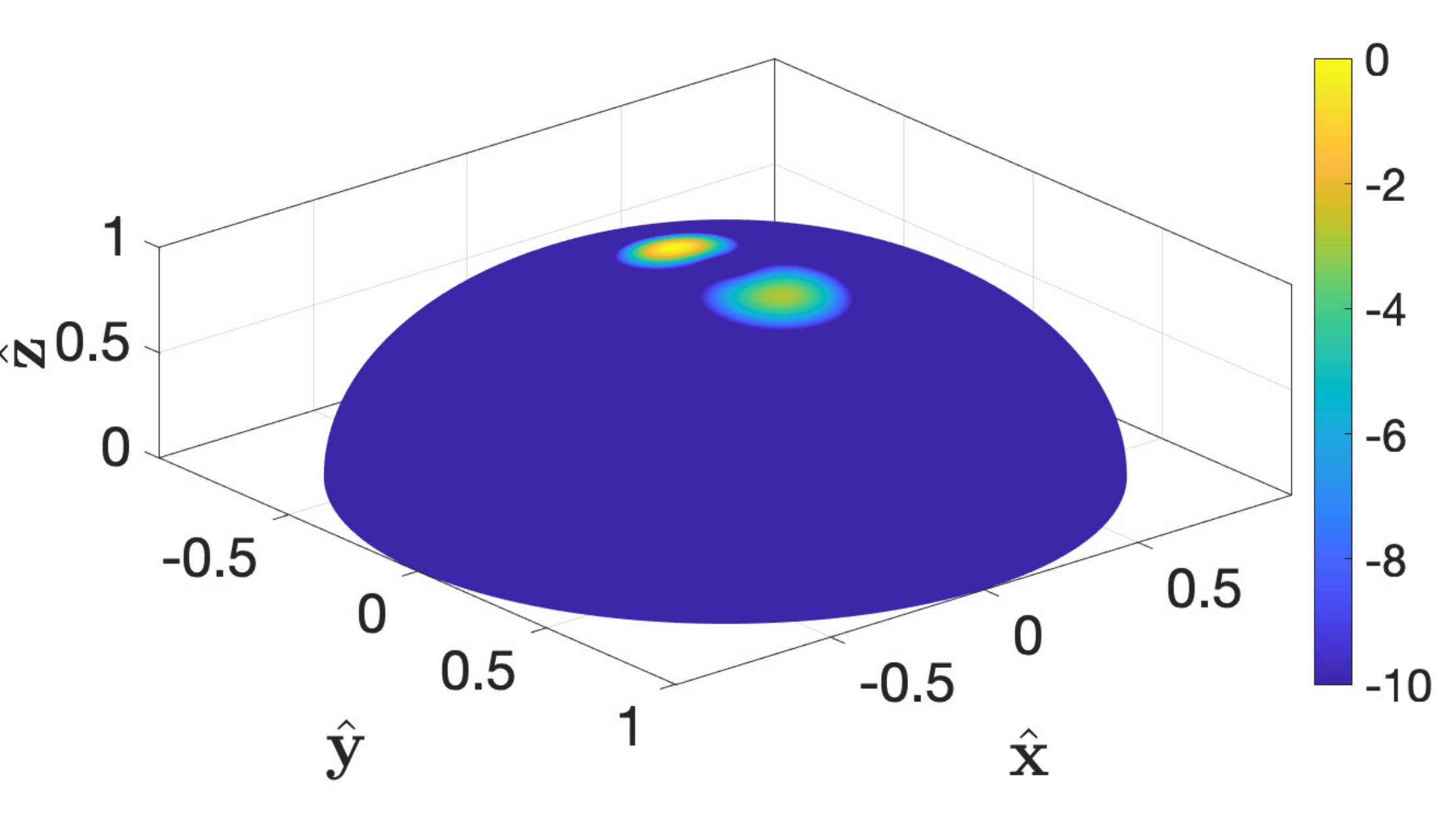}
  \caption{Non-isotropic mixture}\label{fig:Sphe_noniso_cluster}
\end{subfigure}
\begin{subfigure}{0.33\textwidth}
  \includegraphics[width=\linewidth]{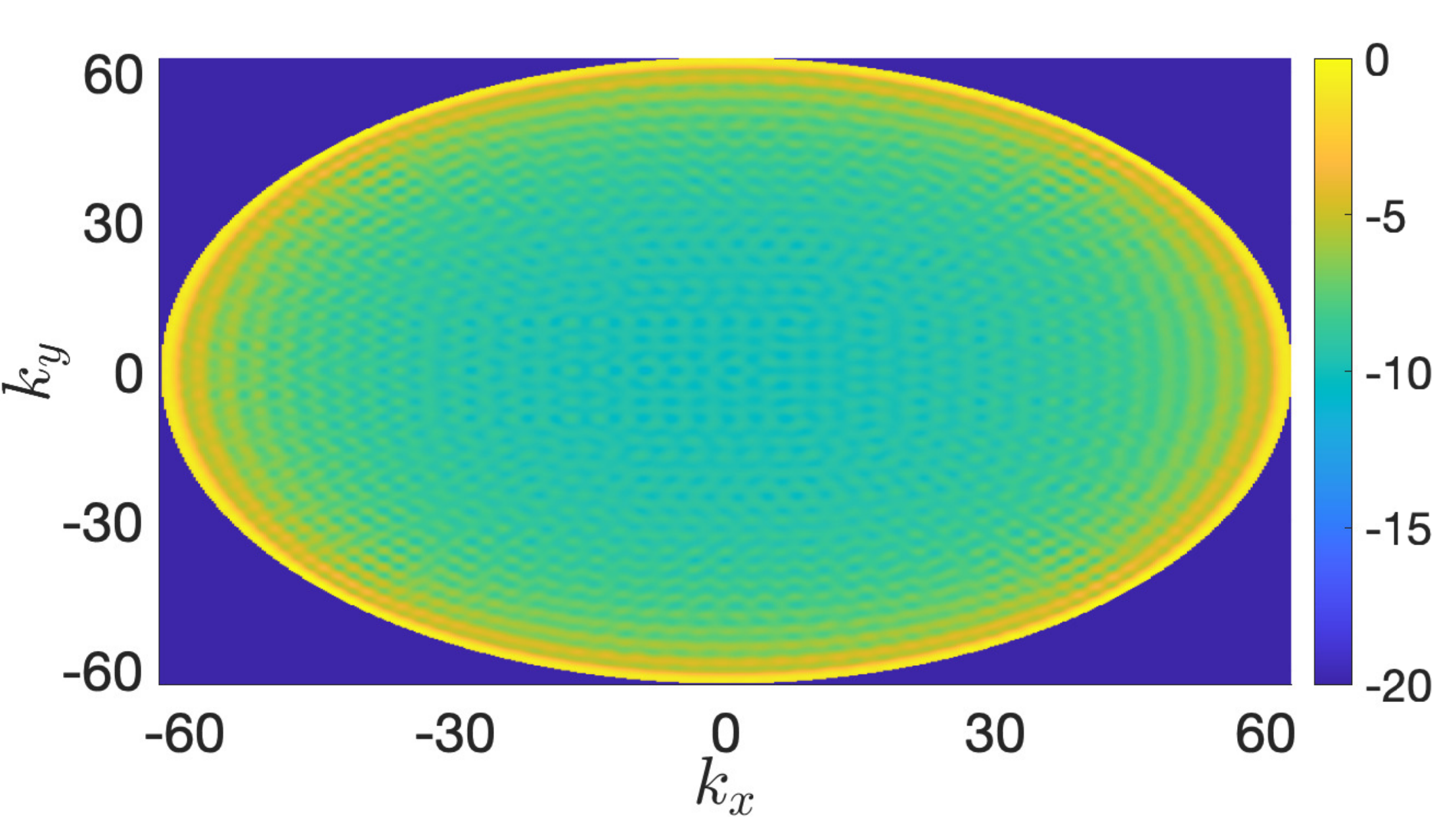}
  \caption{Isotropic}\label{fig:Wave_iso}
\end{subfigure}\hfill
\begin{subfigure}{0.33\textwidth}
  \includegraphics[width=\linewidth]{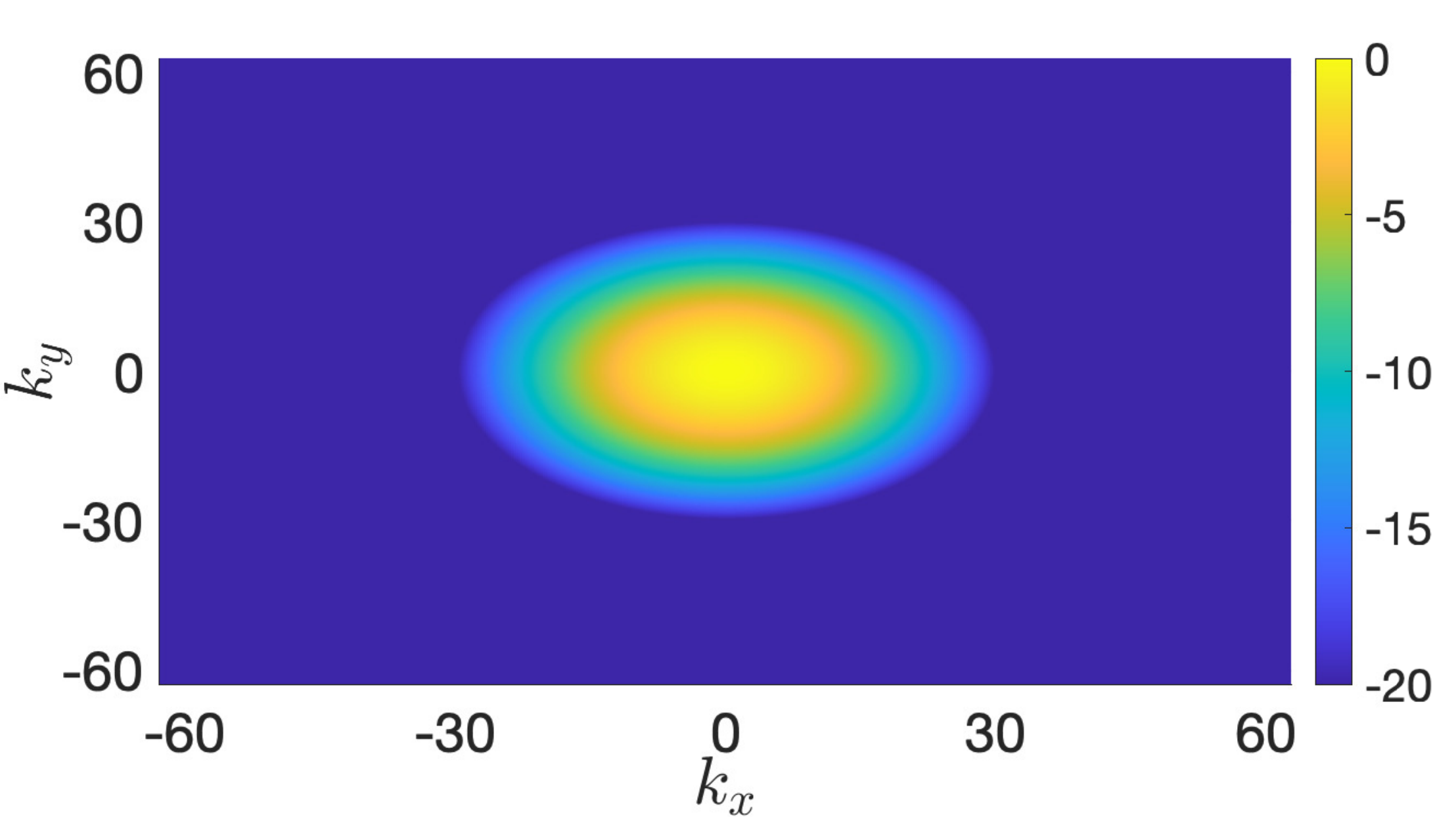}
  \caption{Non-isotropic}\label{fig:Wave_noniso}
\end{subfigure}\hfill
\begin{subfigure}{0.33\textwidth}%
  \includegraphics[width=\linewidth]{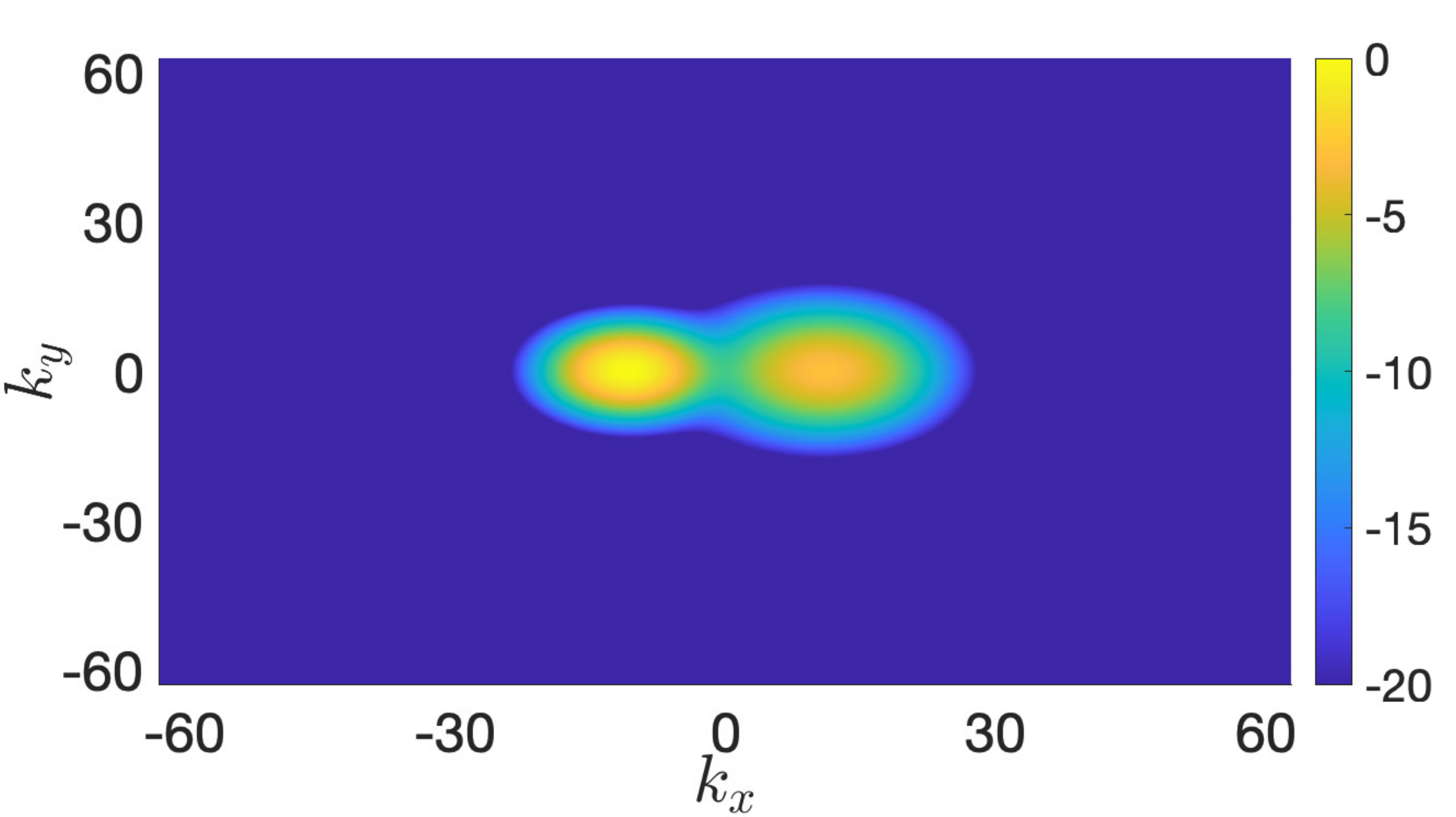}
  \caption{Non-isotropic mixture}\label{fig:Wave_noniso_cluster}
\end{subfigure}
\caption{{\small Scattering under isotropic and non-isotropic conditions. The first row shows the spectral factor $\sin(\theta) A(\theta,\phi)$ in~\eqref{mixture_pdf} with $(\theta,\phi) \in [0,\pi/2] \times [0,2\pi)$. The second row shows the PSD $S(\vect{k})$ in~\eqref{psd}.}}%with centered single cluster, and non-isotropic mixture with two clusters}
\label{fig:Spectrum}
\end{figure*}  

 \begin{figure*}[t!]
\begin{subfigure}{0.33\textwidth}
\vspace{-.3cm}
  \includegraphics[width=\linewidth]{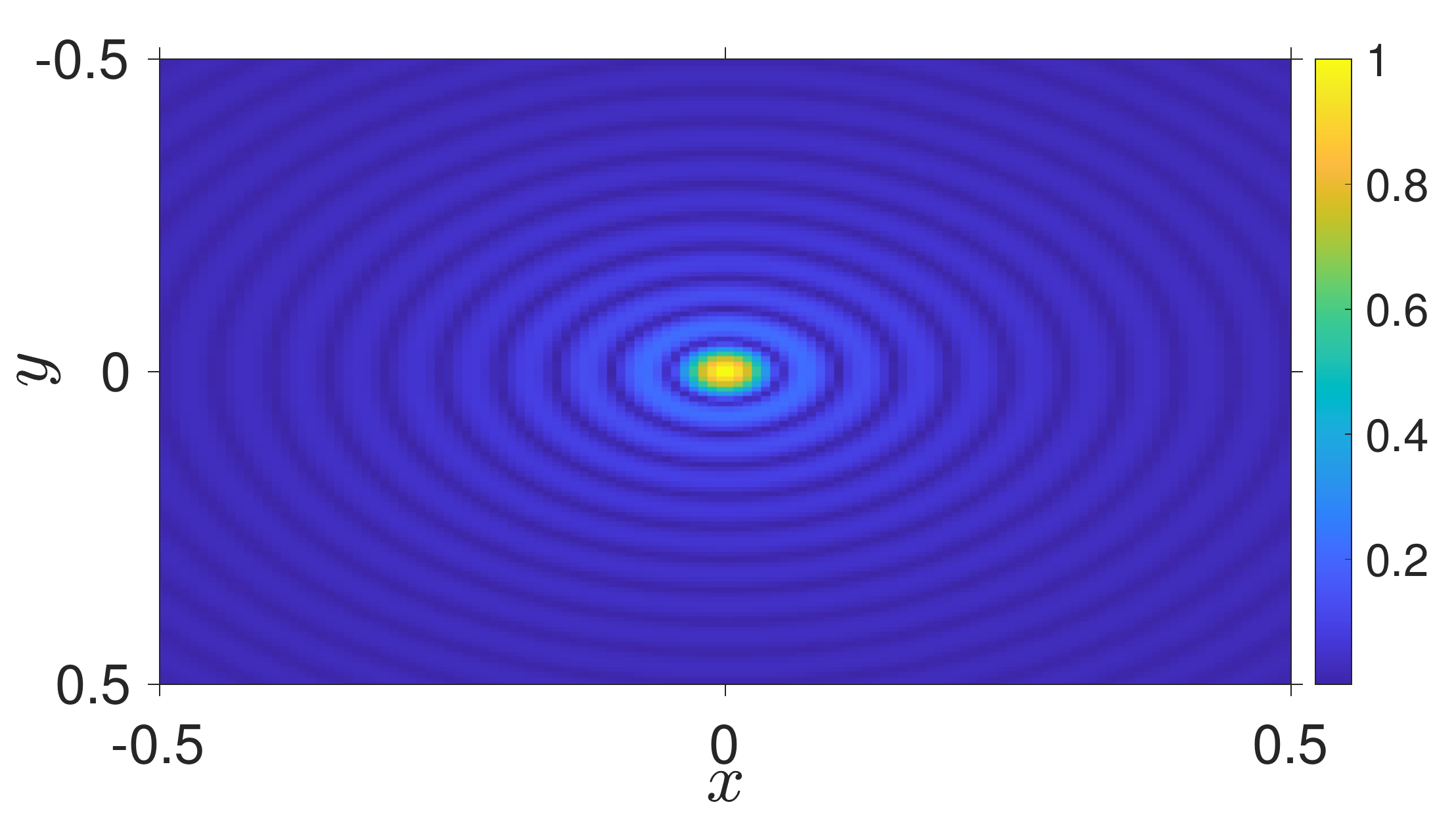}
  \caption{Isotropic}\label{fig:Corr_iso}
\end{subfigure}\hfill
\begin{subfigure}{0.33\textwidth}
\vspace{-.3cm}
  \includegraphics[width=\linewidth]{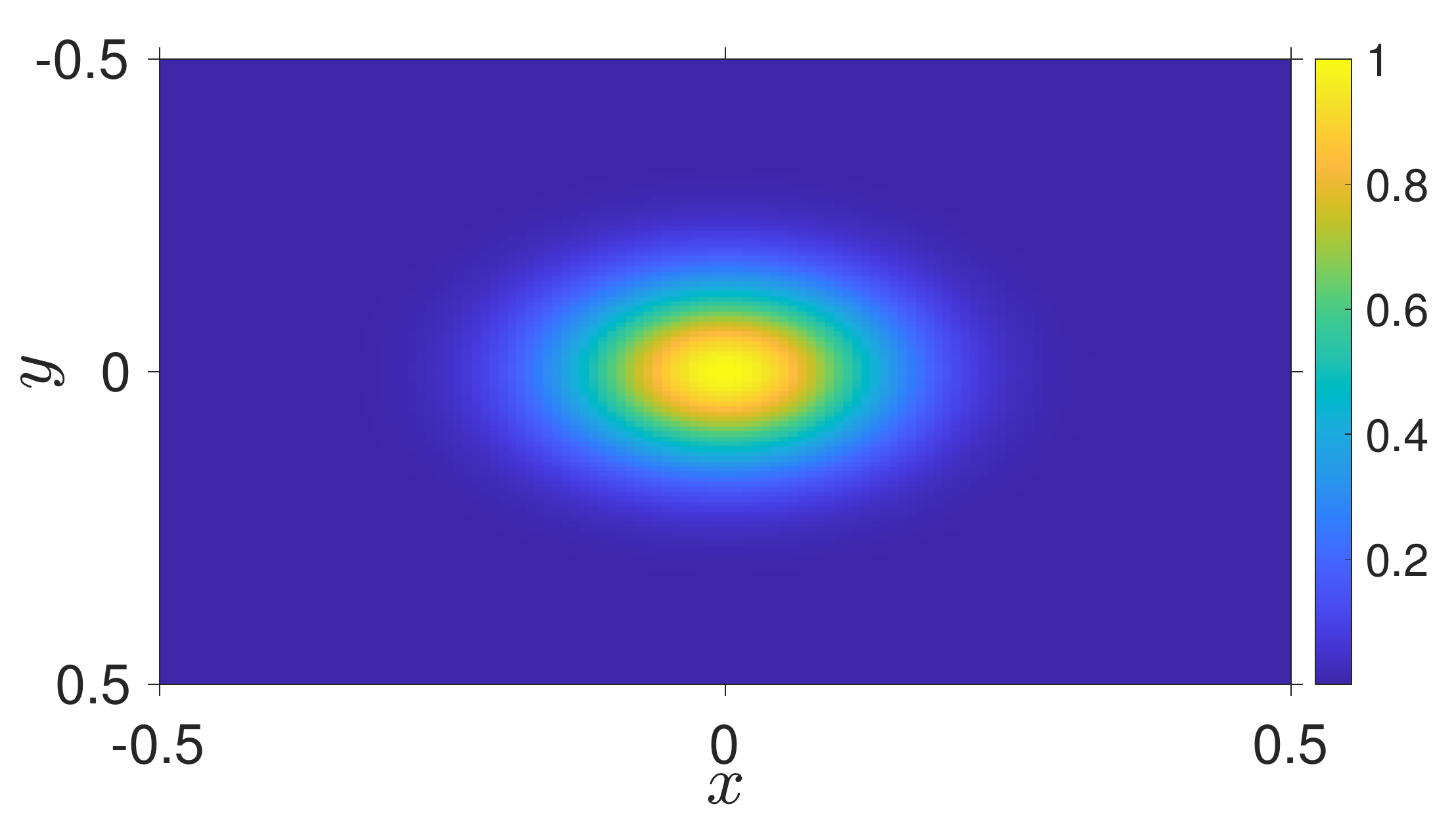}
  \caption{Non-isotropic}\label{fig:Corr_noniso}
\end{subfigure}\hfill
\begin{subfigure}{0.33\textwidth}%
\vspace{-.3cm}
  \includegraphics[width=\linewidth]{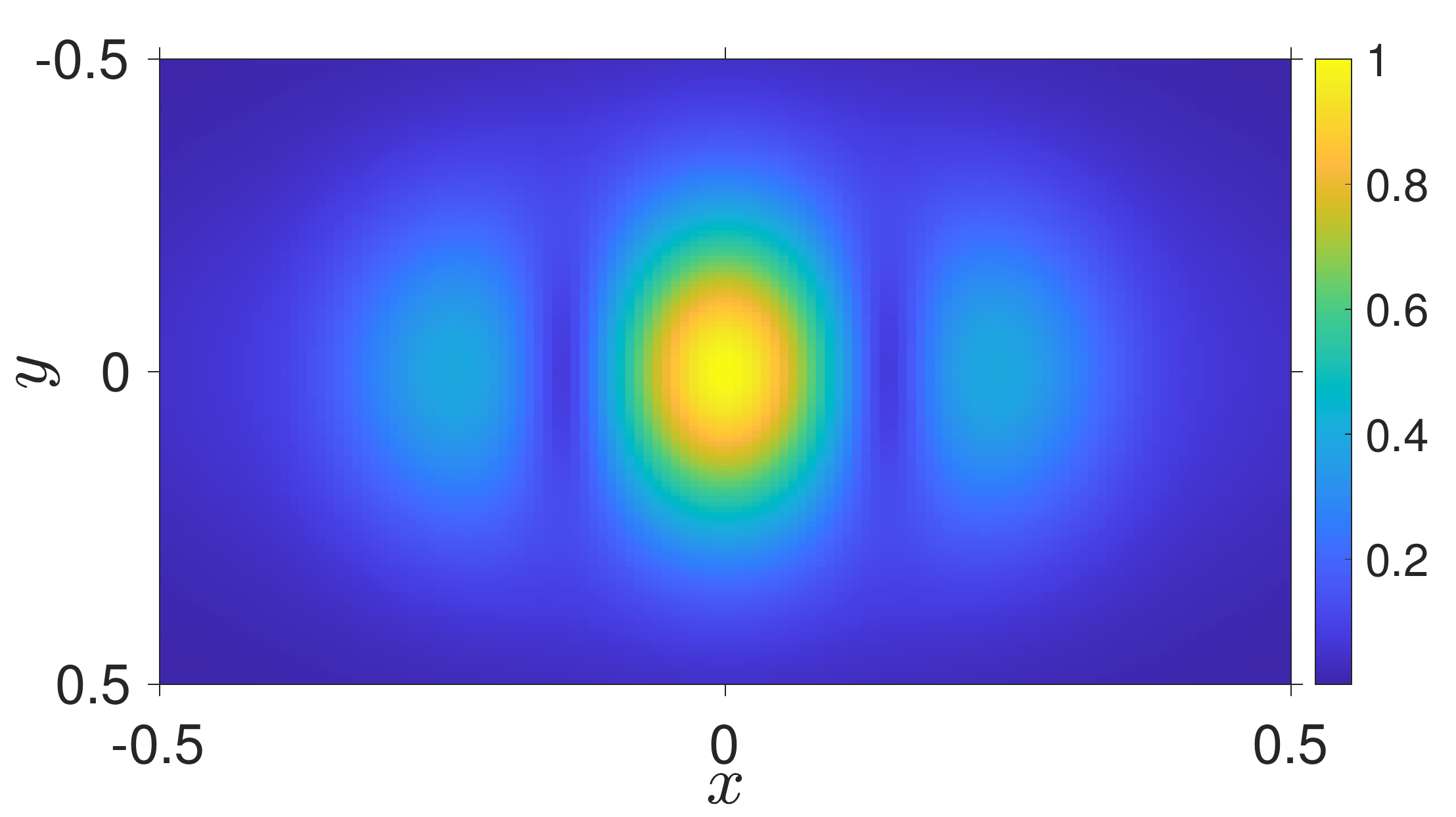}
  \caption{Non-isotropic mixture}\label{fig:Corr_noniso_cluster}
\end{subfigure}
\caption{{\small ACF $c(\vect{r})$ in~\eqref{2Dautocorr_spherical} under the scattering conditions of Fig.~\ref{fig:Spectrum}.}}
\label{fig:VMF} \vspace{-.3cm}
\end{figure*}

The autocorrelation function (ACF) $c(\vect{r}) = \Ex\{e(\vect{r}^\prime) e^*(\vect{r}^\prime+\vect{r})\}$ of $e(\vect{r})$ obeys the Helmholtz equation in~\eqref{Helmholtz} \cite[Eq.~(4)]{MarzettaISIT}: 
\begin{equation} \label{Helmholtz_corr}
\nabla^2 c(\vect{r}) + \kappa^2 c(\vect{r}) = 0.
\end{equation}
Here, the dependence on $z$ is removed as it is immaterial to the statistics of the field.
An analogue procedure developed in Section~\ref{sec:electromagnetic} for a deterministic $e(\vect{r})$ yields \cite{PizzoJSAC20,MarzettaISIT}:
 \begin{equation} \label{3Dautocorr_z}
%c(\vect{r},z) =  \int_{\Real^2} S(\vect{k})  e^{\imagunit k_z(\vect{k}) z} e^{\imagunit \vect{k}^{\Ttran} \vect{r}} \, d\vect{k}
c(\vect{r}) = \frac{1}{(2\pi)^2} \int_{\Real^2} S(\vect{k})  e^{\imagunit \vect{k}^{\Ttran} \vect{r}} \, d\vect{k}
\end{equation}
where the power spectral density (PSD) of $e(\vect{r})$ is defined as
\begin{equation} \label{psd}
S(\vect{k}) = \frac{A^2(\vect{k})}{k_z(\vect{k})} \mathbbm{1}_{\mathcal{D}}(\vect{k})
\end{equation}
with $A(\vect{k})$ being an arbitrary non-negative function, called \emph{spectral factor}.
Compared to~\eqref{Helmholtz_Fourier_general_solution}, the integration region is restricted to $\vect{k}\in\mathcal{D}$ as $k_z(\vect{k})$ must be real-valued in order for $c(\vect{r})$ to satisfy the standard Hermitian-symmetry property.

The random field $e(\vect{r})$ is obtained by expanding a white-noise complex field $W(\vect{k})$ with unit variance over the same 2D Fourier basis in~\eqref{3Dautocorr_z}, which yields the \emph{Fourier spectral representation} of $e(\vect{r})$ \cite[Sec.~III.C]{PizzoJSAC20}:
\begin{equation} \label{spectral_representation_z0}
%e(\vect{r}) = \frac{1}{2\pi} \int_{\Real^2}  {\color{blue}\underbrace{S^{1/2}(\vect{k}) W(\vect{k})}_{E(\vect{k})}} e^{\imagunit \vect{k}^{\Ttran} \vect{r}} \, d\vect{k}.
e(\vect{r}) = \frac{1}{2\pi} \int_{\Real^2}  {S^{1/2}(\vect{k}) W(\vect{k})} e^{\imagunit \vect{k}^{\Ttran} \vect{r}} \, d\vect{k}
\end{equation}
%\begin{align} \label{field_spectral}
%e(\vect{r},z) & =  \int_{\mathcal{D}} \frac{A(\vect{k}) W(\vect{k})}{\sqrt{k_z(\vect{k})}}  e^{\imagunit k_z(\vect{k}) z} e^{\imagunit \vect{k}^{\Ttran} \vect{r}} \, d\vect{k}. 
%\end{align}  
%If we remove the phase-shift -- the statistics of $W(\vect{k}) e^{\imagunit k_z(\vect{k}) z}$ do not change along $z$ for any $\vect{k}\in\mathcal{D}$ -- and omit the dependance on $z$ we obtain the 2D \emph{Fourier spectral representation} 
where equality must be understood in the MSE sense for any field with finite average  energy~\cite{PizzoJSAC20}
\begin{align} \label{power}
\sigma^2 & = \int_{\Real^2} \Ex\{|{e}(\vect{r})|^2\} \, d\vect{r} = \frac{1}{2\pi} \int_{\Real^2}  S(\vect{k}) \, d\vect{k}.
\end{align}
%with singularly-integrable power spectral density
%The field's power is given by
%\begin{equation} \label{power}
%\sigma^2 = c(\vect{0},0) = \frac{1}{(2\pi)^2} \int_{\Real^2}  S(\vect{k}) \, d\vect{k}.
%\end{equation}
%\begin{align} \label{2Dautocorr_psd}
%c(\vect{r}) & 
%%=  \int_{\mathcal{D}} \frac{A^2(\vect{k})}{k_z(\vect{k})}  e^{\imagunit \vect{k}^{\Ttran} \vect{r}} \, d\vect{k} \\ \label{2Dautocorr_psd}
% = \frac{1}{2\pi} \int_{\Real^2}  S(\vect{k}) e^{\imagunit \vect{k}^{\Ttran} \vect{r}} \, d\vect{k}.
%\end{align}
%
%By reintroducing the $z$-dependance in~\eqref{spectral_representation_z0}, we can interpret $e(\vect{r},z)$ as physically generated by an integral superposition of \emph{propagating} {plane-waves} impinging on the point $(\vect{r},z)$ for any $z\ge 0$, each one of which has circularly-symmetric complex-Gaussian amplitude that is statistically independent from one direction to another.
We can interpret~\eqref{spectral_representation_z0} as a 2D inverse Fourier transform of a wavenumber spectrum $E(\vect{k}) = {S^{1/2}(\vect{k}) W(\vect{k})}$ consisting of a continuum of circularly symmetric complex Gaussian and statistically independent coefficients~\cite{PizzoJSAC20,MarzettaISIT}.

\subsection{DoF of Electromagnetic Random Fields} \label{sec:DoF_statistic}

The computation of the DoF of a stationary random process requires to identify  a basis set of functions that yields uncorrelated coefficients -- also statistically independent for a jointly Gaussian process. This is at the basis of the Karhunen-Loeve series expansion \cite[Sec.~3.3.2]{VanTreesBook}.
Finding this basis set is hard in practice, as an explicit solution is only available for a few cases. As an example, for a stationary random process $e(t)$ with a constant and bandlimited spectrum, it can be found by solving the Slepian's concentration problem \cite{Landau} \cite[Sec.~2]{FranceschettiBook}. 
The same theory extends to a stationary random field $e(\vect{r})$ of constant and bandlimited spectrum~\cite{Franceschetti},\cite[Sec.~3]{FranceschettiBook}.

A way to overcome the difficulty in finding a Karhunen-Loeve series expansion is to operate in the \emph{asymptotic regime}.
Particularly, as the observation interval $T$ of a stationary process $e(t)$ becomes large, for a given bandwidth $\Omega$, i.e., $\Omega T \gg 1$, the Karhunen-Loeve series expansion tends to a Fourier series expansion with statistically independent coefficients, whose variances are obtained by sampling the PSD of $e(t)$ \cite[Sec.~3.4.6]{VanTreesBook}. 
The translation of this result to the spatial domain is ensured by a time-frequency and space-wavenumber duality \cite{PizzoIT21}\cite[Sec.~V]{PizzoJSAC20}.

As the size $L$ of a squared observation region increases, for a given wavelength $\lambda$, i.e., $L/\lambda \gg 1$, a 2D Fourier series with fundamental frequency of $2\pi/L$ rad/m may be defined, which enjoys the same statistical independence properties \cite{PizzoSPAWC20}. As $L/\lambda \to \infty$, series are replaced by integrals like the one in \eqref{spectral_representation_z0} with continuous spectrum $E(\vect{k})$ replacing Fourier coefficients.
%Its support determines the Nyquist sampling and field reconstruction through the spectral factor modeling the scattering selectivity. For example, under isotropic scattering, $E(\vect{k})$ has support given by $\mathcal{D}$.
The DoF are obtainable by counting the wavenumber lattice samples taken at $2\pi/L$ rad/m apart within the spectral support of $E(\vect{k})$. Under isotropic scattering, with spectral support given by the disk $\mathcal{D}$ of radius $\kappa=2\pi/\lambda$, this coincides to the Gauss circle problem \cite{Wolfram}.

%can be extended to every ensemble of a random electromagnetic field $e({\bf r})$ after replacing the field's spectrum with its power spectral density and must be interpreted on average. 

%\subsection{Reconstruction of a Random Electromagnetic Field}
%
%%Optimal reconstruction of $e(\vect{r})$ from its spatial samples can be regarded as a continuous-space estimation problem \cite{Agrell2004}. %under arbitrary scattering conditions
%
%Samples of the field are obtained by solving the Nyquist sampling problem, as done in Section~\ref{sec:Nyquist_sampling} for a deterministic field, while replacing the field's spectrum with its power spectral density.

\section{Numerical Analysis} \label{sec:numerical}
 Numerical results are provided next to validate the above results on Nyquist sampling and field's reconstruction. 
Simulations are carried out for a stationary electromagnetic random field under isotropic and non-isotropic scattering conditions, which are analytically modeled next. 
%the same scattering conditions illustrated in Fig.~\ref{fig:Spectrum}.
%[TWO SOURCES OF ERROR IN NON-ASYMPTOTIC REGIME: CORRELATED SAMPLES AND FINITE NUMBER OF SAMPLES. BOTH VANISHES AS L/LAMBDA GOES TO INF.]

\subsection{Isotropic and Non-isotropic Scattering}
In elevation and azimuth angles,~\eqref{3Dautocorr_z} can be rewritten as
\begin{align}  
c(\vect{r}) & \mathop{=}^{(a)}  \int_0^{\frac{\pi}{2}} \!\!\!\int_0^{2 \pi} \!\!\!A^2(\theta,\phi) e^{\imagunit \vect{k}^{\Ttran}(\theta,\phi) \vect{r}} \sin(\theta) d\theta d\phi  \\  \label{2Dautocorr_spherical}
& \mathop{=}^{(b)} \int_0^{\frac{\pi}{2}} \!\!\!\int_0^{2 \pi} \!\!\!A^2(\theta,\phi) e^{\imagunit \kappa \sin(\theta) (\cos(\phi) x + \sin(\phi) y)} \sin(\theta) d\theta d\phi 
\end{align}
where $(a)$ is due to a change of variables with Jacobian $d\vect{k}/k_z(\vect{k}) \propto \sin(\theta) d\theta d\phi$ and $(b)$ is obtained by substituting~\eqref{wavevector} \cite{PizzoIT21}. Here, constant terms due to change of variables are embedded into $A^2(\theta,\phi)$.

An isotropic scattering scenario is modeled as $A(\theta,\phi) = 1$ and plotted (inclusive of the Jacobian $\sin(\theta)$) in Fig.~\ref{fig:Sphe_iso} on the upper hemisphere $(\theta,\phi) \in [0,\pi/2] \times [0,2\pi)$.
In the wavenumber domain, we have that $A(\vect{k}) = 1$ so that the PSD in~\eqref{psd} is given by $S(\vect{k}) = \mathbbm{1}_{\mathcal{D}}(\vect{k})/k_z(\vect{k})$, as we expected from the deterministic case in Fig.~\ref{fig:circle_packing}. This is plotted in Fig.~\ref{fig:Wave_iso}.
%Here, $e(\vect{r})$ has a spectrum that is non-zero over $\vect{k}\in\mathcal{D}$, 
The isotropic ACF is shown in Fig.~\ref{fig:Corr_iso}. This is computed from~\eqref{2Dautocorr_spherical} by setting $x=\|\vect{r}\|$ and $y=0$ (due to polar symmetry), 
\begin{align} \label{clarke_acf}
c(\vect{r})
%& = \iint_{[0,\pi/2] \times [0,2\pi)} \hspace{-1.6cm} e^{\imagunit \kappa \|\vect{r}\| \sin(\theta) \cos(\phi)} \sin(\theta) d\theta d\phi 
= \sinc(2 {\|\vect{r}\|}/{\lambda})
\end{align}
and coincides to the Clarke's model~\cite[Sec.~2.4]{PaulrajBook}.
%\footnote{A different proof of~\eqref{clarke_acf} is available in \cite[App.~III]{PizzoJSAC20} that leads to the same result. }

Non-isotropic scattering can be considered alike by specifying a non-uniform spectral factor $A(\theta,\phi)$.
In this case, $A(\vect{k})$ (thus $S(\vect{k})$) will have arbitrary shape, as first assumed in the deterministic settings in Fig.~\ref{fig:non_isotropic_spectrum}. A clustered propagation environment is modeled as a mixture of spectral factors \cite{PizzoIT21}
\begin{equation} \label{mixture_pdf}
A^2(\theta,\phi) = \sum_{i=1}^{N_{\rm c}}  w_i \, A^2_{i}(\theta,\phi)
\end{equation}
weighted by $w_i\ge 0$ whose sums yields $1$. Here, $A^2_{i}(\theta,\phi)$ is representative of the $i$-th cluster. 
%of 3D von Mises-Fisher (vMF) family of density distributions 
%\begin{equation} \label{VMF_pdf_2}
%A^2_{i}(\theta,\phi) = c(\alpha_{i}) e^{\alpha_{i} \big(\sin(\theta) \sin({\theta_{r,i}}) \cos(\phi-{\phi_{r,i}}) + \cos(\theta) \cos({\theta_{r,i}})\big)}
%\end{equation}
%with normalization constant $c(\alpha_{i}) = {\alpha_{i}}/{(4\pi \sinh(\alpha_{i}))}$ 
A simple way to model $A^2_{i}(\theta,\phi)$ is by using the 3D von Mises-Fisher distribution \cite{PizzoIT21}, parametrized by a modal direction 
\begin{equation} \label{modal_direction}
\hat{\boldsymbol{\xi}}_{i} = 
\begin{pmatrix}
\sin({\theta_{r,i}}) \cos({\phi_{r,i}}) \\
\sin({\theta_{r,i}}) \sin({\phi_{r,i}}) \\
\cos({\theta_{r,i}})
\end{pmatrix}
\end{equation}
and an angular concentration ${\alpha_{i} \in [0,\infty)}$ describing the power spread around $\hat{\boldsymbol{\xi}}_{i}$. Clearly, isotropic scattering in Fig.~\ref{fig:Sphe_iso} is obtained from~\eqref{mixture_pdf} by setting ${N_{\rm c}=1}$ and ${\alpha=0}$.
Non-isotropic scattering is considered in Fig.~\ref{fig:Sphe_noniso} and Fig.~\ref{fig:Sphe_noniso_cluster}. 
In these cases, the support of $\mathcal{K}$ is measured by considering the bandwidth at $-20$~dB. We will see that this criterion provides an accurate measure of the field's spectral support.

 \begin{figure*}[th!] 
 \begin{subfigure}{0.33\textwidth}
 \hspace{-.31cm}
  \includegraphics[width=\linewidth]{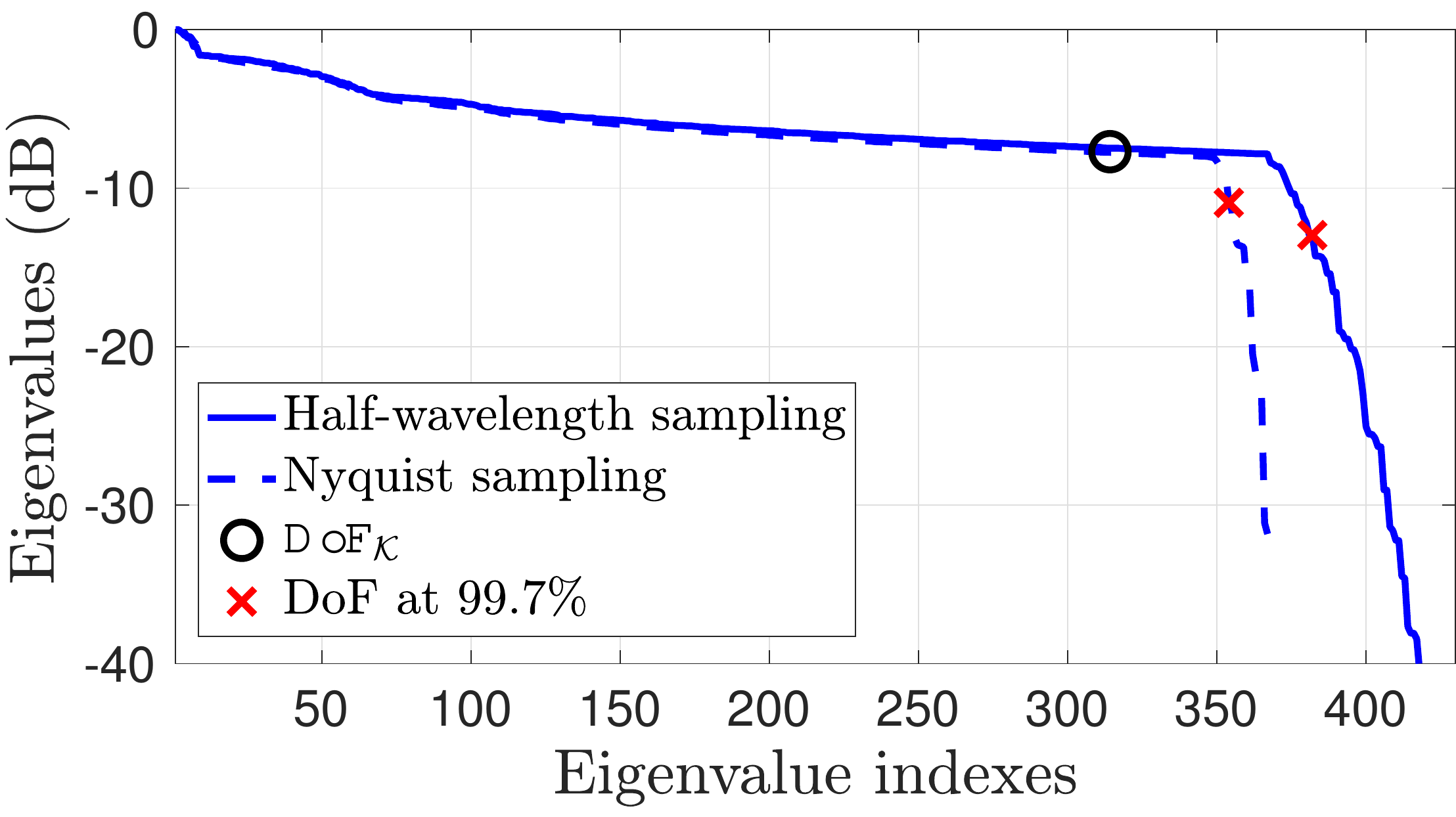}
  \caption{Isotropic}\label{fig:eigenvalues_Iso}
\end{subfigure} 
\begin{subfigure}{0.33\textwidth}
\hspace{-.25cm}
  \includegraphics[width=\linewidth]{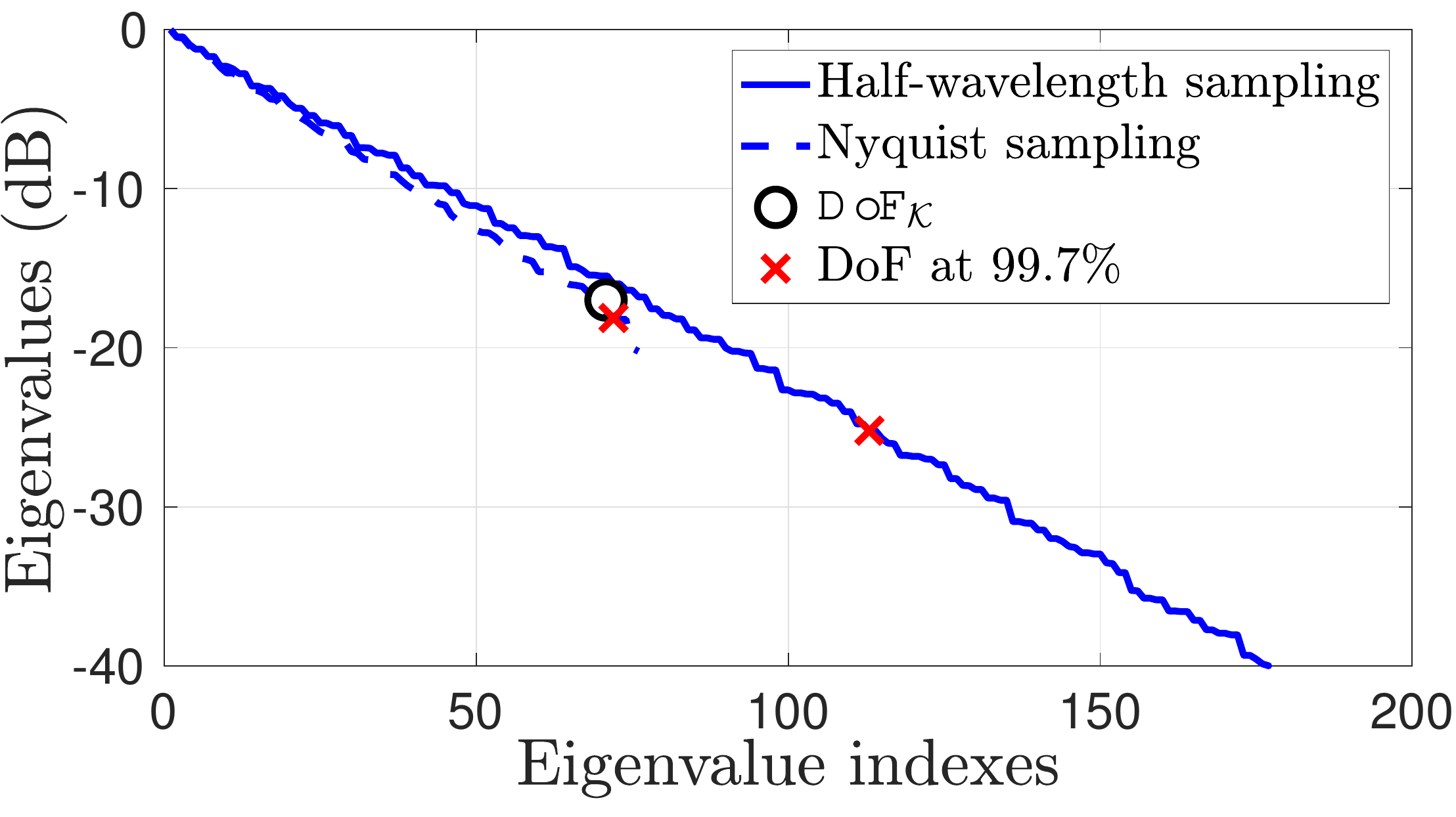}
  \caption{Non-isotropic}\label{fig:eigenvalues_NIsoCPL}
\end{subfigure}
\begin{subfigure}{0.33\textwidth}%
\hspace{-.2cm}
  \includegraphics[width=\linewidth]{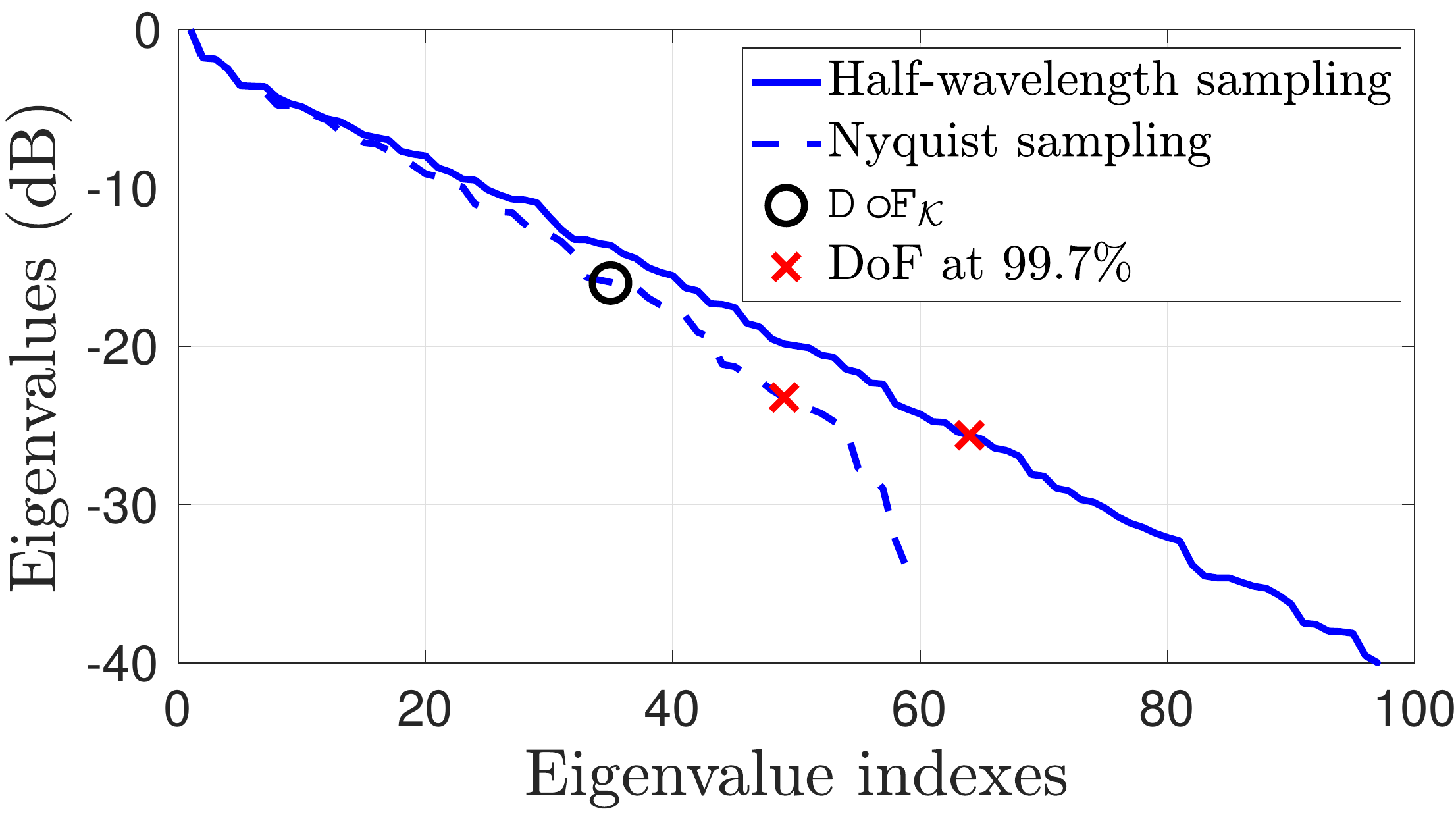}
  \caption{Non-isotropic mixture}\label{fig:eigenvalues_NIsoCluster}
\end{subfigure}
\caption{
{\small Eigenvalues of $\vect{C}$ (in dB) reported in a descending order under the scattering conditions depicted in Fig.~\ref{fig:Spectrum}. Nyquist samples of $e(\vect{r})$ are collected within a squared region $\mathcal{A}$ of side length $L = 10 \lambda$.}}
%The elongated hexagonal sampling is compared to classical $\lambda/2$-rectangular sampling with sampling matrices $\vect{Q}_{\mathcal{E}}^\star$ in~\eqref{sampling_matrix_ellipse} and $\vect{Q}_{\mathcal{R}}^\star$ in~\eqref{rectangular_sampling}, respectively.}%with centered single cluster, and non-isotropic mixture with two clusters}
\label{fig:eigenvalues}
\end{figure*} 

In Fig.~\ref{fig:Sphe_noniso}, we have a single cluster with ${{\theta_r} ={\phi_r} = 0^\circ}$, and $\alpha = 40$.
% such that the normalized circular variance is ${\nu^2=0.05}$. 
The wavenumber domain is showed in Fig.~\ref{fig:Wave_noniso} with circular support $\mathcal{K}$ of approximately $0.47 \kappa$ radius, i.e., ${\sqrt{d_1}=\sqrt{d_2} = 0.47}$, and $k_\phi=0$; see Fig.~\ref{fig:non_isotropic_spectrum}. 
The ACF is reported in~Fig.~\ref{fig:Corr_noniso} and shows a larger coherence area with respect to the isotropic case, due to a higher selectivity of the scattering. 

In general, $\mathcal{K}$ has arbitrary shape -- not necessarily circular -- as first assumed in the deterministic settings in Fig.~\ref{fig:non_isotropic_spectrum}. This is shown in Fig.~\ref{fig:Sphe_noniso_cluster} and Fig.~\ref{fig:Wave_noniso_cluster} in both representation domains for a two-clustered scenario with ${w_1=w_2=1/2}$, ${\theta_{r,1} = 0^\circ}$, ${\theta_{r,2} = 10^\circ}$, ${\phi_{r,1} = 180^\circ}$, ${\phi_{r,2} = 0^\circ}$, $\alpha_1 = 200$ and $\alpha_2 = 100$.
%${\nu^2 = \{0.01, 0.02\}}$ 
We resort to an elliptical embedding with $\sqrt{d_1}=0.5$, $\sqrt{d_2}=0.35$, and $k_\phi=0$. 
The corresponding ACF is given in Fig.~\ref{fig:Corr_noniso_cluster} and shows the impact of asymmetries in the PSD of $e(\vect{r})$.

%The simplest choice is obtained by modeling $A(\theta,\phi)$ as a bounded piecewise constant function over non-overlapped angular sets $\Theta_i$, i.e., $A(\theta,\phi)=\mathbbm{1}_{\Theta}(\theta,\phi)$, with $\Theta = \cup_i \Theta_i$ \cite{PoonDoF}. An $A(\theta,\phi)$ of this sort physically corresponds to having a clustered scattering environment that simply cuts-off some propagation directions without altering their amplitudes. A few examples are given in Fig.~\ref{fig:Spectrum}.
%In particular, in Fig.~\ref{fig:Sphe_noniso} (or Fig.~\ref{fig:Wave_noniso} in the wavenumber domain) a non-isotropic scattering condition with centered single cluster is illustrated with $\Theta = \Theta_1 = \{(\theta,\phi) \in[0,xxx] \times [0,\pi/2]\}$.
%Hence, it is simply obtained by limiting the isotropic spectral factor in elevation angle, which means a smaller angle spread with respect to the isotropic case.
%This reflects to the autocorrelation function $c(\vect{r})$ in~Fig.~\ref{fig:Corr_noniso} that decays to zero less rapidly due to a larger coherence distance, as expected. 
%As seen, the spherical coordinates provide an intuitive physical understanding of the wave propagation phenomena angularly. Despite this useful property, for sampling purposes, we next pursue our analysis in the wavenumber domain instead. This will allow us to exploit the connection to the time-domain and Fourier theory.

%\subsection{Non-Asymptotic Regime and Implementation }
\subsection{Non-Asymptotic Regime}

%Bandlimited fields must be defined everywhere, as dictated by the uncertainty principle \cite[Sec.~2.3]{FranceschettiBook}. 
Electromagnetic fields are physically observable on a spatial region $\mathcal{A} \subset \Real^2$ of finite size, from which only a finite number $N$ of samples is available for field's reconstruction, and given by
\begin{equation} \label{N_antennas}
N = \lceil \dof \rceil = m(\Lambda(\mathcal{A})) % = \lfloor \mu \; m(\mathcal{A})\rfloor
\end{equation}
where $\Lambda(\mathcal{A}) \subseteq \Integer^2$ denotes the truncated 2D lattice used for Nyquist sampling. The reconstruction formula obtained by truncating the 2D cardinal series expansion in~\eqref{cardinal_series_spatial} up to $N$ terms yields
\begin{equation} \label{reconstruction_formula}
\hat{e}(\vect{r}) = \sum_{\vect{n} \in \Lambda(\mathcal{A})} e(\vect{Q}^\star \vect{n}) f_{\mathcal{K}}(\vect{r} - \vect{Q}^\star \vect{n})
\end{equation}
to which is associated a pointwise mean-squared error (MSE)
\begin{equation} \label{error_mse}
\sigma^2_N(\vect{r}) = \Ex\{|{e}(\vect{r}) - \hat{e}(\vect{r})|^2 \}
\end{equation} 
where the expectation is taken with respect to all possible configurations of scatterers.
The error magnitude is inevitably higher at the boundary region of $\mathcal{A}$, due to the truncation and non-orthogonality of the interpolating set of functions in~\eqref{reconstruction_formula} over $\vect{r}\in\mathcal{A}$.
%Errors are due to the non-orthogonality of the interpolating functions in (67) over r ∈ A. These concentrate at the boundaries of A as a consequence of the truncation operation that cuts off the tails of nearby interpolating functions.
%an error field}
%\begin{equation} \label{error}
%e_N(\vect{r}) = {e}(\vect{r}) - \hat{e}(\vect{r})
%\end{equation} 
The accuracy of the reconstruction procedure in a non-asymptotic regime can be measured by the average energy of the error per unit of area \cite{FranceschettiBook}
%\begin{equation} \label{MSE}
%E_N =  \frac{1}{m(\mathcal{A})} \int_{\mathcal{A}} (e_N(\vect{r}))^2 \, d\vect{r}.
%\end{equation}
\begin{equation} \label{mse_average}
\sigma^2_N = \frac{1}{m(\mathcal{A})} \int_{\mathcal{A}} \sigma^2_N(\vect{r}) \, d\vect{r}
\end{equation}
which tends to zero asymptotically as $N\to\infty$ (i.e., $\mathcal{A} = \Real^2$), when \eqref{reconstruction_formula} reduces to \eqref{cardinal_series_spatial}.
Remarkably, under the Nyquist condition, the reconstruction formula~\eqref{reconstruction_formula} with the interpolating function~\eqref{interp_function} provides us with the best linear interpolator of $e(\vect{r})$ in the MSE sense \cite[Sec.~VI]{PETERSEN1962279}\cite[Sec.~III]{Agrell2004}. 
%Other sources of error are introduced by constraints imposed by hardware implementation of the sampling operation.
%Hence, fields are subjected to a low-pass filtering operation prior to sampling that may introduce aliasing, given the impossibility of implementing a filter with infinite cutoff transitions \cite{Unser}. 

Also noteworthy is that the sampling operation of electromagnetic fields undergoes the same practical constraints imposed by hardware implementation issues (as for time-domain signals). 
For example, analog-to-digital converters can only provide a finite precision description of analog samples measured at every points, which inevitably introduces a quantization loss \cite[Sec.~4.8]{OppenheimBook}. Nevertheless, one way to compensate for these losses is to trade-off some efficiency by sampling above Nyquist density \cite{Kumar2011}.

\subsection{Sampling and Reconstruction in Non-Asymptotic Regime}
%with Finite Number of Samples}

\begin{figure} [t!]\vspace{-0.5cm}
        \centering
         \begin{subfigure}[t]{\columnwidth} \centering  
        	\begin{overpic}[width=.999\columnwidth,tics=10]{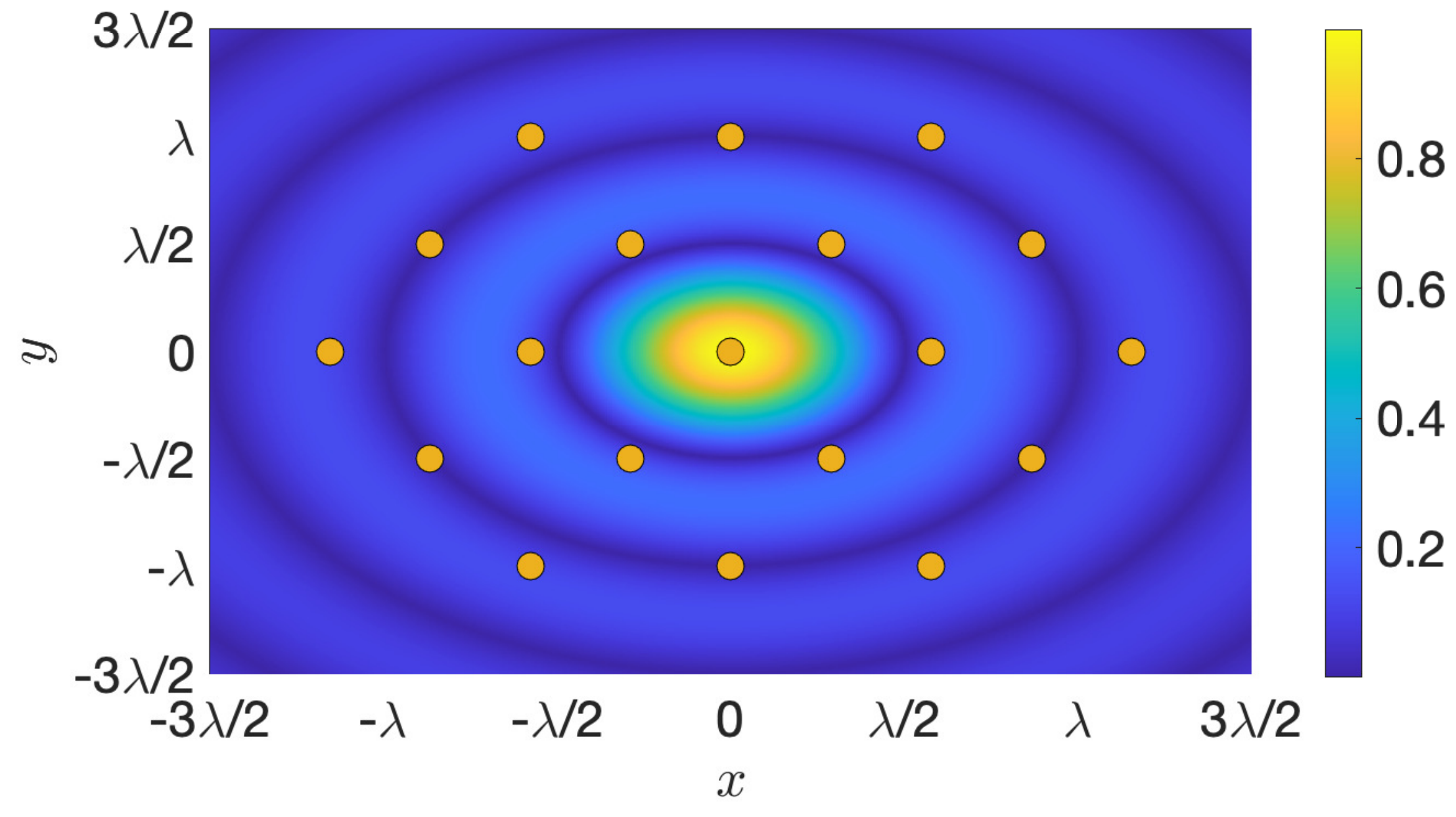}
\end{overpic}  \vspace{-0.2cm}
                \caption{Isotropic scattering.}  \vspace{0.0cm}
                \label{fig:sinc_sampling}
        \end{subfigure}
         \begin{subfigure}[t]{\columnwidth} \centering  
	\begin{overpic}[width=.999\columnwidth,tics=10]{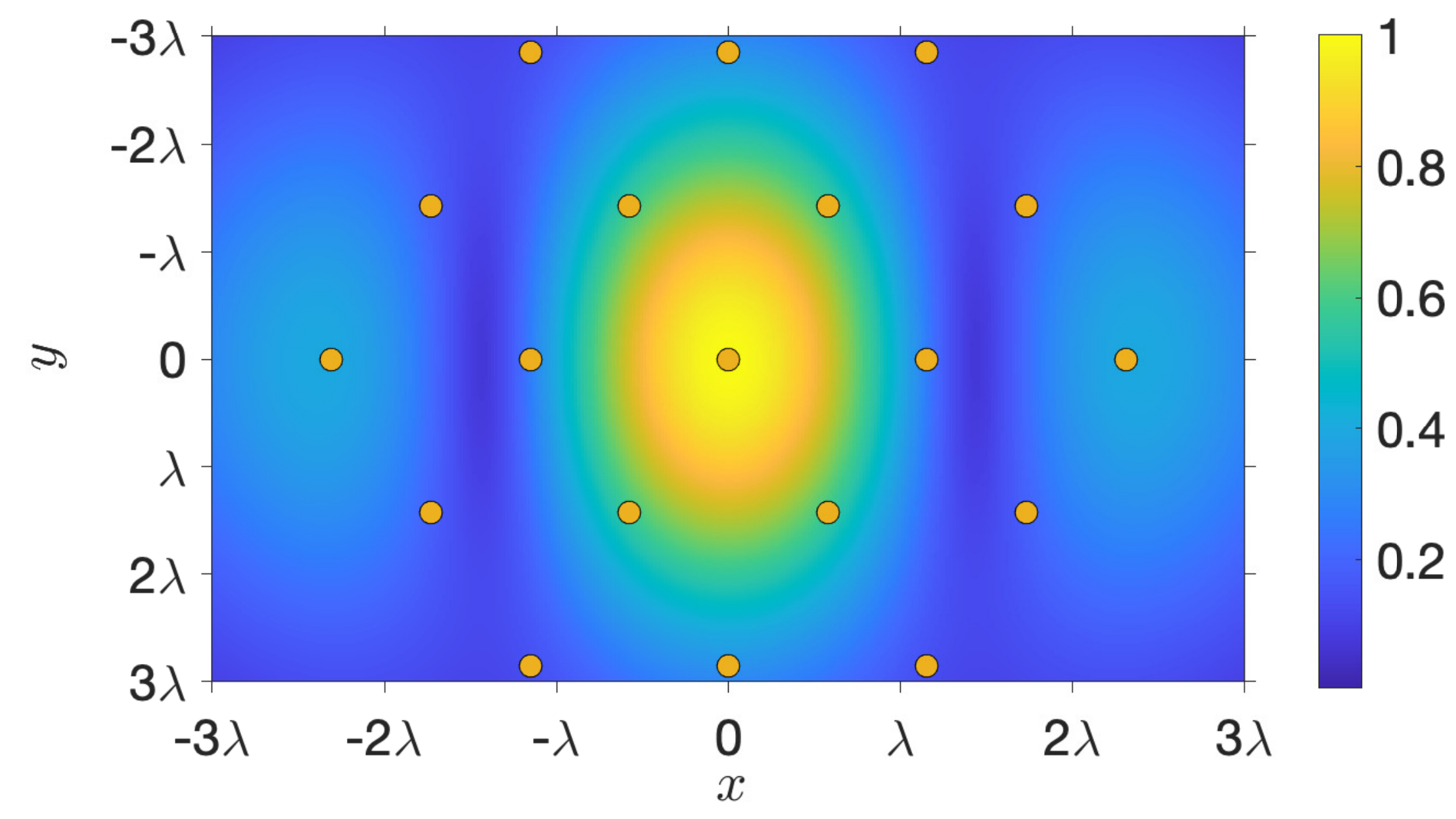}
\end{overpic}  \vspace{-0.2cm}
                \caption{Non-Isotropic scattering} 
                \label{fig:noniso_sampling}  
        \end{subfigure}   
        \caption{{\small Nyquist sampling plotted against the ACF $c(\vect{r})$.}}
        \label{fig:sampling_ACF}\vspace{-0.6cm}
\end{figure}

To quantify the impact of Nyquist sampling and the number of DoF, we consider the eigenvalues of the spatial autocorrelation matrix ${\vect{C} \in \Complex^{N \times N}}$, obtained under the scattering conditions depicted in Fig.~\ref{fig:Spectrum} and a squared region $\mathcal{A}$ of side length ${L = 10 \lambda}$. 
% (i.e., ${m(\mathcal{A}) = 100 \, \lambda^2}$~m$^2$). 
In the isotropic scenario, $\vect{C}$ is obtained by hexagonal sampling of~\eqref{clarke_acf} according to~\eqref{hexagonaLampling_2}. For any non-isotropic scenario, we sample~\eqref{2Dautocorr_spherical} by using the elongated hexagonal sampling in~\eqref{sampling_matrix_ellipse} with corresponding parameters. Comparisons are made with the (redundant) half-wavelength sampling in~\eqref{rectangular_sampling}. The number of DoF is computed on the basis of the results in Section~\ref{sec:DoF} and is numerically compared to the number of samples sufficient to capture $99.7\%$ of the total field's power. This will be slightly higher because of the finite value of $L/\lambda$. As $L/\lambda$ increases, the eigenvalue curves assume a step-like function behavior with transition corresponding to the number of DoF approximately.

%Notice also that in the system design, the threshold on how many eigenvalues one should consider depends on the signal-to-noise ratio and targeted reconstruction accuracy.
%The isotropic case is studied in Fig.~\ref{fig:eigenvalues_Iso}, where the optimal (Nyquist rate) hexagonal sampling, with sampling matrix $\vect{Q}_{\mathcal{D}}^\star$ in~\eqref{hexagonaLampling_2}, is compared to the sub-optimal (oversampling) $\lambda/2$-rectangular sampling, with $\vect{Q}_{\mathcal{R}}^\star$ in~\eqref{rectangular_sampling}.
%{\color{red}The number of samples follows by using~\eqref{density_sampling_hex} and~\eqref{density_sampling_rect} into~\eqref{N_antennas} as $N = 346$ and $N=400$, respectively.}

In Fig.~\ref{fig:eigenvalues_Iso}, the isotropic case is studied.
From~\eqref{DoF_D}, we obtain $\dof_{\mathcal{D}} = 315$ while $354$ eigenvalues are obtained with the $99.7\%$ criterion. The number of DoF at $99.7\%$ with $\lambda/2$ sampling increases up to $382$. This is not due to oversampling but to the larger area covered with the rectangular sampling, compared to hexagonal sampling, since it is  impossible to fit the hexagonal and rectangular grids within the same area.\footnote{This effect is even more pronounced in a non-isotropic scenario, due the enhanced sparsity of the elongated hexagonal sampling.} 
%\textcolor{red}{The number of observable samples of the $\lambda/2$ sampling is higher than the one of hexagonal sampling while providing the same (?) significant number of eigenvalues, and hence, it is less efficient.} 
The non-isotropic cases are considered in Fig.~\ref{fig:eigenvalues_NIsoCPL} and Fig.~\ref{fig:eigenvalues_NIsoCluster}. 
From~\eqref{DoF_K}, the number of DoF reduces to ${\dof_{\mathcal{K}} = 71}$ and ${\dof_{\mathcal{K}} = 35}$, which is due to the scattering selectivity that cuts off some of the resolvable propagation directions. Compared to Fig.~\ref{fig:eigenvalues_Iso}, the eigenvalue curves decay rapidly due to the spatial correlation; see also Fig.~\ref{fig:VMF}. Differences between the two curves are again due to numerical comparison issues.
%\textcolor{red}{The Nyquist sampling exhibits a slightly higher correlation than rectangular sampling, which is due to the numerical approximation of the field's spectral support with the $-20$~dB criterion.}[NON USIAMO LO STESSO CRITERIO A -20 DB PER ENTRAMBI I SAMPLINGS? Yes, infatti non sto capendo tanto quella frase in rosso.]

From  Fig.~\ref{fig:eigenvalues_Iso}, it follows that spatial correlation exists even under isotropic propagation (see also \cite[Sec.~IV]{PizzoTWC21}).
This is because $c(\vect{r})$ is not sampled at its zeros, as illustrated in Fig.~\ref{fig:sinc_sampling}. 
The hexagonal sampling provides the best trade off between collecting uncorrelated samples and minimizing the search area. Hence, it is the one that exploits the available number of DoF per unit of area. The non-isotropic case is shown in Fig.~\ref{fig:noniso_sampling} for completeness. As seen, to the higher correlation in Fig.~\ref{fig:eigenvalues_NIsoCluster}, compared to the isotropic case, corresponds an increased difficulty in achieving the above trade off due to a larger coherence area of the ACF.

A numerical example is provided next to verify the reconstruction formula in~\eqref{reconstruction_formula} with interpolating function in~\eqref{interp_function}, under the scattering conditions of Fig.~\ref{fig:Sphe_noniso} (Fig.~\ref{fig:Wave_noniso}). Samples are obtained within a squared region $\mathcal{A}$ of finite size. 
In Fig.~\ref{fig:Recon_comparison_20lambda}, we gather measurements from a square of dimension ${L/\lambda = 40}$ to recover the real part of $e(\vect{r})$ within a segment of length $L/\lambda = 6$ along the $x$-axis. Nyquist sampling~\eqref{sampling_matrix_ellipse} with ${\sqrt{d_1}=\sqrt{d_2} = 0.47}$, and ${k_\phi=0}$ is compared to half-wavelength rectangular sampling. Sampling at Nyquist density achieves similar performance while saving $60\%$ of samples/m$^2$. Both have a non-zero error due to the finite number of samples exploited for reconstruction. 
Nevertheless, the error is small in both cases as only a tiny contribution to reconstruction is added by the tails of the interpolating functions.

The normalized MSE $\sigma^2_N/\sigma^2$, defined as the ratio between \eqref{error_mse} and \eqref{power}, is plotted in Fig.~\ref{fig:Normalized_MSE} (in dB) as a function of $L/\lambda \in [2,20]$.
Nyquist sampling is compared to rectangular sampling with spacing $\sqrt{d_1} \lambda/2$; that corresponds to a minimal squared embedding of $\mathcal{K}$ in Fig.~\ref{fig:Wave_noniso}. Both are plotted against hexagonal and half-wavelength samplings.
%The elongated hexagonal sampling is compared to rectangular sampling. For each case, we consider the minimal (corresponding to Nyquist sampling in the hexagonal case) and maximal embeddings (corresponding to isotropic propagation).
Oversampling provides a better reconstruction in the finite $L/\lambda$ regime. 
%{\color{red}In practice, measurements are corrupted by some noise. Hence, the number of DoF provides to a better estimate of the sampling density, above which samples are below the noise level and add little information to reconstruction.}
%{\color{red}However, this benefit reduces with noisy measurements {\bf(?? Non abbiamo mai parlato di rumore, ed ora in due parole si tira fuori)}, due to some samples being below the noise level.} 
All curves monotonically decrease with $L/\lambda$, due to the larger dataset available, and merge at infinity where the MSE is zero.
%{\color{red} Togliamo la Fig. 10 e si chiude qui; la mettiamo in fase di revisione!}We now evaluate the impact of the Nyquist sampling rate on the reconstruction capabilities of~\eqref{reconstruction_formula}.
%To this end, we define the average MSE  over $\vect{r}\in \mathcal{A}$ as~\cite[Sec.~III]{Agrell2004}
%\begin{equation}
%\sigma_{\rm e}^2(\Lambda(\mathcal{A})) = \frac{1}{m(\mathcal{A})} \int_{\mathcal{A}} \Ex\{|\hat{e}(\vect{r}) - {e}(\vect{r})|^2 \}  \, d\vect{r}.
%\end{equation}
%In Fig.~\ref{fig:Normalized_MSE}, the normalized MSE $\sigma_{\rm e}^2(\Lambda(\mathcal{A}))/\sigma^2$ (normalized with respect to the field's power) is plotted as a function of $L/\lambda \in[5,150]$. To avoid reconstruction errors that accumulate at the border of $\mathcal{A}$ due to truncation of the lattice $\Lambda(\mathcal{A})$, we evaluate the error on the $90\%$ of the total area of $\mathcal{A}$.
%As observed, the normalized MSE decays rapidly as the number of samples used for reconstruction increases and tends to zero asymptotically as $L/\lambda \to \infty$.

\begin{figure} [t!]\vspace{-0.5cm}
        \centering
	\begin{overpic}[width=.999\columnwidth,tics=10]{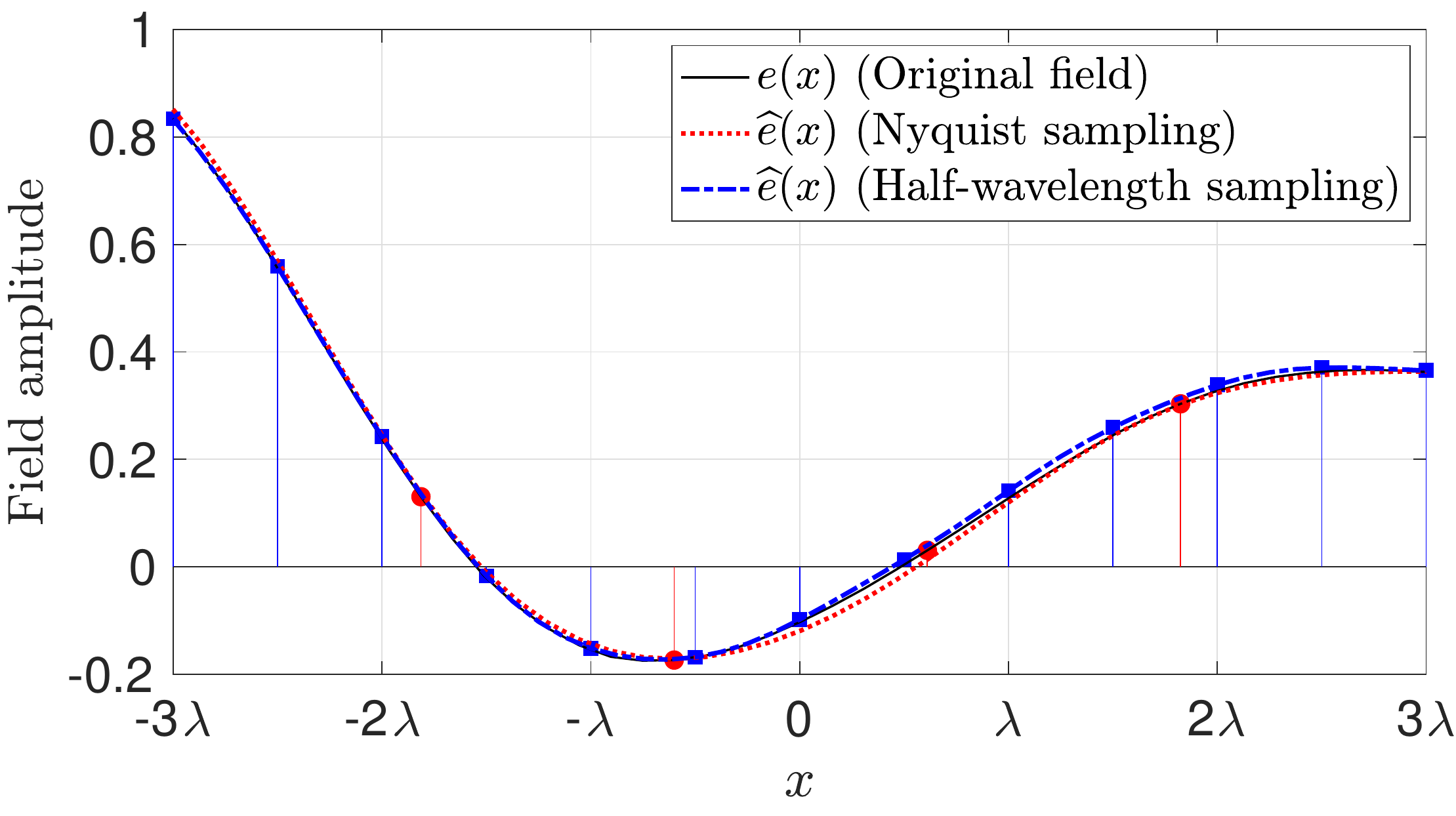}
\end{overpic} \vspace{-0cm}
                \caption{{\small Reconstruction of the real part of $e(\vect{r})$ over a segment of length $L/\lambda = 6$, under the scattering conditions of Fig.~\ref{fig:Sphe_noniso}.}} \vspace{0.5cm}
                \label{fig:Recon_comparison_20lambda} 
\end{figure}

\begin{figure} [t!]\vspace{-0.7cm}
        \centering
	\begin{overpic}[width=.999\columnwidth,tics=10]{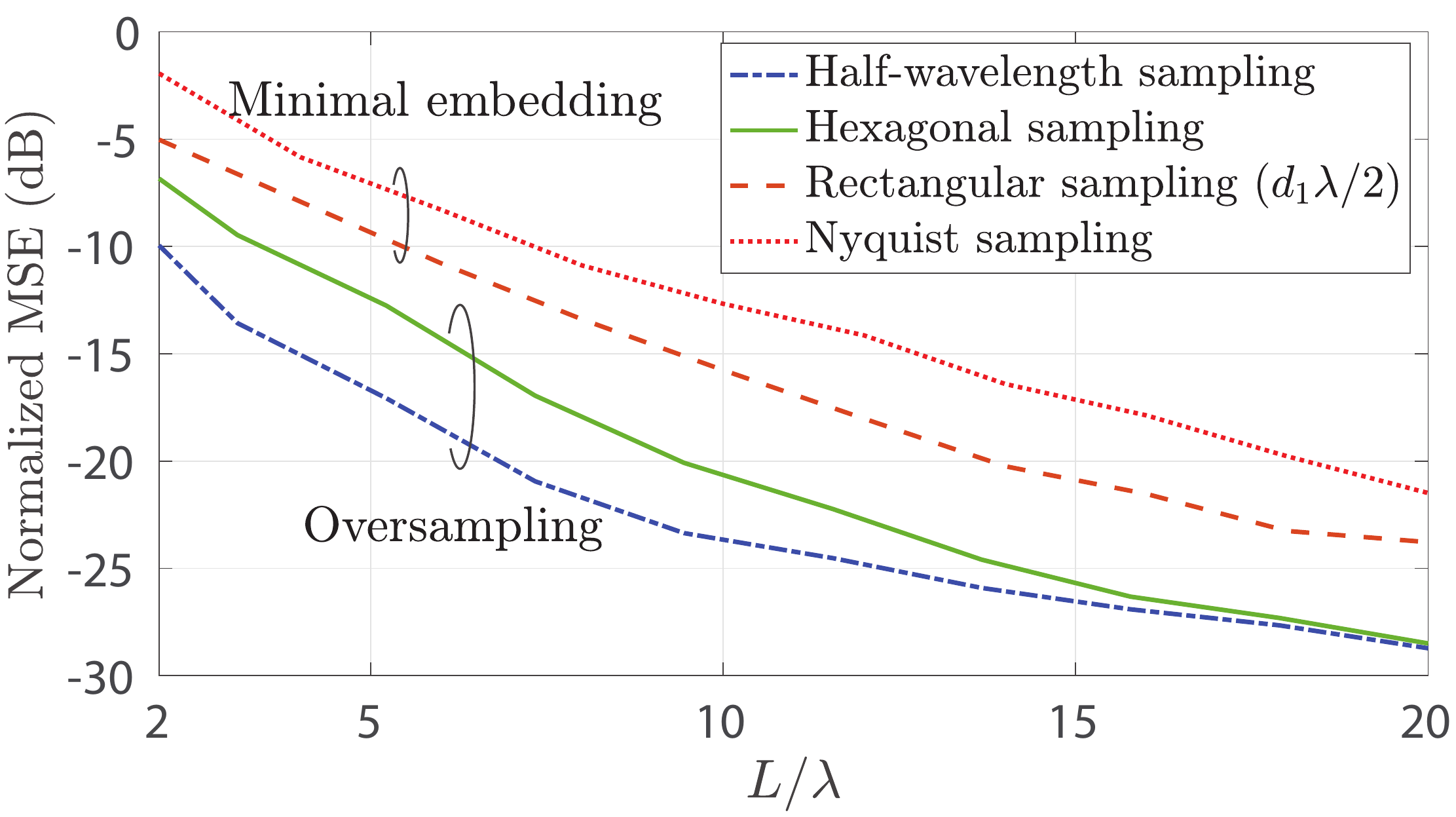}
\end{overpic} 
                \caption{{\small Normalized MSE as a function of $L/\lambda \in [2,20]$.}} \vspace{-0.5cm}
                \label{fig:Normalized_MSE} 
\end{figure}

%\begin{figure} [t!]
%        \centering
%        \begin{subfigure}[t]{\columnwidth} \centering 
%	\begin{overpic}[width=0.9\columnwidth,tics=10]{}
%\end{overpic} \vspace{-0.0cm}
%                \caption{$L/\lambda = 20$.} 
%                \label{fig:Recon_comparison_20lambda} 
%        \end{subfigure}     
%                \begin{subfigure}[t]{\columnwidth} \centering  
%        	\begin{overpic}[width=0.9\columnwidth,tics=10]{Recon_comparison_150lambda}
%\end{overpic}  \vspace{-0.0cm}
%                \caption{$L/\lambda = 150$.} 
%                \label{fig:Recon_comparison_150lambda}
%        \end{subfigure}\vspace{0.0cm}
%        \caption{Reconstruction of the real part of $e(\vect{r})$ from its Nyquist samples under the scattering conditions of Fig.~\ref{fig:Sphe_noniso}.}
%        \label{fig:reconstruction}\vspace{-0.0cm}
%\end{figure}

\section{Conclusions} \label{sec:conclusions}

We applied signal processing tools such as the multidimensional sampling theorem and Fourier theory to study the Nyquist sampling and number of DoF of an electromagnetic field, under deterministic and stationary random scattering conditions. Starting from first principles of wave propagation theory, we showed that this is naturally modeled as a 2D bandlimited signal--without any prior low-pass filtering operation--whose bandwidth is determined by the richness of the underlying scattering mechanism. This is maximum under isotropic propagation, where the field is circularly bandlimited of radius inversely proportional to the wavelength. Hence, conventional signal processing techniques for bandlimited signals, e.g.,\cite{OppenheimBook}, can be applied to electromagnetic fields due to a fundamental space-time duality and linear system-theoretic interpretation of wave propagation. 

For electrically large (i.e., relative to the wavelength) observation areas, the DoF per unit of area are the Nyquist samples per squared meters needed for reconstruction. With prior knowledge of the scattering conditions, we provided an efficient representation by solving an ellipse packing problem, which yields an elongated hexagonal sampling structure stretching inversely proportional to scattering selectivity. Under isotropic propagation, it leads to a $13\%$ less samples/m$^2$ compared to the classical half-wavelength sampling. This gap increases as the scattering becomes more selective angularly, thereby offering a substantial complexity reduction. 

A possible extension of this paper is to consider nonuniform sampling \cite{Unser}. In principle, this leads to another ellipse packing problem of equal sizes, but with arbitrary orientation and spacing, due to an irregular sampling structure.

\appendices

%\section{{\color{blue}Proof of Lemma~\ref{sec:interpolation_function_iso}}} \label{app:interp_iso}
\section{Proof of Lemma~\ref{sec:interpolation_function_iso}} \label{app:interp_iso}

From~\eqref{interp_function}, for a wavenumber support $\mathcal{K} = \mathcal{D}$ we obtain
\begin{align}
f_{\mathcal{D}}(\vect{r}) & = \frac{|\det(\vect{Q}^\star_{\mathcal{D}})|}{(2\pi)^2} \int_{\mathcal{D}} e^{\imagunit \vect{k}^{\Ttran} \vect{r}} \, d\vect{k} \\ \label{FT_iso}
& = \frac{1}{{2\sqrt{3}}} \int_{\|\vect{u}\|\le1} e^{\imagunit \kappa \vect{u}^{\Ttran} \vect{r}} \, d\vect{u}
\end{align}
where we applied a change of integration variables $\vect{u} = \kappa^{-1} \vect{I}_2 \vect{k} = (u_x,u_y)$ and substituted $|\det(\vect{Q}^\star_{\mathcal{D}})| = {\lambda^2}/({2\sqrt{3}})$ in~\eqref{density_sampling_hex} while using $\kappa=2\pi/\lambda$.
Being~\eqref{FT_iso} an inverse spatial Fourier transform of a rotationally symmetric spectrum, its result will be invariant under rotation, i.e., $f_{\mathcal{D}}(\vect{r}) = f_{\mathcal{D}}(r)$ with $\|\vect{r}\| = r$. Hence, we evaluate~\eqref{FT_iso} at the point $\vect{r}=(r,0)$:
\begin{align}
f_{\mathcal{D}}(\vect{r}) & = \frac{1}{\sqrt{3}}  \int_{-1}^1 e^{\imagunit \kappa r u_x} \sqrt{1-u_x^2} \, du_x \\
& \mathop{=}^{(a)}  \frac{1}{\sqrt{3}} \int_{0}^\pi e^{\imagunit \kappa r \cos(\theta)} \sin^2(\theta) \, d\theta
\end{align}
where in $(a)$ we applied a change of integration variables $u_x = \cos(\theta)$.
By using the Euler's formula and exploiting symmetry of the integrand we obtain
\begin{align} \label{FT_iso_2}
f_{\mathcal{D}}(\vect{r}) =  \frac{\pi}{\sqrt{3}} \frac{J_1(\kappa r)}{\kappa r} = \frac{\pi}{\sqrt{3}} \jinc(\kappa r)
\end{align}
where ${J_1(\alpha) \!=\! \frac{\alpha}{\pi} \!\!\int_{0}^\pi \!\! \cos(\alpha \cos(\theta)) \sin^2(\theta) d\theta}$ \cite[Eq.~(9.1.20)]{AbraSteg72}. %and $\jinc(x) = J_1(x)/x$ is the jinc function.

\section{Proof of Lemma~\ref{sec:interpolation_function_noniso}} \label{app:interp_noniso}

From~\eqref{interp_function}, for a wavenumber support $\mathcal{K} = \mathcal{E}$ we obtain
\begin{align}
f_{\mathcal{E}}(\vect{r}) & = \frac{|\det(\vect{Q}^\star_{\mathcal{E}})|}{(2\pi)^2} \int_{\mathcal{E}} e^{\imagunit \vect{k}^{\Ttran} \vect{r}} \, d\vect{k} \\ \label{FT_noniso}
& = |\det(\vect{G}^{1/2})|  \frac{|\det(\vect{Q}^\star_{\mathcal{E}})|}{(2\pi)^2} \int_{\mathcal{D}} e^{\imagunit (\vect{G}^{1/2} \vect{k}^\prime)^{\Ttran} \vect{r}} \, d\vect{k}^\prime
\end{align}
where we applied the change of integration variables ${\vect{k}^\prime = \varphi^{-1}(\vect{k}) =  \vect{G}^{-1/2} \vect{k}}$ with $\vect{k}^\prime \in \mathcal{D}$. Plugging~\eqref{sampling_matrix_ellipse},
\begin{align}  \label{FT_noniso_2}
f_{\mathcal{E}}(\vect{r})
& = \frac{|\det(\vect{Q}^\star_{\mathcal{D}})|}{(2\pi)^2} \int_{\mathcal{D}} e^{\imagunit (\vect{G}^{1/2} \vect{k}^\prime)^{\Ttran} \vect{r}} \, d\vect{k}^\prime.
\end{align}
Re-arranging the terms at the exponential of~\eqref{FT_noniso_2} to match~\eqref{FT_iso} yields ${f_{\mathcal{E}}(\vect{r}) =  f_{\mathcal{D}}((\vect{G}^{1/2})^{\Ttran} \vect{r})}$
%\begin{align}  \label{FT_noniso_3}
%f_{\mathcal{E}}(\vect{r})
%& = \kappa^2 \frac{|\det(\vect{Q}^\star_{\mathcal{D}})|}{(2\pi)^2} \int_{\|\vect{u}\|\le1} e^{\imagunit \kappa \vect{u}^{\Ttran} (\vect{G}^{1/2})^{\Ttran} \vect{r}} \, d\vect{u}\\ & 
%=  f_{\mathcal{D}}((\vect{G}^{1/2})^{\Ttran} \vect{r})
%\end{align}
Also, using~\eqref{transf_matrix} and exploiting the rotation invariance of $f_{\mathcal{D}}(\vect{r})$ yields $f_{\mathcal{E}}(\vect{r}) =  f_{\mathcal{D}}(\vect{D}^{1/2} \vect{r})$.
%The final expression in~\eqref{interpolation_function_noniso} follows directly from~\eqref{interpolation_function_iso}.

\bibliographystyle{IEEEtran}
\bibliography{IEEEabrv,refs}

% Generated by IEEEtran.bst, version: 1.14 (2015/08/26)
\begin{thebibliography}{10}
\providecommand{\url}[1]{#1}
\csname url@samestyle\endcsname
\providecommand{\newblock}{\relax}
\providecommand{\bibinfo}[2]{#2}
\providecommand{\BIBentrySTDinterwordspacing}{\spaceskip=0pt\relax}
\providecommand{\BIBentryALTinterwordstretchfactor}{4}
\providecommand{\BIBentryALTinterwordspacing}{\spaceskip=\fontdimen2\font plus
\BIBentryALTinterwordstretchfactor\fontdimen3\font minus
  \fontdimen4\font\relax}
\providecommand{\BIBforeignlanguage}[2]{{%
\expandafter\ifx\csname l@#1\endcsname\relax
\typeout{** WARNING: IEEEtran.bst: No hyphenation pattern has been}%
\typeout{** loaded for the language `#1'. Using the pattern for}%
\typeout{** the default language instead.}%
\else
\language=\csname l@#1\endcsname
\fi
#2}}
\providecommand{\BIBdecl}{\relax}
\BIBdecl

\bibitem{PizzoSPAWC20}
A.~{Pizzo}, T.~L. {Marzetta}, and L.~{Sanguinetti}, ``Degrees of freedom of
  {H}olographic {MIMO} channels,'' in \emph{IEEE Int. Workshop Signal Process.
  Adv. Wireless Commun. (SPAWC)}, 2020, pp. 1--5.

\bibitem{Shannon_Noise}
C.~Shannon, ``Communication in the presence of noise,'' \emph{Proc. IRE},
  vol.~37, no.~1, pp. 10--21, 1949.

\bibitem{whittaker_1928}
J.~M. Whittaker, ``The ``{Fourier}'' theory of the cardinal function,''
  \emph{Proc. Edinburgh Math. Soc.}, vol.~1, no.~3, p. 169–176, 1928.

\bibitem{Kotelnikov2001}
V.~A. Kotel'nikov, \emph{{On the Transmission Capacity of the ``Ether'' and
  Wire in Electrocommunications}}.\hskip 1em plus 0.5em minus 0.4em\relax
  Birkh{\"a}user Boston, 2001.

\bibitem{Unser}
M.~Unser, ``Sampling-50 years after shannon,'' \emph{Proc. IEEE}, vol.~88,
  no.~4, pp. 569--587, 2000.

\bibitem{SamplingBook}
R.~Marks, \emph{Introduction to {Shannon} Sampling and Interpolation
  Theory}.\hskip 1em plus 0.5em minus 0.4em\relax Springer-Verlag New York,
  1991.

\bibitem{MultivariateSPBook}
D.~E. Dudgeon and R.~M. Mersereau, \emph{Multidimensional Digital Signal
  Processing}.\hskip 1em plus 0.5em minus 0.4em\relax Prentice Hall, 1990.

\bibitem{PETERSEN1962279}
D.~P. Petersen and D.~Middleton, ``Sampling and reconstruction of
  wave-number-limited functions in {N}-dimensional euclidean spaces,''
  \emph{Information and Control}, vol.~5, no.~4, pp. 279--323, 1962.

\bibitem{Agrell2004}
H.~Kunsch, E.~Agrell, and F.~Hamprecht, ``Optimal lattices for sampling,''
  \emph{IEEE Trans. Inf. Theory}, vol.~51, no.~2, pp. 634--647, 2005.

\bibitem{FranceschettiBook}
M.~Franceschetti, \emph{Wave Theory of Information}.\hskip 1em plus 0.5em minus
  0.4em\relax Cambridge University Press, 2017.

\bibitem{Migliore_Horse}
M.~D. Migliore, ``Horse (electromagnetics) is more important than horseman
  (information) for wireless transmission,'' \emph{IEEE Trans. Antennas
  Propag.}, vol.~67, no.~4, pp. 2046--2055, 2019.

\bibitem{OppenheimBook}
A.~V. Oppenheim and R.~W. Schafer, \emph{{Discrete-Time Signal
  Processing}}.\hskip 1em plus 0.5em minus 0.4em\relax Pearson, 2010.

\bibitem{Rappaport2019}
T.~S. Rappaport, Y.~Xing, O.~Kanhere, S.~Ju, A.~Madanayake, S.~Mandal,
  A.~Alkhateeb, and G.~C. Trichopoulos, ``Wireless communications and
  applications above 100 {GHz}: Opportunities and challenges for {6G} and
  beyond,'' \emph{IEEE Access}, vol.~7, pp. 78\,729--78\,757, 2019.

\bibitem{Bucci87}
O.~Bucci and G.~Franceschetti, ``On the spatial bandwidth of scattered
  fields,'' \emph{IEEE Trans. Antennas Propag.}, vol.~35, no.~12, pp.
  1445--1455, 1987.

\bibitem{Bucci98}
O.~Bucci, C.~Gennarelli, and C.~Savarese, ``Representation of electromagnetic
  fields over arbitrary surfaces by a finite and nonredundant number of
  samples,'' \emph{IEEE Trans. Antennas Propag.}, vol.~46, no.~3, pp. 351--359,
  1998.

\bibitem{PoonDoF}
A.~S.~Y. Poon, R.~W. Brodersen, and D.~N.~C. Tse, ``Degrees of freedom in
  multiple-antenna channels: a signal space approach,'' \emph{IEEE Trans. Inf.
  Theory}, vol.~51, no.~2, pp. 523--536, Feb 2005.

\bibitem{Kennedy2007}
R.~A. Kennedy, P.~Sadeghi, T.~D. Abhayapala, and H.~M. Jones, ``Intrinsic
  limits of dimensionality and richness in random multipath fields,''
  \emph{IEEE Trans. Signal Process.}, vol.~55, no.~6, pp. 2542--2556, 2007.

\bibitem{Hanlen2007}
L.~W. Hanlen and T.~D. Abhayapala, ``Space-time-frequency degrees of freedom:
  Fundamental limits for spatial information,'' in \emph{IEEE Int. Symposium
  Inf. Theory (ISIT)}, 2007, pp. 701--705.

\bibitem{Sayeed2002}
A.~M. {Sayeed}, ``Deconstructing multiantenna fading channels,'' \emph{IEEE
  Trans. Signal Process.}, vol.~50, no.~10, 2002.

\bibitem{PizzoIT21}
A.~Pizzo, L.~Sanguinetti, and T.~L. Marzetta, ``Spatial characterization of
  electromagnetic random channels,'' \emph{IEEE Open J. Commun. Soc.}, vol.~3,
  pp. 847--866, 2022.

\bibitem{PizzoJSAC20}
A.~{Pizzo}, T.~L. {Marzetta}, and L.~{Sanguinetti}, ``Spatially-stationary
  model for {Holographic MIMO} small-scale fading,'' \emph{IEEE J. Sel. Areas
  Commun.}, vol.~38, no.~9, pp. 1964--1979, 2020.

\bibitem{MarzettaISIT}
T.~L. {Marzetta}, ``Spatially-stationary propagating random field model for
  {Massive MIMO} small-scale fading,'' in \emph{IEEE Int. Symposium Inf. Theory
  (ISIT)}, June 2018, pp. 391--395.

\bibitem{MarzettaNokia}
T.~L. Marzetta, ``{BLAST} arrays of polarimetric antennas,'' Nokia Proprietary,
  ~ ITD-01-41984K, 2001.

\bibitem{Hu2018}
S.~Hu, F.~Rusek, and O.~Edfors, ``Beyond massive mimo: The potential of data
  transmission with large intelligent surfaces,'' \emph{IEEE Trans. Signal
  Process.}, vol.~66, no.~10, pp. 2746--2758, 2018.

\bibitem{AgrellSPL}
E.~Agrell and B.~Csébfalvi, ``Multidimensional sampling of isotropically
  bandlimited signals,'' \emph{IEEE Signal Process. Lett.}, vol.~25, no.~3, pp.
  383--387, 2018.

\bibitem{Kennedy2004}
R.~Kennedy and T.~Abhayapala, ``Source-field wave-field concentration and
  dimension: towards spatial information content,'' in \emph{IEEE Int.
  Symposium Inf. Theory (ISIT)}, 2004, pp. 243--.

\bibitem{ChewBook}
W.~C. Chew, \emph{Waves and Fields in Inhomogenous Media}.\hskip 1em plus 0.5em
  minus 0.4em\relax Wiley-IEEE Press, 1995.

\bibitem{Wang88}
J.~Wang, ``An examination of the theory and practices of planar near-field
  measurement,'' \emph{IEEE Trans. Antennas Propag.}, vol.~36, no.~6, pp.
  746--753, 1988.

\bibitem{Miller}
D.~A.~B. Miller, ``Communicating with waves between volumes: evaluating
  orthogonal spatial channels and limits on coupling strengths,'' \emph{Appl.
  Opt.}, vol.~39, no.~11, pp. 1681--1699, Apr 2000.

\bibitem{Franceschetti}
M.~Franceschetti, ``On {L}andau's eigenvalue theorem and information
  cut-sets,'' \emph{IEEE Trans. Inf. Theory}, vol.~61, no.~9, Sept 2015.

\bibitem{BoydBook}
S.~Boyd and L.~Vandenberghe, \emph{{Convex Optimization}}.\hskip 1em plus 0.5em
  minus 0.4em\relax USA: Cambridge University Press, 2004.

\bibitem{Landau}
H.~Landau and H.~Widom, ``Eigenvalue distribution of time and frequency
  limiting,'' \emph{Journal of Mathematical Analysis and Applications},
  vol.~77, no.~2, pp. 469--481, 1980.

\bibitem{PizzoTWC21}
A.~Pizzo, L.~Sanguinetti, and T.~L. Marzetta, ``Fourier plane-wave series
  expansion for {Holographic MIMO} communications,'' \emph{IEEE Trans. Wireless
  Commun.}, pp. 1--1, 2022.

\bibitem{VanTreesBook}
H.~L. Van~Trees, \emph{Detection Estimation and Modulation Theory, Part
  I}.\hskip 1em plus 0.5em minus 0.4em\relax Wiley, 1968.

\bibitem{Wolfram}
\BIBentryALTinterwordspacing
E.~W. Weisstein, ``{Gauss}'s circle problem,'' from MathWorld--A Wolfram Web
  Resource. [Online]. Available:
  \url{https://mathworld.wolfram.com/GausssCircleProblem.html}
\BIBentrySTDinterwordspacing

\bibitem{PaulrajBook}
A.~Paulraj, R.~Nabar, and D.~Gore, \emph{{Introduction to Space-Time Wireless
  Communications}}.\hskip 1em plus 0.5em minus 0.4em\relax Cambridge, UK:
  Cambridge University Press, 2003.

\bibitem{Kumar2011}
A.~Kumar, P.~Ishwar, and K.~Ramchandran, ``High-resolution distributed sampling
  of bandlimited fields with low-precision sensors,'' \emph{IEEE Trans. Inf.
  Theory}, vol.~57, no.~1, pp. 476--492, 2011.

\bibitem{AbraSteg72}
M.~Abramowitz and I.~A. Stegun, Eds., \emph{{Handbook of Mathematical Functions
  with Formulas, Graphs, and Mathematical Tables}}.\hskip 1em plus 0.5em minus
  0.4em\relax U.S. Government Printing Office, 1972.

\end{thebibliography}


% Generated by IEEEtran.bst, version: 1.14 (2015/08/26)
\begin{thebibliography}{10}
\providecommand{\url}[1]{#1}
\csname url@samestyle\endcsname
\providecommand{\newblock}{\relax}
\providecommand{\bibinfo}[2]{#2}
\providecommand{\BIBentrySTDinterwordspacing}{\spaceskip=0pt\relax}
\providecommand{\BIBentryALTinterwordstretchfactor}{4}
\providecommand{\BIBentryALTinterwordspacing}{\spaceskip=\fontdimen2\font plus
\BIBentryALTinterwordstretchfactor\fontdimen3\font minus
  \fontdimen4\font\relax}
\providecommand{\BIBforeignlanguage}[2]{{%
\expandafter\ifx\csname l@#1\endcsname\relax
\typeout{** WARNING: IEEEtran.bst: No hyphenation pattern has been}%
\typeout{** loaded for the language `#1'. Using the pattern for}%
\typeout{** the default language instead.}%
\else
\language=\csname l@#1\endcsname
\fi
#2}}
\providecommand{\BIBdecl}{\relax}
\BIBdecl

\bibitem{PizzoSPAWC20}
A.~{Pizzo}, T.~L. {Marzetta}, and L.~{Sanguinetti}, ``{Degrees of Freedom of
  Holographic {MIMO} Channels},'' in \emph{2020 IEEE 21st Int. Workshop Signal
  Process. Adv. Wireless Commun. (SPAWC)}, 2020, pp. 1--5.

\bibitem{Shannon_Noise}
C.~Shannon, ``Communication in the presence of noise,'' \emph{Proceedings of
  the IRE}, vol.~37, no.~1, pp. 10--21, 1949.

\bibitem{whittaker_1928}
J.~M. Whittaker, ``{The “Fourier” Theory of the Cardinal Function},''
  \emph{Proceedings of the Edinburgh Mathematical Society}, vol.~1, no.~3, p.
  169–176, 1928.

\bibitem{Kotelnikov2001}
V.~A. Kotel'nikov, \emph{{On the Transmission Capacity of the ``Ether'' and
  Wire in Electrocommunications}}.\hskip 1em plus 0.5em minus 0.4em\relax
  Boston, MA: Birkh{\"a}user Boston, 2001, pp. 27--45.

\bibitem{Unser}
M.~Unser, ``Sampling-50 years after shannon,'' \emph{Proceedings of the IEEE},
  vol.~88, no.~4, pp. 569--587, 2000.

\bibitem{SamplingBook}
R.~Marks, \emph{{Introduction to Shannon Sampling and Interpolation
  Theory}}.\hskip 1em plus 0.5em minus 0.4em\relax Springer-Verlag New York,
  1991.

\bibitem{MultivariateSPBook}
D.~E. Dudgeon and R.~M. Mersereau, \emph{Multidimensional Digital Signal
  Processing}.\hskip 1em plus 0.5em minus 0.4em\relax Prentice Hall, 1990.

\bibitem{PETERSEN1962279}
D.~P. Petersen and D.~Middleton, ``Sampling and reconstruction of
  wave-number-limited functions in n-dimensional euclidean spaces,''
  \emph{Information and Control}, vol.~5, no.~4, pp. 279--323, 1962.

\bibitem{Migliore_Horse}
M.~D. Migliore, ``Horse (electromagnetics) is more important than horseman
  (information) for wireless transmission,'' \emph{IEEE Transactions on
  Antennas and Propagation}, vol.~67, no.~4, pp. 2046--2055, 2019.

\bibitem{PizzoTWC21}
A.~Pizzo, L.~Sanguinetti, and T.~L. Marzetta, ``{Holographic MIMO
  Communications},'' \emph{CoRR}, vol. abs/2105.01535, 2021. Online:
  \url{https://arxiv.org/abs/2105.01535}.

\bibitem{MarzettaISIT}
T.~L. {Marzetta}, ``Spatially-stationary propagating random field model for
  {Massive MIMO} small-scale fading,'' in \emph{2018 IEEE Int. Symposium Inf.
  Theory (ISIT)}, June 2018, pp. 391--395.

\bibitem{Bucci87}
O.~Bucci and G.~Franceschetti, ``On the spatial bandwidth of scattered
  fields,'' \emph{IEEE Transactions on Antennas and Propagation}, vol.~35,
  no.~12, pp. 1445--1455, 1987.

\bibitem{Bucci98}
O.~Bucci, C.~Gennarelli, and C.~Savarese, ``Representation of electromagnetic
  fields over arbitrary surfaces by a finite and nonredundant number of
  samples,'' \emph{IEEE Transactions on Antennas and Propagation}, vol.~46,
  no.~3, pp. 351--359, 1998.

\bibitem{PoonDoF}
A.~S.~Y. Poon, R.~W. Brodersen, and D.~N.~C. Tse, ``Degrees of freedom in
  multiple-antenna channels: a signal space approach,'' \emph{IEEE Trans. Inf.
  Theory}, vol.~51, no.~2, pp. 523--536, Feb 2005.

\bibitem{Kennedy2007}
R.~A. Kennedy, P.~Sadeghi, T.~D. Abhayapala, and H.~M. Jones, ``Intrinsic
  limits of dimensionality and richness in random multipath fields,''
  \emph{IEEE Transactions on Signal Processing}, vol.~55, no.~6, pp.
  2542--2556, 2007.

\bibitem{Hanlen2007}
L.~W. Hanlen and T.~D. Abhayapala, ``Space-time-frequency degrees of freedom:
  Fundamental limits for spatial information,'' in \emph{2007 IEEE
  International Symposium on Information Theory}, 2007, pp. 701--705.

\bibitem{Sayeed2002}
A.~M. {Sayeed}, ``Deconstructing multiantenna fading channels,'' \emph{IEEE
  Trans. Signal Process.}, vol.~50, no.~10, 2002.

\bibitem{PizzoIT21}
A.~Pizzo, L.~Sanguinetti, and T.~L. Marzetta, ``{Spatial Characterization of
  Electromagnetic Random Channels},'' \emph{CoRR}, vol. abs/2103.15666, 2021.
  Online: \url{https://arxiv.org/abs/2103.15666}.

\bibitem{Hu2018}
S.~Hu, F.~Rusek, and O.~Edfors, ``Beyond massive mimo: The potential of data
  transmission with large intelligent surfaces,'' \emph{IEEE Transactions on
  Signal Processing}, vol.~66, no.~10, pp. 2746--2758, 2018.

\bibitem{Agrell2004}
H.~Kunsch, E.~Agrell, and F.~Hamprecht, ``{Optimal lattices for sampling},''
  \emph{IEEE Transactions on Information Theory}, vol.~51, no.~2, pp. 634--647,
  2005.

\bibitem{Landau1986}
H.~Landau, ``Extrapolating a band-limited function from its samples taken in a
  finite interval,'' \emph{IEEE Transactions on Information Theory}, vol.~32,
  no.~4, pp. 464--470, 1986.

\bibitem{ChewBook}
W.~C. Chew, \emph{Waves and Fields in Inhomogenous Media}.\hskip 1em plus 0.5em
  minus 0.4em\relax Wiley-IEEE Press, 1995.

\bibitem{Franceschetti}
M.~Franceschetti, ``On {L}andau's eigenvalue theorem and information
  cut-sets,'' \emph{IEEE Trans. Inf. Theory}, vol.~61, no.~9, Sept 2015.

\bibitem{AgrellSPL}
E.~Agrell and B.~Csébfalvi, ``Multidimensional sampling of isotropically
  bandlimited signals,'' \emph{IEEE Signal Processing Letters}, vol.~25, no.~3,
  pp. 383--387, 2018.

\bibitem{BoydBook}
S.~Boyd and L.~Vandenberghe, \emph{{Convex Optimization}}.\hskip 1em plus 0.5em
  minus 0.4em\relax USA: Cambridge University Press, 2004.

\bibitem{FranceschettiBook}
M.~Franceschetti, \emph{Wave Theory of Information}.\hskip 1em plus 0.5em minus
  0.4em\relax Cambridge University Press, 2017.

\bibitem{Landau}
H.~Landau and H.~Widom, ``Eigenvalue distribution of time and frequency
  limiting,'' \emph{Journal of Mathematical Analysis and Applications},
  vol.~77, no.~2, pp. 469--481, 1980.

\bibitem{MarzettaIT}
T.~L. Marzetta, E.~G. Larsson, and T.~B. Hansen, ``{Massive MIMO and Beyond},''
  in \emph{Information Theoretic Perspectives on 5G Systems and Beyond}, S.~S.
  I.~Maric, O.~Simeone, Ed.\hskip 1em plus 0.5em minus 0.4em\relax Cambridge:
  Cambridge University Press, 2020.

\bibitem{PizzoJSAC20}
A.~{Pizzo}, T.~L. {Marzetta}, and L.~{Sanguinetti}, ``{Spatially-Stationary
  Model for Holographic {MIMO} Small-Scale Fading},'' \emph{IEEE J. Sel. Areas
  Commun.}, vol.~38, no.~9, pp. 1964--1979, 2020.

\bibitem{PaulrajBook}
A.~Paulraj, R.~Nabar, and D.~Gore, \emph{{Introduction to Space-Time Wireless
  Communications}}.\hskip 1em plus 0.5em minus 0.4em\relax Cambridge, UK:
  Cambridge University Press, 2003.

\bibitem{VanTreesBook}
H.~L. Van~Trees, \emph{Detection Estimation and Modulation Theory, Part
  I}.\hskip 1em plus 0.5em minus 0.4em\relax Wiley, 1968.

\bibitem{AbraSteg72}
M.~Abramowitz and I.~A. Stegun, Eds., \emph{{Handbook of Mathematical Functions
  with Formulas, Graphs, and Mathematical Tables}}.\hskip 1em plus 0.5em minus
  0.4em\relax Washington, DC, USA: U.S. Government Printing Office, 1972.

\end{thebibliography}

% that's all folks
\end{document}